\documentclass[10pt,journal]{IEEEtran}
\ifCLASSINFOpdf
   \usepackage[pdftex]{graphicx}
  
\else
  \usepackage[dvips]{graphicx}
\fi
\usepackage[cmex10]{amsmath}
\usepackage{amssymb}
\newcommand{\id}{1\hspace{-0,9ex}1}
\usepackage{algorithm}
\usepackage{algorithmic}

\usepackage{caption}
\usepackage{url}

\usepackage[usenames,dvipsnames,svgnames,table]{xcolor}

\newcommand{\revised}[1]{\textcolor{black}{#1}}

\usepackage{paralist}
\usepackage{subfig}
\usepackage{multirow}
\usepackage{enumitem}
\usepackage{bm}

\usepackage{color}
\usepackage{algorithm}
\usepackage{algorithmic}
\usepackage{url}
\usepackage{lastpage}
\usepackage{verbatim}
\usepackage{graphicx}

\usepackage{amsthm}
\usepackage{amsmath}
\usepackage{etoolbox}

\usepackage{tabularx}

\theoremstyle{definition}
\newtheorem{definition}{Definition}[section]
\newtheorem{proposition}{Proposition}[section]
\newtheorem{theorem}{Theorem}[section]
\newtheorem{remark}{Remark}[section]

\renewenvironment{proof}{}{{$\hfill\blacksquare$}}

\usepackage{kbordermatrix}
\renewcommand{\kbldelim}{(}
\renewcommand{\kbrdelim}{)}

\definecolor{darkgreen}{rgb}{0.0, 0.2, 0.13}
\definecolor{darkred}{rgb}{0.55, 0.0, 0.0}
\definecolor{darkblue}{rgb}{0.0, 0.0, 0.55}
\definecolor{earthyellow}{rgb}{0.88, 0.66, 0.37}
\definecolor{fulvous}{rgb}{0.86, 0.52, 0.0}
\definecolor{darkcyan}{rgb}{0.0, 0.55, 0.55}
\definecolor{harvardcrimson}{rgb}{0.79, 0.0, 0.09}

\begin{document}
%
\title{SAF: Stochastic Adaptive Forwarding \\in Named Data Networking}
%
%
%

\author{Daniel~Posch,
         Benjamin~Rainer,
         and~Hermann~Hellwagner

}

%
%

\markboth{}
{Posch, Rainer, and Hellwagner: Stochastic Adaptive Forwarding in Named Data Networking}
%

\author{\IEEEauthorblockN{Daniel Posch, Benjamin Rainer, and Hermann Hellwagner}\\
\IEEEauthorblockA{Institute of Information Technology\\
Alpen-Adria-Universit\"at Klagenfurt, Austria\\
Email: \{first.lastname\}@itec.aau.at}
}



\IEEEoverridecommandlockouts
\IEEEpubid{\makebox[\columnwidth]{Submitted to IEEE/ACM Transactions on Networking 
\hfill} \hspace{\columnsep}\makebox[\columnwidth]{ }} 

\maketitle

\begin{abstract}
Forwarding decisions in classical IP-based networks are predetermined by routing. This is necessary to avoid loops, inhibiting opportunities to implement an adaptive and intelligent forwarding plane. Consequently, content distribution efficiency is reduced due to a lack of inherent multi-path transmission. In Named Data Networking (NDN) instead, routing shall hold a supporting role to forwarding, providing sufficient potential to enhance content dissemination at the forwarding plane. In this paper we design, implement, and evaluate a novel probability-based forwarding strategy, called Stochastic Adaptive Forwarding (SAF) for NDN. SAF imitates a self-adjusting water pipe system, intelligently guiding and distributing Interests through network crossings circumventing link failures and bottlenecks. Just as real pipe systems, SAF employs overpressure valves enabling congested nodes to lower pressure autonomously. Through an implicit feedback mechanism it is ensured that the fraction of the traffic forwarded via congested nodes decreases. By conducting simulations we show that our approach outperforms existing forwarding strategies in terms of the Interest satisfaction ratio in the majority of the evaluated scenarios. This is achieved by extensive utilization of NDN's multipath and content-lookup capabilities without relying on the routing plane. SAF explores the local environment by redirecting requests that are likely to be dropped anyway. This enables SAF to identify new paths to the content origin or to cached replicas, circumventing link failures and resource shortages without relying on routing updates.
\end{abstract}



%
\IEEEpeerreviewmaketitle

\section{Introduction}
\label{sec:introduction}

Today's Internet is based on a legacy host-based architecture, resulting in limitations. One key issue is forwarding, which is strictly predetermined by routing to ensure loop-free communication. There are few opportunities to implement an adaptive forwarding plane in classical IP-based networks since routing dictates the forwarding options. In IP, forwarding planes are stateless and routing protocols are responsible to deal with all kinds of short- and long-term topology changes. Routing protocols struggle with all the imposed responsibility.
\revised{For instance, Labovitz et al.~\cite{Labovitz:2000,Labovitz:2001} showed that BGP convergence times up to 15 minutes recovering from a single multi-homed fault (switching to an alternative, redundant route) are possible, if no additional mechanisms (e.g. backup configurations providing alternative routes~\cite{Kini:2009}) are used.}
The concept of Information-Centric Networking (ICN) \cite{icnSurvey,Xylomenos:2014} instead rests on a content-centric communication model, where content only is addressed. There are a variety of approaches for ICN architectures \cite{ netinf, psirp, Jacobson:2009, dona, ndn:2014}. In this paper, our understanding of an ICN is coincident with Named Data Networking (NDN) \cite{ndn:2014}. In NDN, data is requested by its name following a strictly receiver-driven communication model. In order to retrieve content, a consumer emits an \emph{Interest} message, which is then forwarded by other NDN nodes until it reaches the desired content. An Interest's final destination could be the content origin, or any intermediate node that holds a cached replica. A \emph{Data} packet carries the requested content object and is always returned on the reverse path of the requesting Interest.
A typical NDN router maintains three data structures: 
\begin{inparaenum}[\itshape i\upshape)]
\item \emph{Content Store} (CS), a cache providing data replicas;
\item \emph{Pending Interest Table} (PIT), a table keeping track of the forwarded (still pending) Interests providing the return path for Data packets; and
\item \emph{Forwarding Information Base} (FIB), a table essentially maintaining routing information.
\end{inparaenum}
\revised{In order to detect duplicate or looping Interests, a \textsc{Nonce} (unique bit pattern) is added to each Interest. Loop prevention is achieved by matching name and \textsc{Nonce} of received Interests to those already in the PIT. This enables an adaptive forwarding plane with inherent multi-path delivery, in contrast to IP that avoids potential loops in the first place.}

It has been shown by Yi et al.~\cite{Yi:2013} that NDN's stateful forwarding plane handles typical network issues, such as short-term link failures and congestion, more effectively than IP networks. Furthermore, in~\cite{yi:routing} Yi et al. argue that \emph{routing} in NDN \emph{shall hold a supporting role to forwarding}. Routing should only provide a reasonable starting point for the forwarding plane, which then should explore different multi-path opportunities. In return, adaptive forwarding enables a more scalable routing plane with relaxed requirements in terms of convergence time and completeness. In this paper our understanding of routing, forwarding and caching and their clear separation is in accordance with Yi et al.~\cite{Yi:2013}. This assumption does not restrict the coupling of these mechanisms as suggested by~\cite{Rossini:2014} and~\cite{Yeh:2014}, but clearly separates their areas of responsibility.
\revised{ The existing adaptive forwarding strategies in NDN do not provide all mechanisms foreseen in~\cite{Yi:2013} and~\cite{yi:routing}.}
For this reason we propose \emph{Stochastic Adaptive Forwarding} (SAF), a probability-based forwarding strategy. \revised{As application field we envisage the infrastructure of large Internet Service Providers (ISPs) interconnecting several autonomous systems and/or access networks.} SAF is specifically designed to meet the following objectives:

\begin{enumerate}
\item Perform stochastic adaptive forwarding on a per-content/per-prefix basis.
\item Provide effective forwarding even with incomplete or partly invalid routing information.
\item Deal with unexpected network topology changes, e.g., link failures, without relying on the routing plane.
\item Discover unknown paths to cached replicas.
\end{enumerate}
\revised{While these attributes are beneficial for the rather versatile edge networks, they are not suitable for the rigid Internet core. Here, simpler approaches may perform better considering the well known scalability principle for networks: ``\emph{complexity at the edge, aggregation at the core}''~\cite{rfc3272}.}

SAF accomplishes the previously mentioned objectives by imitating a water pipe system. Network nodes act as crossings for an incoming flow of water (Interests). Returning Data packets act as input for a probability distribution which determines the share of the flow that is forwarded via the different pipes (outgoing interfaces). Each crossing maintains an overpressure valve. If the pressure on a node increases, e.g., due to network congestion, it may use the overpressure valve to lower the pressure autonomously. The Interests that pass the overpressure valve can be used for different purposes. For instance, they can be used as scouts to investigate unknown paths to complement routing information. However, in some cases it is best to simply discard these Interests. The discarded and therefore unsatisfied Interests provide implicit feedback to the requesting nodes, indicating that they should reduce the number of requests forwarded to the considered node.

SAF is based on an exchangeable measure that defines the target of the adaptive forwarding.
\revised{The objective is to maximize the Interest/Data satisfaction ratio with (optional) respect to delay, hop-count, and/or transmission cost considerations}.
Based on the requested content, the specific service and/or the network operator's ambitions, this measure has to be chosen carefully. It determines the preferred paths for forwarding. However, adaptivity provided by SAF is not only achieved by intelligent multi-path transmission using redundant paths. SAF is also able to exploit content-based information to further improve the forwarding decisions. For instance, this is beneficial for multimedia scenarios where the relative priority of packets (e.g., VoIP vs. file transfer) is more important than a purely throughput- and/or delay-based metric may indicate.

The remainder of this paper is organized as follows. Section~\ref{sec:related_work} discusses existing forwarding strategies for the ICN/NDN approach. Section~\ref{sec:forwarding_strategy} provides the stochastic forwarding model SAF rests upon. Section~\ref{sec:evaluation} presents an evaluation of SAF, comparing it to other state-of-the-art forwarding strategies. Section~\ref{sec:conclusion} concludes our findings and discusses potential future work.

\vspace{-0.1cm}
\section{Related Work}
\label{sec:related_work}

As the work in this paper focuses on forwarding in NDN, we first discuss the proposed forwarding strategies introduced by the NDN community. Yi et al.~\cite{Yi:2013,adaptiveForwarding} classify \mbox{(inter-)faces} based on a simple color scheme. Faces can be marked as \textsc{Green}, \textsc{Yellow} and \textsc{Red}, which corresponds to the meaning that faces return data (\textsc{Green}), may or may not return data (\textsc{Yellow}), or they do not work at all (\textsc{Red}). Within this classification, faces are ranked, e.g., based on the delay of receiving Data packets. This basic scheme is used by all proposed forwarding strategies that are provided within version 1.0 of the ns3/ndnSIM simulator~\cite{ndnSIM}:

\begin{itemize}
\item Flooding: Interests are forwarded to all \textsc{Green} and \textsc{Yellow} faces supplied by the FIB.

\item SmartFlooding: Interests are forwarded to the highest-ranked \textsc{Green} face. If no \textsc{Green} face is available, an Interest is forwarded via all \textsc{Yellow} faces.

\item BestRoute:  Interests are forwarded to the highest-ranked \textsc{Green} face. If no \textsc{Green} face is available, an Interest is forwarded via the highest-ranked \textsc{Yellow} face.
\end{itemize}

Furthermore, forwarding strategies in ns3/ndnSIM~v1.0 can be supplemented by several enhancements~\cite{ndnSIM}. One example are \textit{Interest Limits}, which conceptually are Token Bucket filters. For instance, if the highest ranked interface reaches its transmission limit, 
an \textit{Interest Limit} ensures that another face is selected for forwarding further requests. Further supplements are negative acknowledgment messages (NACKs), which can be returned to an Interest issuer to provide immediate feedback if a request can not be satisfied (e.g., due to the imposed Interest Limits).

Recently ns3/ndnSIM~v2.0~\cite{ndnSIM2.0} was released. The newer version no more re-implements basic NDN primitives, such as forwarding, but uses code from the \textit{NDN Forwarding Daemon} (NFD) \cite{nfd}. This allows realistic simulations since the code-base of the NFD is used, which has been developed for physical hardware. The step towards realistic simulations resulted in major changes in the simulator also affecting the implemented forwarding strategies. In ns3/ndnSIM~v2.0 strategies no longer rest on the aforementioned color scheme and cannot take advantage of additional features such as \textit{Interest Limits} or NACKs. We shortly outline the forwarding strategies available in the NFD~\cite{nfd}:

\begin{itemize}

\item Broadcast: Interests are forwarded to all faces supplied by the FIB.

\item ShortestRoute: Interests are forwarded to the lowest-cost (e.g., hop count) upstream face indicated by routing. Actually this strategy is referred to as BestRoute in the NFD~\cite{nfd}; however, to avoid confusion with BestRoute in ns3/ndnSIM~v1.0~\cite{ndnSIM}, we renamed it for this paper.

\item NCC: Interests are forwarded to those faces that provide data packets with the lowest delay.
\end{itemize}

As the aforementioned strategies are available in the ns3/ndnSIM simulator, they can easily be used for comparison with new approaches. Since we are focusing on a realistic approach with SAF, we regard ns3/ndnSIM~v.2.0 as the platform most suitable for evaluations and performance measurements. For this reason, we consider the algorithms Broadcast, ShortestRoute, and NCC for comparison to SAF, \revised{and do not take into account their predecessors, Flooding, SmartFlooding, and BestRoute in ndnSIM~v.1.0.}

\revised{In~\cite{Rossini:2014} Rossini and Rossi propose [ideal] Nearest Replica Routing ([i]NRR), an approach to couple caching and forwarding extending aNET~\cite{Rosensweig:2010}.
The iterative algorithm provided in~\cite{Rossini:2014} makes use of an oracle providing information on the availability of content in all caches in the network. The strategy selects the face with the shortest distance to the content, which prefers nearby caches rather than forwarding the Interest towards the content origin. iNRR is implemented in the ccnSIM simulator~\cite{Chiocchetti:2013}. Although the authors indicate that an implementation of a perfect oracle is not feasible in a real environment,
we consider iNRR with perfect knowledge about the individual content chunks in the caches as a competitor to SAF. Please note that in~\cite{Rossini:2014} some practical approaches using off-path exploration to assess the necessary information provided by the oracle for the concept of NRR are proposed.}

Chiocchetti et al.~\cite{INFORM} developed INFORM, which is an adaptive hop-by-hop forwarding strategy using reinforcement learning inspired by the Q-routing framework. INFORM is able to discover temporary copies of content not present in the routing table, thus increasing the effectiveness of forwarding. \revised{The authors have implemented and evaluated INFORM within ccnSIM~\cite{Chiocchetti:2013}. Unfortunately, they indicate on
their website (\url{http://goo.gl/2XJa9R}) that due to a lack of manpower no release candidate or source code package for this strategy can be provided. However, the authors also indicate that one should prefer iNRR over INFORM for comparison, as it outperforms INFORM. Since iNRR is already on our list as a competitor to SAF, we do not consider INFORM.}

In~\cite{Carofiglio:2013} Carofiglio et al. derive a set of optimal dynamic multipath congestion control protocols and request forwarding strategies from a multi-commodity flow problem. The proposed Request Forwarding Algorithm (RFA) in~\cite{Carofiglio:2013} is outlined in detail and is a good candidate for comparison. The idea of the algorithm is simple yet effective. For each content-prefix and for each face, RFA monitors the PIT entries. The forwarding probability of a face is then determined by a weight, which is actually a moving average over the reciprocal count of the PIT entries. We re-implemented this algorithm for ndnSIM~v2.0 and NFD for comparison (more details are provided in the evaluation section).

Qian et al.~\cite{probForwarding} proposed the concept of Probability-based Adaptive Forwarding. The basic idea is to select faces based on a probability distribution, which is also similar to our approach. However, noticeable differences are that~\cite{probForwarding} is inspired by ant colony optimization and focuses on delay minimization. As described later on in the paper, SAF is generic providing opportunities to adapt forwarding in additional dimensions. We introduce a virtual face that enables content- and context-aware adaptation. Furthermore, we do not introduce distinguished Interest packets for probing only. 

Yeh et al.~\cite{Yeh:2014} proposed VIP, a framework for \emph{joint dynamic forwarding and caching} in NDN. In this system, \emph{Virtual Interest Packets} (VIPs) capture the measured demand for respective data objects. The VIP count in a part of a network represents the local level of interest in a given object. The VIP framework employs a virtual control plane which operates on the VIPs. Distributed control algorithms are used to guide caching and forwarding strategies. We do not consider the VIP framework as competitor due to: 
\begin{inparaenum}[\itshape i\upshape)]
\item absence of a reference implementation;
\item difficulties to precisely reproduce the results. The illustration in~\cite{Yeh:2014} leaves room for interpretation regarding the individual parts (virtual control plane, communication protocols, control algorithms) and their interplay.
\item no estimation regarding the communication overhead is provided. 
\end{inparaenum}

\revised{
Recently, Udugama et al.~\cite{multipathInterestFW} published On-demand Multi-Path Interest Forwarding (OMP-IF). This forwarding strategy uses multiple node disjoint paths for Interest forwarding simultaneously. Each router may only use a single face (from the FIB) for forwarding per content-prefix to ensure node disjointness. The client is responsible for triggering the multipath transmission by utilizing a weighted round-robin mechanism based on path delays to distribute Interests over multiple faces. However, considering only node disjoint paths may leave some network resources unused. We consider OMP-IF as competitor and implemented the strategy as presented in~\cite{multipathInterestFW}.
}

\section{Stochastic Adaptive Forwarding}
\label{sec:forwarding_strategy}

This section deals with the terminology and design of SAF. First, the network, content and node models are discussed, which state the necessary preconditions for the presented approach. Subsequently, we show that SAF enables adaptive forwarding based on a given measure. The design of $\mathcal{M}_T$, an exemplary measure maximizing the Interest satisfaction ratio, is illustrated. SAF defines update operations that modify the forwarding probabilities for (inter-)faces. The ultimate goal of SAF is to optimize a node's forwarding behavior such that it performs optimally in terms of a given measure. The design of these update operations is presented and exemplarily discussed based on $\mathcal{M}_T$. Finally, two examples are presented illustrating the functionality of SAF.

\subsection{Network, Content and Node Models}
\label{subsec:network_content_node_model}
SAF rests on the following network model:
$\mathcal{N}(\mathcal{V},\mathcal{E})$ denotes a network consisting of a set of nodes $\mathcal{V}$ and a set of edges/links $\mathcal{E} \subseteq \mathcal{V} \times \mathcal{V} $. Each node  $v \in \mathcal{V}$ maintains a physical face $F_{(v, u)}$ (e.g., wireless network interface) for each tuple $(v, u) \in \mathcal{E}$. We define the list of physical faces on $v$ as $\mathcal{F}'_{v} := \bigcup_{(v,u) \in \mathcal{E}} F_{(v, u)}$, where $|\mathcal{F}'_{v}|$ denotes the number of physical links/faces of $v$. A node may receive Interests on any $F_{in} \in \mathcal{F}'_{v}$ and tries to satisfy these requests by either returning a locally stored copy of the requested data or by forwarding the Interest to a suitable face $F_{out} \in \mathcal{F}'_{v} \setminus \{F_{in}\}$.
The content catalogue in $\mathcal{N}$ is determined by a set $\mathcal{C}$. Each $c \in \mathcal{C}$ denotes content that can be retrieved using a common prefix. 

In addition to the physical faces $\mathcal{F}'_{v}$, each node maintains a distinguished virtual face $F_{D_v}$, which acts as overpressure valve. The set $\mathcal{F}_v = \mathcal{F}'_{v} \cup \{F_{D_v}\}$ denotes the entire set of faces known to $v$. SAF is an algorithm local to each node, requiring no explicit communication between nodes. Therefore further discussion on SAF focuses only on a single node, which allows us to omit the subscripts $v$ identifying a specific node for the remainder of this paper. The virtual face $F_D$ is treated as an ordinary face by SAF, however, any Interest forwarded to this face is discarded; the use of this \textit{dropping face} will be made clear in the course of the paper.

Every node maintains a so called \textit{Forwarding Table} (FWT). 
This table is a two-dimensional matrix, where the rows correspond to the set of faces $\mathcal{F}$ and the columns correspond to the different contents from the catalogue $\mathcal{C}$. The elements of the matrix indicate the confidence (probability) with which a certain outgoing face can provide data for a certain prefix. For instance, the following matrix represents an example FWT for a node with $\mathcal{F} = \{F_D,F_0,F_1,F_2\}$ and $\mathcal{C} = \{c_0,c_1,c_2\}$:

\vspace{-0.15cm}
\small
\[
  \text{FWT} = \kbordermatrix{
    & c_0 & c_1 & c_2 \\
    F_D & 0 & 0 & 1/3  \\
    F_0 & 1/3 & 1/2 & 0  \\
    F_1 & 2/3 & 0 & 2/3  \\
    F_2 & 0 & 1/2 &  0
  }
\]
\normalsize

The FWT provides the probability of forwarding an Interest for content $c_l$ on face $F_i$. We denote $p(F_i, c_l)$ as the forwarding probability that an Interest asking for $c_l$ will be forwarded on $F_i$, for instance, $p(F_1, c_0) = \frac{2}{3}$. Note that, for any $c_l$, the corresponding column of the FWT specifies a discrete probability distribution: ${\textstyle\sum_{F_i \in \mathcal{F}} p(F_i, c_l) = 1}.$ 
%
%
The decision to forward a given Interest on a face is as simple as drawing a random number from a uniform distribution $U(0,1-p(F_{in}, c_l))$. Algorithm~\ref{alg:chooseFace} sketches the face selection process, which is also known as \emph{inverse transform sampling}. The function \emph{nextDouble()} (cf. Alg.~\ref{alg:chooseFace} line 2) draws a number from the distribution $U$ and \emph{pop()} (cf. Alg.~\ref{alg:chooseFace} line 4) removes and returns the top element of a list ($F_{list}$).
The algorithm requires $O(|\mathcal{F}|)$ steps. As the number of faces is usually constant on a node, the algorithm is actually in $O(1)$.
\begin{algorithm}[tbh!]
\small
\caption{Select outgoing Face for an Interest}
\label{alg:chooseFace}
\begin{algorithmic}[1]
\STATE $F_{list} \leftarrow \mathcal{F} \setminus\{F_{in},F_D\}$, $limit \leftarrow 0.0$
\STATE $rand \leftarrow U(0,1-p(F_{in}, c_l)).nextDouble()$

\WHILE {$F_{list} \neq \emptyset$}
\STATE $F_{cur} \leftarrow F_{list}.pop()$
\STATE $limit \leftarrow limit + p(F_{cur},c_l)$
\IF {$rand \leq limit$}
\RETURN{$F_{cur}$}
\ENDIF
\ENDWHILE

\RETURN{$F_D$}
\end{algorithmic}
\end{algorithm}

Figure~\ref{fig:saf_node} illustrates an NDN node using SAF. The design of the FWT provides two opportunities to perform adaptive forwarding, which are implemented by the \emph{Adaptation Engine}. First, modifications of the probabilities within a single column (\emph{solid, red lines} in Figure~\ref{fig:saf_node}) of an FWT change the forwarding probabilities for Interests asking for a specific $c_l$. These updates modify which faces/paths are preferred for forwarding Interests. Second, shifting forwarding probabilities among different columns (\emph{dashed, blue lines} in Figure~\ref{fig:saf_node}) allows prioritization of specific content types/prefixes. For instance, assume $c_{l}$ is more important than $c_{k}$. In this case it can be beneficial to increase the probability of dropping Interests for $c_{k}$ in favor of $c_{l}$. The necessary statistical information for these operations is provided by the \emph{Statistic Collector}, which monitors the Interests received and satisfied by faces.
\begin{figure}[tbh!]
\centering
\includegraphics[scale=0.80]{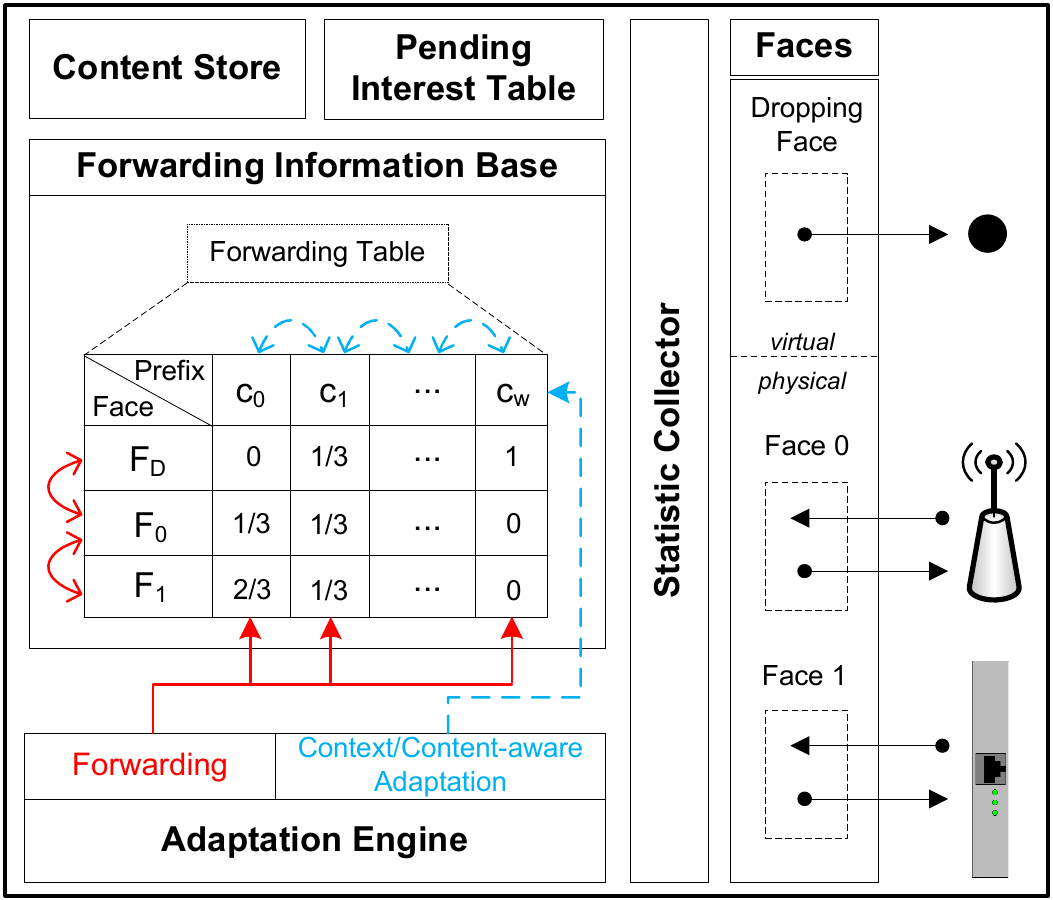}
\caption{The model of an NDN node using SAF.}
\label{fig:saf_node}
\end{figure}

This paper focuses on the core forwarding part of SAF (cf. \emph{solid, red lines} in Figure~\ref{fig:saf_node}). Additional context- and content-aware adaptations (cf. \emph{dashed, blue lines} in Figure~\ref{fig:saf_node}) are inherently supported by the design of SAF, however, they are considered as out of scope for this paper. Note that the forwarding core of SAF operates separately within a single column (content prefix) of the FWT. Therefore, and for the sake of readability we omit the notion of different $c \in \mathcal{C}$ for the remainder of this paper.




\subsection{A Throughput-based Forwarding Measure}
\label{sec:thput_measure}
SAF is based on the measure $\mathcal{M}$ as defined as follows:
\begin{definition}
\label{def:measure}
Let $(\Omega,\mathcal{A})$ be a measurable space with a set $\Omega$, the corresponding $\sigma$-Algebra $\mathcal{A} \subseteq \mathcal{P}(\Omega)$ and \\
$\mathcal{M},S,U: \Omega \rightarrow \mathbb{R}$ $\mathcal{A}$ - $\mathcal{B}$ - measurable functions with: \\
$\forall F_i \in \mathcal{F}\setminus\{F_D\}, \forall A \in \mathcal{A}: \mathcal{M}_{F_i}(\emptyset) = 0,$\\
$\mathcal{M}_{F_i}(A) = S_{F_i}(A) - U_{F_i}(A)$ and $\forall A \in \mathcal{A}: \mathcal{M}_{F_D}(A) = 0$.
\end{definition}
%
\noindent$S_{F_i}(A)$ provides a measure for \emph{satisfying Interests} and $U_{F_i}(A)$ defines a measure for \emph{not satisfying Interests}. Note that, since $U_{F_i}(A)$ is the complementary measure to $S_{F_i}(A)$, it suffices to define $S_{F_i}(A)$. $S_{F_i}(A)$ is not predefined by SAF, which allows to use an arbitrary measure. For instance, one may define $S_{F_i}(A)$ as the number of Interests that are satisfied by Data packets below a certain delay threshold. The objective of $\mathcal{M}$, therefore the definition of $S_{F_i}(A)$, guides the update operations for the FWT. These operations are issued periodically, e.g., once every second.
SAF finds the optimal forwarding strategy by maximizing $\mathcal{M}$ over the periods. The forwarding core of SAF operates separately on each column (prefix) of the FWT, thus providing for decoupling of periods for the different contents (prefixes). This allows to spread the node's work load over time.

In this paper we use the measure $\mathcal{M}^{T}$, which maximizes the throughput by investigating the Interest satisfaction ratio on individual faces. Therefore we consider for each node the set of forwarded Interests $\mathcal{I}_n$ during period $n$, for which a Data packet has been received (\emph{satisfied} Interest) or a timeout has occurred (\emph{unsatisfied} Interest) within that period $n$. Interests neither satisfied nor unsatisfied within a period are termed \emph{pending}.
We define $S_{F_i}(\mathcal{I}_n) := |\{j \in \mathcal{I}_n : j$ is satisfied by a Data packet on $F_i\}|$ and $U_{F_i}(\mathcal{I}_n) := |\{j \in \mathcal{I}_n : j$ is not satisfied on $F_i\}|$ $\forall F_i \in \mathcal{F}\setminus\{F_D\}$, and $S_{F_D}(\mathcal{I}_n) = |\{j \in \mathcal{I}_n : j$ is satisfied by $F_D\}|$ ($F_D$ satisfies Interests by definition, but note that $\mathcal{M}_{F_D}(A) = 0$); $U_{F_D}(\mathcal{I}_n) = 0$). This implicitly defines $\mathcal{M}^T$. 

{\small
\revised{
\begin{remark}
The ongoing discussion of SAF is agnostic to a concrete instantiation of a measure $\mathcal{M}$. One may use other measures than $\mathcal{M}^{T}$ considering the hop count or delay of received Interests. For instance, one may define a hop count-based measure $\mathcal{M}^{H}$ as follows: $S_{F_i}(\mathcal{I}_n) := |\{j \in \mathcal{I}_n : j$ is satisfied by a Data packet $d$ on $F_i$, where $d$ traversed fewer than $h$~hops$\}|,  \forall F_i \in \mathcal{F}\setminus\{F_D\}$, and $S_{F_D}(\mathcal{I}_n) = |\{j \in \mathcal{I}_n : j$ is satisfied by $F_D\}|$.
\end{remark}
}
}

For the sake of simplicity we write $S_{F_i}$ for $S_{F_i}(\mathcal{I}_n)$ and $U_{F_i}$ for $U_{F_i}(\mathcal{I}_n)$ for the remainder of the paper. Table~\ref{tab:basic_definitions} depicts variables and expressions which are used for the definition of $\mathcal{M}$ and for SAF's update operations. Since these operations are executed during the transition from one period to another, variables hold the observed system state that is observable at the \emph{end} of a given period.

Note that SAF is also scalable concerning space complexity. While classical approaches maintain a list of outgoing faces per prefix, SAF requires a vector of forwarding probabilities instead (cf. Figure~\ref{fig:saf_node}). For instance, assume that those probabilities are quantized into the range of a byte. Furthermore for each combination of prefixes and faces, SAF requires two counting variables (e.g., 4~bytes) representing $S_{F_i}$ and $U_{F_i}$. Then the space complexity for SAF is given by $|\mathcal{C}|\cdot|\mathcal{F}|\cdot9~$bytes. Since $\mathcal{F}$ is usually fixed on a node, space complexity is in $O(|\mathcal{C}|)$, as for most forwarding strategies (cf. Section~\ref{sec:related_work}).
\renewcommand{\arraystretch}{1.15}
\begin{table}[t!]
\small
\centering
\begin{tabular}{|l|l|}
\hline
\textsc{Var./Exp.} &\textsc{Definition / Explanation}\\
\hline
\hline
$S_{F_i}$ & Number of satisfied Interests on $F_i$.\\& $F_D$ satisfies Interests by definition.\\\hline
$U_{F_i}$ & $p(F_i) \cdot I - S_{F_i}$, unsatisfied Interests.\\\hline 
$I$ & $\textstyle\sum_{F_i \in \mathcal{F}}\left[ S_{F_i} + U_{F_i}\right]$, satisfied and \\&  unsatisfied without pending Interests.\\\hline 
$ST_{F_i}$ & \hspace{-0.25cm}
$\begin{cases}
\frac{S_{F_i}}{I} & \text{if } I > 0 \\
0        & \text{otherwise.}
\end{cases}$
$\begin{array}{l}
\text{Satisfied traffic} \\ \text{fraction on $F_i$.} 
\end{array}$\\\hline
$UT_{F_i}$ & \hspace{-0.25cm}
$\begin{cases}
\frac{U_{F_i}}{I} & \text{if } I > 0 \\
0        & \text{otherwise.}
\end{cases}$
$\begin{array}{l}
\text{Unsatisfied traffic} \\ \text{fraction on $F_i$.} 
\end{array}$\\\hline
$R_{F_i}$ & \hspace{-0.25cm}
$\begin{cases}
\frac{S_{F_i}}{S_{F_i}+U_{F_i}} & \hspace{-0.15cm} \text{if } S_{F_i}+U_{F_i} > 0 \\
1        & \hspace{-0.15cm} \text{otherwise. Reliability of $F_i$.}
\end{cases}$
\\\hline
$p(F_i) $ & Forwarding probability for face $F_i$. \\\hline 
$t$ & $t \in [t_{min},t_{max}]$, reliability threshold.\\& $t_{max} \in ]0,1[$, and $t_{min} \in ]0,t_{max}$[. \\\hline 
$\mathcal{F_R} $ & $\{F_i \in \mathcal{F} \setminus \{F_D\} \mid R_{F_i} \geq t\}$, reliable faces. \\\hline 
$\mathcal{F_U} $ & $\mathcal{F} \setminus (\mathcal{F}_R \cup \{F_D\})$, unreliable faces. \\\hline 
$\mathcal{F_S} $ & $\{F_i \in \mathcal{F_R } \mid ST_{F_i} + UT_{F_i} > 0\}$,\\& faces that may take additional traffic.\\\hline
$\mathcal{F_P} $ & $\mathcal{F_R} \setminus \mathcal{F_S}$, faces used for probing.\\\hline
$\delta \in [0,1]$ & Total unsatisfied traffic fraction.\\\hline
$\delta_\mathcal{U} \in [0,1]$ & Unsatisfied traffic on faces in $\mathcal{F_U}$.\\\hline
$\alpha_{F_i} \in \left]0,1\right]$ & $\frac{1}{1+\sqrt{Var(X)}}$, traffic stability indicator, where \\& X is a window over $S_{F_i}$ with length $N$.\\\hline 
$\Delta \in [0,1]$ & $\delta_\mathcal{U}$ with respect to $\alpha_{F_i} \forall F_i \in \mathcal{F_U}$.\\\hline
$\sigma_{F_i}$ & Resources for additional Interests on $F_i$.\\\hline
$\rho \in [0,1]$ & Traffic fraction used for probing.\\\hline
\end{tabular}
\caption{Variables and expressions for SAF.}
\label{tab:basic_definitions}
\end{table}
\renewcommand{\arraystretch}{1.05}

\subsection{Identifying Unsatisfied Traffic}

For all update operations, SAF requires knowledge about the total unsatisfied traffic fraction $\delta$. It provides information about the traffic percentage that has been forwarded towards wrong faces, given the measure $\mathcal{M}$.
Using the expressions from Table~\ref{tab:basic_definitions}, $\delta$ can be defined as given by Equation~\ref{eq:delta}.

\begin{equation}
\label{eq:delta}
\small
\delta =
\begin{cases}
1 - \sum\limits_{F_i \in \mathcal{F}} ST_{F_i} = \sum\limits_{F_i \in \mathcal{F}} UT_{F_i} & \text{if } I > 0, \\
0        & \text{otherwise.}
\end{cases}
\end{equation}

$\delta$ denotes the amount of traffic that should be forwarded on other faces during the next period. However, it is not yet known which faces provide a \textit{poor service} and should therefore receive less traffic, and vice versa. 
For this reason, SAF splits the set of all physical faces $\mathcal{F}\setminus\{F_D\}$ into two disjoint subsets: (i) $\mathcal{F_R}$, the set of reliable faces, and (ii) $\mathcal{F_U}$, the set of unreliable faces. This partitioning is based on the definition of the reliability of a face $R_{F_i}$ and the dynamic reliability threshold $t$. The threshold $t$ is adapted based on a node's \textit{health} status in the interval $t \in [t_{min},t_{max}]$, which will be discussed later. Note that the definition of $R_{F_i}$ is based only on the defined measure $S_{F_i}$. 

The partitioning of $\mathcal{F}\setminus\{F_D\}$ into $\mathcal{F_R}$ and $\mathcal{F_U}$ provides a starting point to improve a node's forwarding decisions. SAF's update operations focus on shifting traffic from the unreliable faces towards the reliable faces. The next step for SAF is to evaluate the amount of traffic from faces in $\mathcal{F_U}$ that can be shifted to faces in $\mathcal{F_R}$ without overloading those. For this purpose, we define $\delta_\mathcal{U}$ (cf. Equation \ref{eq:delta_u}), which specifies the accumulated unsatisfied traffic from faces in $\mathcal{F_U}$ only. 
%
\begin{equation}
\small
\label{eq:delta_u}
\delta_\mathcal{U}=
\begin{cases}
\sum\limits_{F_i \in \mathcal{F_U}} UT_{F_i} & \text{if } I > 0, \\
0        & \text{otherwise.}
\end{cases}
\end{equation}
%

\subsection{Update Operations}
\label{sec:upd_ops}
At the end of each period, SAF's objective is to shift the traffic fraction $\delta_\mathcal{U}$ from $\mathcal{F_U}$ to $\mathcal{F_R}$, or, in the worst case, to the virtual face $F_D$. In order to neglect short-term effects, the shifting of traffic is relaxed by $\alpha_{F_i} \in ]0,1]$. $\alpha_{F_i}$ is an indicator for the stability of the satisfied traffic over $F_i$ and is defined in Table~\ref{tab:basic_definitions}. It is determined by the standard deviation of the satisfied Interests over a given number of periods. Note that, the larger the standard deviation of $S_{F_i}$ over the periods, the smaller $\alpha_{F_i}$, and vice versa. This ensures that SAF balances the updates of the FWT taking into account traffic stability, which allows stronger changes if the observed state is steady over the periods. Plugging $\alpha_{F_i}$ into Equation~\ref{eq:delta_u} provides the relaxed unsatisfied traffic fraction $\Delta$, as denoted in Equation~\ref{eq:delta_large}.
\begin{equation}
\small
\label{eq:delta_large}
\Delta = 
\begin{cases}
\sum\limits_{F_i \in \mathcal{F_U}} UT_{F_i} \cdot \alpha_{F_i} & \text{if } I > 0, \\
0        & \text{otherwise.}
\end{cases}
\end{equation}

\begin{algorithm}[tbh]
\small
\caption{Pseudocode for the FWT updates in SAF}
\label{alg:updateOperations}
\begin{algorithmic}[1]

\STATE $\Delta \leftarrow \texttt{determineUnsatisfiedTraffic()}$
\STATE {$\Gamma \leftarrow \Delta + p(F_D)$}
\IF {$\Gamma > 0$}
    \STATE $\mathcal{F_S}, \mathcal{F_P} \leftarrow \texttt{splitSet(} \mathcal{F_R} \texttt{)}$
    \STATE {$\texttt{shiftTraffic(} \mathcal{F_U} \texttt{,} \mathcal{F_S} \texttt{,} \Gamma \texttt{)}$}
    \STATE {$p(F_D) = \Gamma - \Gamma'$}
    
    \IF {$p(F_D) > 0$}
        \STATE {$\texttt{probeOnFaces(} \mathcal{F_P} \texttt{)}$}
        \IF {$p(F_D) > (1-t)$}
            \STATE {$\texttt{decreaseReliability(t)}$}
        \ENDIF
    \ENDIF

\ELSIF {$I > 0$}
    \STATE {$\texttt{increaseReliability(t)}$}
\ENDIF

\end{algorithmic}
\end{algorithm}



Algorithm~\ref{alg:updateOperations} outlines SAF's update procedure, which can be used as the road map for Subsections~\ref{sec:upd_ops} to \ref{sec:prob}.
In line 2 of Algorithm~\ref{alg:updateOperations}, $\Gamma$ is introduced. $\Gamma$ denotes the sum of the unsatisfied traffic fraction $\Delta$ and the current forwarding probability $p(F_D)$ of the virtual face $F_D$. It is important to consider the sum of $\Delta$ and $p(F_D)$, as $\Delta$ does not consider the discarded traffic fraction on $F_D$.  Note that $\Delta > 0 \Leftrightarrow \mathcal{F}_U \neq \emptyset$ as indicated by Proposition~\ref{prop:delta}. \textit{Remark}: $\Delta = 0 \not\Rightarrow UT_{F_i} = 0 : \forall F_i \in \mathcal{F}$.

\begin{proposition}
\label{prop:delta}
Suppose $\forall F_i \in \mathcal{F} : \alpha_{F_i} \in ]0,1]$ and $I > 0$. Then, it holds that $\Delta > 0 \iff \mathcal{F_{U}} \neq \emptyset$.
\end{proposition}
\begin{proof}[Proof of Prop. \ref{prop:delta}] Follows directly from the definition of $\Delta$ and $\mathcal{F_U}$ (cf. Table \ref{tab:basic_definitions}).
\end{proof}

Algorithm~\ref{alg:updateOperations} has a trivial case, which eventuates if $\Gamma = 0$. In this case no changes in the FWT are required, since neither unsatisfied traffic exists nor any traffic is dropped in advance. In this favorable case, $t$ is increased if $I > 0$ in this period. However, if $\Gamma > 0$, SAF resolves the unsatisfied traffic using the two following approaches:
\begin{inparaenum}[\itshape i\upshape)]
\item adaptation of the forwarding probabilities within the FWT (cf. Alg.~\ref{alg:updateOperations} line 4-5);
\item identification of yet unknown paths to the desired content via probing (cf. Alg.~\ref{alg:updateOperations} line 8).
\end{inparaenum}
For the sake of simplicity, we separate the further discussion of Algorithm~\ref{alg:updateOperations} into these two parts.
The adaptation of the reliability threshold (cf. Alg.~\ref{alg:updateOperations} lines 10 and 14) is discussed at the end of the second part.

\revised{Note that Algorithm~\ref{alg:updateOperations} is executed for each content prefix in the FWT (cf. Figure~\ref{fig:saf_node}). If we assume the worst case for Algorithm~\ref{alg:updateOperations}  ($\Gamma > 0$, $p(F_D) > 0$, and $p(F_D) > (1-t)$), then the asymptotic time complexity is given by $O(\mathcal{|C|})$. This is due to the fact that only simple arithmetic operations are used (cf. Subsection~\ref{sec:adap_fp} and~\ref{sec:prob}), and their quantity solely depends on the number of faces $|\mathcal{F}|$, which is usually constant.}

\subsection{Adaptation of Forwarding Probabilities}
\label{sec:adap_fp}

SAF shifts traffic between faces only if $\Gamma > 0$. The principal objective in this step is to shift the unsatisfied traffic from the unreliable faces $\mathcal{F_U}$ towards the reliable faces $\mathcal{F_R}$ (cf. Alg.~\ref{alg:updateOperations} line 5) without overloading them. Of course, this is not always possible, which in the worst case forces the algorithm to forward some Interests towards the virtual face $F_D$ (cf. Alg.~\ref{alg:updateOperations} line 6). Before any actions are taken, SAF splits the set $\mathcal{F_R}$ into two disjoint subsets. The subset $\mathcal{F_S} \subseteq \mathcal{F_R}$ includes only faces from $\mathcal{F_R}$ which have successfully forwarded Interests in the current period. \revised{The second subset, $\mathcal{F_P} = \mathcal{F_R} \setminus \mathcal{F_S}$ (cf. Table~\ref{tab:basic_definitions}), includes faces that are considered as reliable only because they have not forwarded any Interests ($\forall F_i \in \mathcal{F_P}: S_{F_i} = U_{F_i} = 0$) in the current period. So it is very likely that faces  in $\mathcal{F_P}$ cannot fulfill requests, which is why we do not consider them for \textit{attracting additional} traffic.}

Before SAF performs the shifting, it determines how much additional traffic the faces in $\mathcal{F_S}$ may take, without decreasing their reliability below $t$. Proposition~\ref{prop:sigma} provides $\sigma_{F_i}$, which denotes the number of additional Interests the face $F_i \in \mathcal{F_S}$ may take without dropping $R_{F_i}$ below $t$.

%
\begin{proposition}
\label{prop:sigma}
For a given reliability $t$, every $F_i \in \mathcal{F_S}$ can satisfy $ 0 \leq \sigma_{F_i} \leq \lfloor\frac{S_{F_i}}{t} - S_{F_i} - U_{F_i}\rfloor$ additional Interests.
\end{proposition}
\begin{proof}[Proof of Prop. \ref{prop:sigma}] Follows directly from the definition of $R_{F_i}$ (cf. Table \ref{tab:basic_definitions}).
\end{proof}
%

Given $\Gamma§$ and $\sigma_{F_i} \forall F_i \in \mathcal{F_S}$, SAF is able to determine the maximum traffic that can/should be shifted from $\mathcal{F_U}$ to $\mathcal{F_S}$. We denote this amount as $\Gamma'$ as defined in Equation~\ref{eq:determine_shift_traffic}.
\begin{equation}
\small
\label{eq:determine_shift_traffic}
\Gamma' = \min\Big(\frac{1}{I} \cdot \sum\limits_{F_i \in \mathcal{F}_S} \sigma_{F_i}, \Gamma \Big)
\end{equation}
The next step for SAF is to determine the forwarding probabilities for period $n+1$ by: 
\begin{inparaenum}[\itshape i\upshape)]
\item decreasing the forwarding probabilities for faces in $\mathcal{F_U}$ by $\Gamma'$, Equation~\ref{eq:decrease_fu}; 
\item increasing the forwarding probabilities for faces in $\mathcal{F_S}$ by $\Gamma'$, Equation~\ref{eq:increase_fs}.
\end{inparaenum}
\begin{equation}
\small
\label{eq:decrease_fu}
\forall F_i \in \mathcal{F_U} :  p_{n+1}(F_i)  \leftarrow p_n(F_i) - UT_{F_i} \cdot \alpha_{F_i}
\end{equation}
\begin{equation}
\small
\label{eq:increase_fs}
\forall F_i \in \mathcal{F_S} : p_{n+1}(F_i) \leftarrow p_n(F_i) + \Gamma' \cdot \frac{\sigma_{F_i}}{\sum\limits_{F_i \in \mathcal{F}_S} \sigma_{F_i}}
\end{equation}
Note that Equation~\ref{eq:decrease_fu}, does not use $\Gamma'$ to determine the amount of the traffic reduction. Instead the complete unsatisfied traffic from a face $F_i \in F_U$ considering $\alpha_{F_i}$ is removed, which exactly adds up to $\Delta$. As the residual traffic  ($\Gamma\ - \Gamma'$) can not be satisfied by any $F_i \in \mathcal{F} \setminus \{F_D\}$, it is beneficial to drop this portion of the traffic. Forwarding those Interests would likely cause congestion and impair the performance. So SAF simply determines the amount of residual unsatisfied traffic and puts it on the virtual face $F_D$, denoted by Equation~\ref{eq:set_pfd} (cf. Alg.~\ref{alg:updateOperations} line 6).
Theorem~\ref{thm:convergence} shows that SAF converges to a steady state, which is defined as a state where $\mathcal{F_U} = \emptyset$.
\begin{equation}
\label{eq:set_pfd}
\small
p(F_D) = \Gamma - \Gamma'
\end{equation}
%
\begin{theorem}
\label{thm:convergence}
Given any state with $\mathcal{F_U} \neq \emptyset$, $t \in [t_{min},t_{max}]$ and without the loss of generality set $\alpha$ sufficiently small (i.e., $\alpha = \min (\alpha_{F_i} : F_i \in \mathcal{F_U})$). $d_{F_i}$ denotes the number of  Interests that can be satisfied on face $F_{i}$, $p^*(F_i) =  \textstyle\frac{d_{F_i}}{I}$ denotes the optimal forwarding probability for face $F_i$ and $I$ denotes the number of Interests that shall be forwarded in every period. Assume that $I$ is constant for every period and that $I \cdot p(F_i) \geq \textstyle\frac{d_{F_i}}{t}$. According to Equations~\ref{eq:decrease_fu},~\ref{eq:increase_fs} and~\ref{eq:set_pfd} SAF converges to $\mathcal{F_{U}} = \emptyset$ after $n$ periods (iterations) bounded by:
\begin{equation}
\small
\label{eq:convergence}
n \leq \underset{F_{i} \in \mathcal{F_U}}{\max}\left(\left\lceil\frac{ln(\frac{d_{F_i}}{t} - d_{F_i}) - ln(p_{0}(F_{i}) \cdot I - d_{F_i})}{ln(1-\alpha)}\right\rceil\right),
\end{equation}
where $p_{0}$ denotes the initial forwarding probability.
\end{theorem}

\begin{proof}[Proof of Theorem \ref{thm:convergence}]
In the case of $\mathcal{F_{U}} \neq \emptyset$, SAF uses Equation \ref{eq:decrease_fu} to reduce the forwarding probabilities on every $F_{i} \in \mathcal{F_U}$ until all become reliable. According to the definition of $R_{F_i}$, $F_i \in \mathcal{F_U}$ iff $p_{n}(F_{i}) \cdot I > \frac{d_{F_i}}{t}$. In this case we either shift the probabilities $\forall F_{i} \in \mathcal{F_{U}}$ to faces in $\mathcal{F_S}$, or we drop the traffic by increasing the forwarding probability of $F_D$. Since we may express $UT_{F_i}$ as $ p_n(F_i) - \frac{d_{F_i}}{I}$ assuming a perfect random distribution for Algorithm~\ref{alg:chooseFace}, and according to Equation~\ref{eq:decrease_fu} we have, $\forall F_i \in \mathcal{F_U}$:
\begin{equation}
\label{eq:optim_eqn}
\small
 p_{n+1}(F_i) = p_n(F_i) \cdot (1-\alpha) + \alpha \cdot \frac{d_{F_i}}{I}.
\end{equation}
Solving the recursion we have,
\begin{equation}
\label{eq:sequence_solved}
\small
p_{n}(F_i) = p_0(F_i) \cdot (1-\alpha)^n + \alpha \cdot \frac{d_{F_i}}{I} \cdot \sum_{j=0}^{n-1} (1-\alpha)^{j}.
\end{equation}
Equation~\ref{eq:sequence_solved} and $\underset{n\rightarrow \infty}{lim} p_n(F_i) = \textstyle\frac{d_{F_i}}{I}$ provides the claim that $p_n(F_i) \rightarrow p^*(F_i)$ and  with $t \rightarrow 1$, $p^*(F_i) = \textstyle\frac{d_{F_i}}{I}$ denotes the optimum for face $F_i$. Further, one easily shows that the convergence speed ($|p_{n+1}(F_i) - p^*(F_i)| \leq M \cdot |p_{n}(F_i) - p^*(F_i)|$) is linear with $M = (1-\alpha)$.
By plugging $p_n(F_i)$ into  $p_{n}(F_{i}) \cdot I > \frac{d_{F_i}}{t}$ and using the formula for the geometric series ($0 < 1-\alpha < 1$) we have,
\begin{equation}
\small
(1-\alpha)^n \cdot (p_0(F_i) \cdot I - d_{F_i}) + d_{F_i} > \frac{d_{F_i}}{t}.
\end{equation}
Then the number of periods $n$ until all unreliable faces become reliable is bounded by Equation~\ref{eq:convergence}.
\end{proof}

{\small
\begin{remark}
In Theorem \ref{thm:convergence} we assumed a constant number of Interests in each period. In order to show that Theorem  \ref{thm:convergence} is still valid given a varying number of Interests for each time dependent period, we may define a period by the number of Interests. Thus, we again have a constant number of Interests during a period and Theorem \ref{thm:convergence} holds. 
\end{remark}
\begin{remark}
\label{remark:stability}
Let $(\mathbb{R},| x - y |)$ be a metric space. According to Equation \ref{eq:optim_eqn}, $p^*(F_i)$ is a fixed point. This fixed point is globally asymptotic stable for $\alpha \in ]0,1]$.
\end{remark}

\begin{proof}[Proof of Remark \ref{remark:stability}]
Let $\Phi(k,\kappa,\xi) = \xi (1-\alpha)^{k-\kappa} + \alpha \cdot \textstyle{\frac{d_{F_i}}{I}} \cdot \textstyle{\sum_{j=\kappa}^{k-1}(1-\alpha)^{j+1}} \cdot \id_{\{k\geq\kappa\}}$ with initial pair $(\xi, \kappa) \in \mathbb{R} \times \mathbb{N}$, $p_\kappa = \xi$ and $k\geq\kappa$ be the general solution to our autonomous difference equation given in Equation \ref{eq:optim_eqn}.
First, we show that $p^*(F_i)$ is a fixed point and it is globally attractive. 
\begin{equation}
\small
p_n(F_i) \cdot (1-\alpha) + \alpha \cdot \frac{d_{F_i}}{I} = p_n(F_i) \nonumber \iff
p_n(F_i) =  \frac{d_{F_i}}{I}. \nonumber
\end{equation}
According to Theorem \ref{thm:convergence}, it follows that for $\alpha \in ]0,1]$ the fixed point $p^*(F_i)$ is globally attractive ($\forall (\xi, \kappa) \in \mathbb{R}\times\mathbb{N} :  {lim}_{n\rightarrow \infty} |p_n(F_i) - p^*(F_i)| = 0$).
Second, we show that the solution $p*=\textstyle{\frac{d_{F_i}}{I}}$ is stable for all $\alpha \in ]0,1]$. For $\alpha \in ]0,1[$, $\forall k\in\mathbb{N}, k\geq\kappa$ and for any $\varepsilon > 0$ we have,
{\small
\begin{align}
&|\Phi(k,\kappa,\xi) - p^*| = \nonumber \\ 
&|\xi \cdot (1-\alpha)^{k-\kappa} + \alpha \cdot \frac{d_{F_i}}{I} \cdot \sum_{j=\kappa}^{k-1}(1-\alpha)^{j+1}  \cdot \id_{\{k>\kappa\}} -  \frac{d_{F_i}}{I}| \leq \nonumber \\ 
&|(1-\alpha)^{k-\kappa} \cdot \Big(\xi - \frac{d_{F_i}}{I}\Big)| \leq |\xi - \frac{d_{F_i}}{I}| < \varepsilon \quad \forall \xi \in B_\delta\Big(\frac{d_{F_i}}{I}\Big), \nonumber
\end{align}
}%
with $\delta = \textstyle{\frac{\varepsilon}{2}}$, where $B_\delta(y) := \{x \in \mathbb{R} : |x - y| < \delta\}$. 
For $\alpha=1$ the proof is analogues. Thus, $p^*$ is globally asymptotic stable.
\end{proof}
}

\subsection{Probing to Identify Unknown Paths}
\label{sec:prob}
After shifting traffic in the previous step, $p(F_D)$ holds the residual traffic that cannot be forwarded on the physical faces. This traffic is going to be dropped by the virtual face $F_D$. SAF's probing mechanism takes a share of this traffic and uses it to discover new paths towards the content origin or to discover nodes holding cached replicas. 
The fraction of the traffic that is used for probing is limited by $\rho$ as denoted in Equation~\ref{eq:rho}. The probe is defined by Equation~\ref{eq:probe}.
\begin{equation}
\label{eq:rho}
\small
\rho = 1 - \textstyle\sum_{F_i \in \mathcal{F} \setminus F_D} ST_{F_i} = 1 - (\delta - ST_{F_D})
\end{equation}
\begin{equation}
\label{eq:probe}
\small
probe = p(F_D) \cdot \rho
\end{equation}

$\rho$ increases based on the proportion of the satisfied and the unsatisfied traffic on physical faces. The larger the fraction of the unsatisfied traffic, the larger $\rho$. The idea behind this is that the more unsatisfied traffic we have, the more important it is to discover additional paths to the content. In the worst case where $p(F_D) = 1$ the entire traffic is used for probing. For instance, this can happen when a previously working path to the content suffers from link failure(s). In this case probing may help to circumvent the broken link or to identify yet unknown path(s) to cached content replicas.
The probe is uniformly distributed on all faces in $\mathcal{F_P}$ as denoted in Equation~\ref{eq:split_probe}. Finally, the forwarding probability of the virtual face $F_D$ needs to be adjusted as outlined in Equation~\ref{eq:fd_probe}.
\begin{equation}
\label{eq:split_probe}
\small
\forall F_i \in \mathcal{F_P} : p_{n+1}(F_i) = p_n(F_i) + \textstyle\frac{probe}{|\mathcal{F_P}|}
\end{equation}
\begin{equation}
\small
\label{eq:fd_probe}
p_{n+1}(F_D) = p_n(F_D) - probe
\end{equation}

After the probing phase has been carried out, SAF checks if it has to decrease the reliability threshold (cf. Algorithm~\ref{alg:updateOperations} lines 9--10). If $p(F_D)$ is larger than $(1-t)$, the reliability threshold has to be decreased since it cannot be retained. In contrast to this, the reliability threshold should be increased if currently all Interests can be satisfied with a reliability of at least $t$ (cf. Algorithm~\ref{alg:updateOperations} lines 13--14). We suggest Equations~\ref{eq:inc_t} and~\ref{eq:dec_t} for the adjustment of $t$, where $\lambda$ denotes the rate of change.
\begin{equation}
\small
\label{eq:inc_t}
t_{n+1} = (1-\lambda) \cdot t_n + \lambda \cdot t_{max}
\end{equation}
\begin{equation}
\small
\label{eq:dec_t}
t_{n+1} = (1-\lambda) \cdot t_n - \lambda \cdot t_{min}
\end{equation}

\subsection{SAF Examples}
We present two examples, one illustrating how SAF approaches the optimal FWT, and a second one illustrating the probing mechanism that allows SAF to recover from a link failure in the given scenario. Figure~\ref{fig:toy_example} depicts the example network with the corresponding link capacities in Mbps. Our object of investigation is the router $R$. This router has to forward Interests for a given content prefix $c$, which can be retrieved at the content provider. $R$ maintains three physical faces which have links to the routers $F_0, F_1,$ and $F_2$ (we use faces from $R$ interchangeably with routers $F_0, F_1,$ and $F_2$). Interests forwarded via $F1$ and $F2$ reach the content provider. 
We assume that $R$ has to satisfy a constant flow of Interests requesting 3~Mbps of Data packets. Furthermore we assume that the initial value for $t = 0.5$ and $\forall F_i \in \mathcal{F} : \alpha_{F_i} = 1$ for all periods. 
The initial FWT for $R$ can either be provided by the routing layer, or, as in this example, the traffic is uniformly distributed among the faces (cf. Figure~\ref{subfig1-tab-init}). During the first period each of the faces $F_0, F_1,$ and $F_2$ receives Interests requesting Data packets for a bitrate of 1~Mbps. While $F_1$ and $F_2$ satisfy all Interests, $F_0$ is not able to satisfy any. Therefore, $F_0$ is classified as an unreliable face and as $\alpha_{F_0} = 1$ the entire unsatisfied traffic is removed (cf. Eq.~\ref{eq:decrease_fu}). The forwarding probability of $F_0$ is shifted towards $F_1$ and $F_2$ (cf. Eq.~\ref{eq:increase_fs} and Figure~\ref{subfig1-tab-step1}). In the second period both, $F_1$ and $F_2$, receive Interests requesting Data packets for 1.5~Mbps. While $F_2$ is able to satisfy all Interests, $F_1$ satisfies only $\textstyle\frac{2}{3}$. However, based on $t$ both faces are considered as reliable and no changes are performed on the FWT (cf. Figure~\ref{subfig1-tab-step2}). Only the threshold $t$ is increased (cf. Alg.~\ref{alg:updateOperations} line 14), assume $t$ is increased to $0.75$. In the third period the traffic is distributed as in the second period, with the same issue that $F_1$ can satisfy only $\textstyle\frac{2}{3}$ of the Interests. However, this time $t > \frac{2}{3}$, which classifies $F_1$ as unreliable. Therefore, the unsatisfied traffic ($\approx 0.08$) is shifted towards $F_2$ (cf. Figure~\ref{subfig1-tab-step3}). With each further period, the FWT approaches the optimal distribution ($p(F_1)=\frac{1}{3}$, $p(F_2)=\frac{2}{3}$).

\begin{figure}[tb!]
\centering
\includegraphics[scale=0.64]{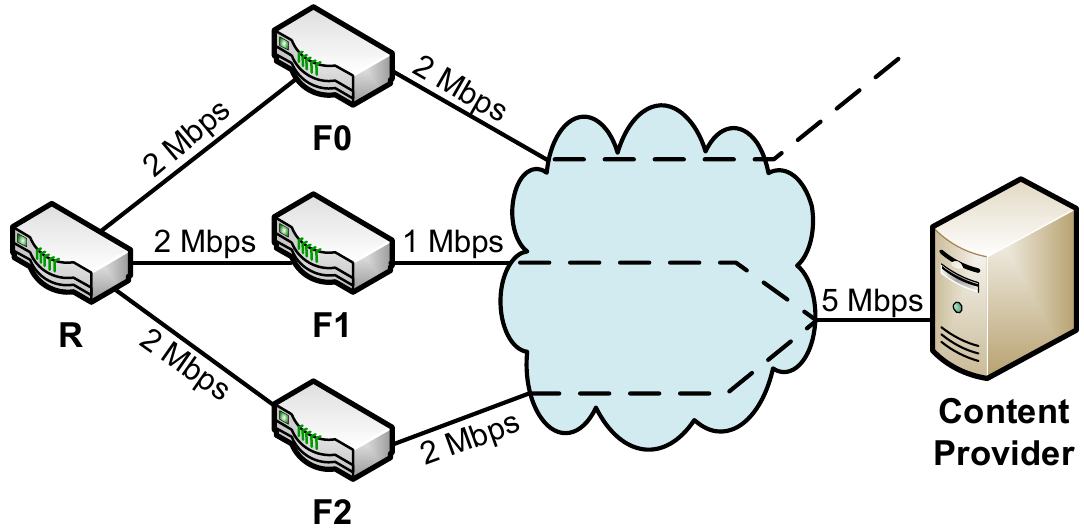}
\vspace{-0.1cm}
\caption{Example network.}
\label{fig:toy_example}
\end{figure}

For the second example we consider the same network as before with one major change. The path traversing $F_0$ ends at a content replica of the content provider. We use the optimal FWT from the previous example as the starting point for the second example, $t=0.99$ and $\forall F_i \in \mathcal{F} : \alpha_{F_i}=1$ for all periods.
We further assume that, at the beginning of the first period, a link failure on the path between $F_2$ and the content provider occurs. Therefore, none of the Interests can be satisfied via $F_2$, which causes this face to be marked as unreliable. As $t=0.99$, $F_1$ is not able to take any additional Interests from $F_2$ causing this traffic portion ($\frac{2}{3}$) to be shifted to the virtual face $F_D$. As $p(F_D) > 0$, probing is issued (cf. Alg.~\ref{alg:updateOperations} line 7). This results in the situation that a share of $p(F_D) \cdot \rho = \frac{2}{3} \cdot \frac{2}{3}$ Interests is used for probing on face $F_0$ (cf. Eq~\ref{eq:rho} and~\ref{eq:split_probe}). Figure~\ref{subfig2-tab-step1} illustrates the FWT after the first iteration. Since $p(F_D) > 1 - t$, $t$ is decreased. Assume $t$ is decreased to $t=0.75$. In the second period, all Interests forwarded via $F_0$ and $F_1$ can be satisfied and SAF is able to discover a new path to the content via $F_0$. After the first iteration, still 22.2\% of the Interests are dropped in advance. Since in the second period all Interests forwarded via $F_0$ and $F_1$ can be satisfied and $\sigma_{F_0} + \sigma_{F_1} > p(F_D)$ ($\approx 0.148+0.111>0.222$) the forwarding probabilities for $F_D$ are distributed according to Eq.~\ref{eq:increase_fs} (cf. Figure~\ref{subfig2-tab-step2}). As the capacity of $F_1$ was already exhausted, a part of the traffic ($\approx 0.43-0.33 \approx 0.1$) is dropped during the third period due to congestion. However, according to $t=0.75$, $F_1$ is still reliable and since all faces are considered as reliable during the third period the FTW is not modified between iteration two and three. Only the reliability threshold is increased, assume $t$ is increased to $0.85$. Therefore, after the fourth period $F_1$ is considered as unreliable and the unsatisfied traffic is completely taken away according to Eq.~\ref{eq:decrease_fu} (as $\forall F_i \in \mathcal{F} : \alpha_{F_i}=1$ is assumed). This traffic portion is then distributed among the reliable face(s) in $\mathcal{F_S}$ (in this case only $\mathcal{F_S}=\{F_0\}$), and as $\sigma_{F_0} > UT_{F_1}$, SAF reaches the optimal FTW as shown in Figure~\ref{subfig2-tab-step3}.

\begin{figure}[tb!]
\hspace{-1.5cm}
\resizebox{.41\linewidth}{!}{
\subfloat[Initial\label{subfig1-tab-init}]{%
\kbordermatrix{
    t=0.5 & c\\
    F_D & 0\\
    F_0 & 1/3\\
    F_1 & 1/3\\
    F_2 & 1/3
  }
}}\hspace{-1.5cm}
\resizebox{.41\linewidth}{!}{
\subfloat[Iteration 1\label{subfig1-tab-step1}]{%
\kbordermatrix{
    t=0.5 & c\\
    F_D & 0\\
    F_0 & 0\\
    F_1 & 1/2\\
    F_2 & 1/2
  }
}}\hspace{-1.5cm}
\resizebox{.41\linewidth}{!}{
\subfloat[Iteration 2\label{subfig1-tab-step2}]{%
\kbordermatrix{
    t=0.75 & c\\
    F_D & 0\\
    F_0 & 0\\
    F_1 & 1/2\\
    F_2 & 1/2
  }
}}\hspace{-1.5cm}
\resizebox{.41\linewidth}{!}{
\subfloat[Iteration 3\label{subfig1-tab-step3}]{%
\kbordermatrix{
    t=0.75 & c\\
    F_D & 0\\
    F_0 & 0\\
    F_1 & 0.42\\
    F_2 & 0.58 
  }
}}
\vspace{-0.1cm}
\caption{The FWTs of router $R$ for the \textbf{first} example.}
\label{fig:example1}
\end{figure}

\begin{figure}[tb!]
\hspace{-1.5cm}
\resizebox{.41\linewidth}{!}{
\subfloat[Initial\label{tab-init}]{%
\kbordermatrix{
    t=0.99 & c\\
    F_D & 0\\
    F_0 & 0\\
    F_1 & 1/3\\
    F_2 & 2/3
  }
}}\hspace{-1.5cm}
\resizebox{.41\linewidth}{!}{
\subfloat[Iteration 1\label{subfig2-tab-step1}]{%
\kbordermatrix{
    t=0.75 & c\\
    F_D & 2/9\\
    F_0 & 4/9\\
    F_1 & 1/3\\
    F_2 & 0\\
  }
}}\hspace{-1.5cm}
\resizebox{.41\linewidth}{!}{
\subfloat[Iter. 2 \emph{and} 3 ($t=0.85$)\label{subfig2-tab-step2}]{%
\kbordermatrix{
    t=0.75 & c\\
    F_D & 0\\
    F_0 & 0.57\\
    F_1 & 0.43\\
    F_2 & 0
  }
}}\hspace{-1.5cm}
\resizebox{.41\linewidth}{!}{
\subfloat[Iteration 4\label{subfig2-tab-step3}]{%
\kbordermatrix{
    t=0.85 & c\\
    F_D & 0\\
    F_0 & 2/3\\
    F_1 & 1/3\\
    F_2 & 0
  }
}}
\vspace{-0.1cm}
\caption{The FWTs of router $R$ for the \textbf{second} example.}
\label{fig:example2}
\end{figure}

\section{Evaluation}
\label{sec:evaluation}

\begin{figure*}[tbh!]
\subfloat[Sample topology \emph{LowCon}\label{uniform_subfig-1:LowCon}]{%
\includegraphics[width=0.24\textwidth]{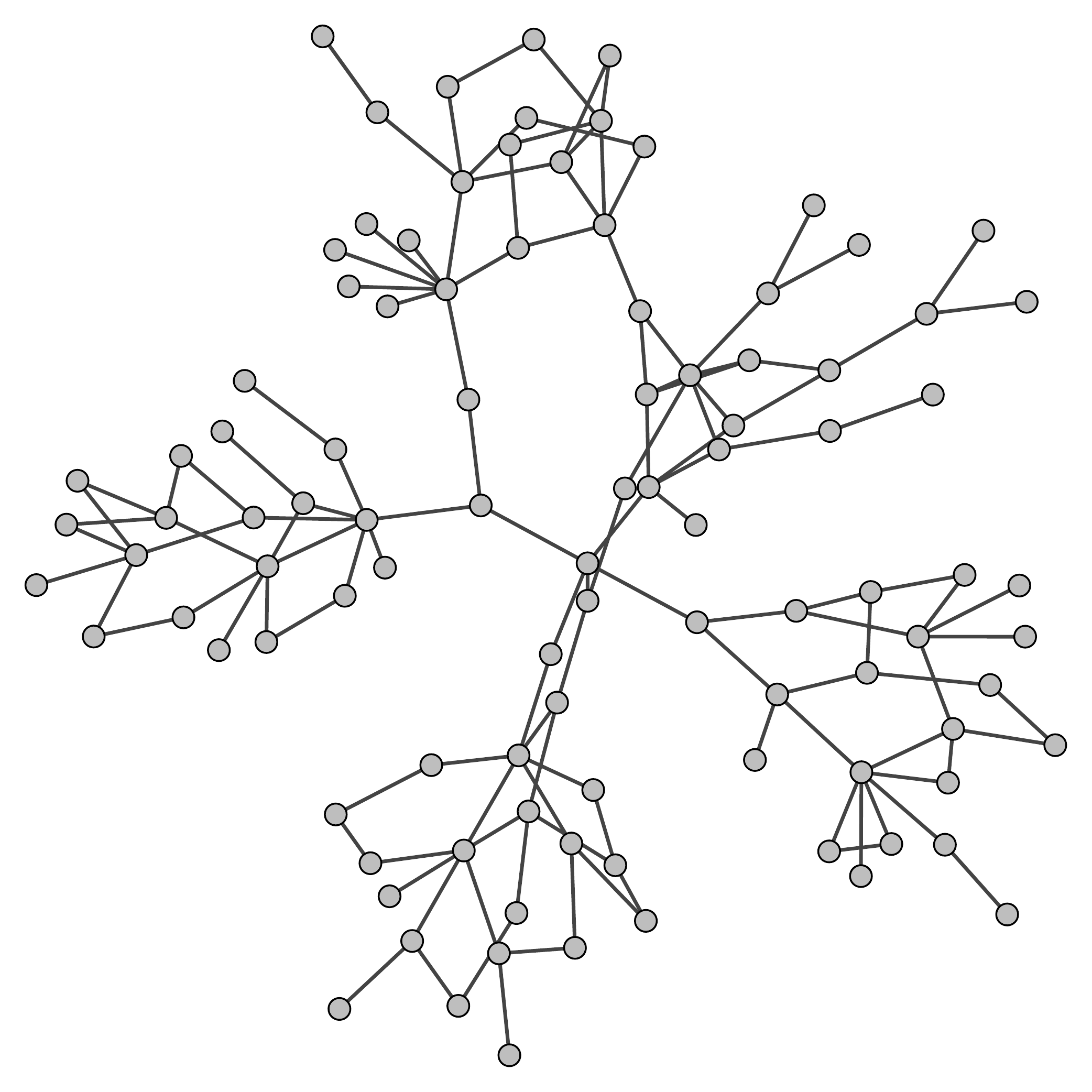}
}
\hfill
\subfloat[Sample topology \emph{MediumCon}\label{uniform_subfig-2:MediumCon}]{%
\includegraphics[width=0.24\textwidth]{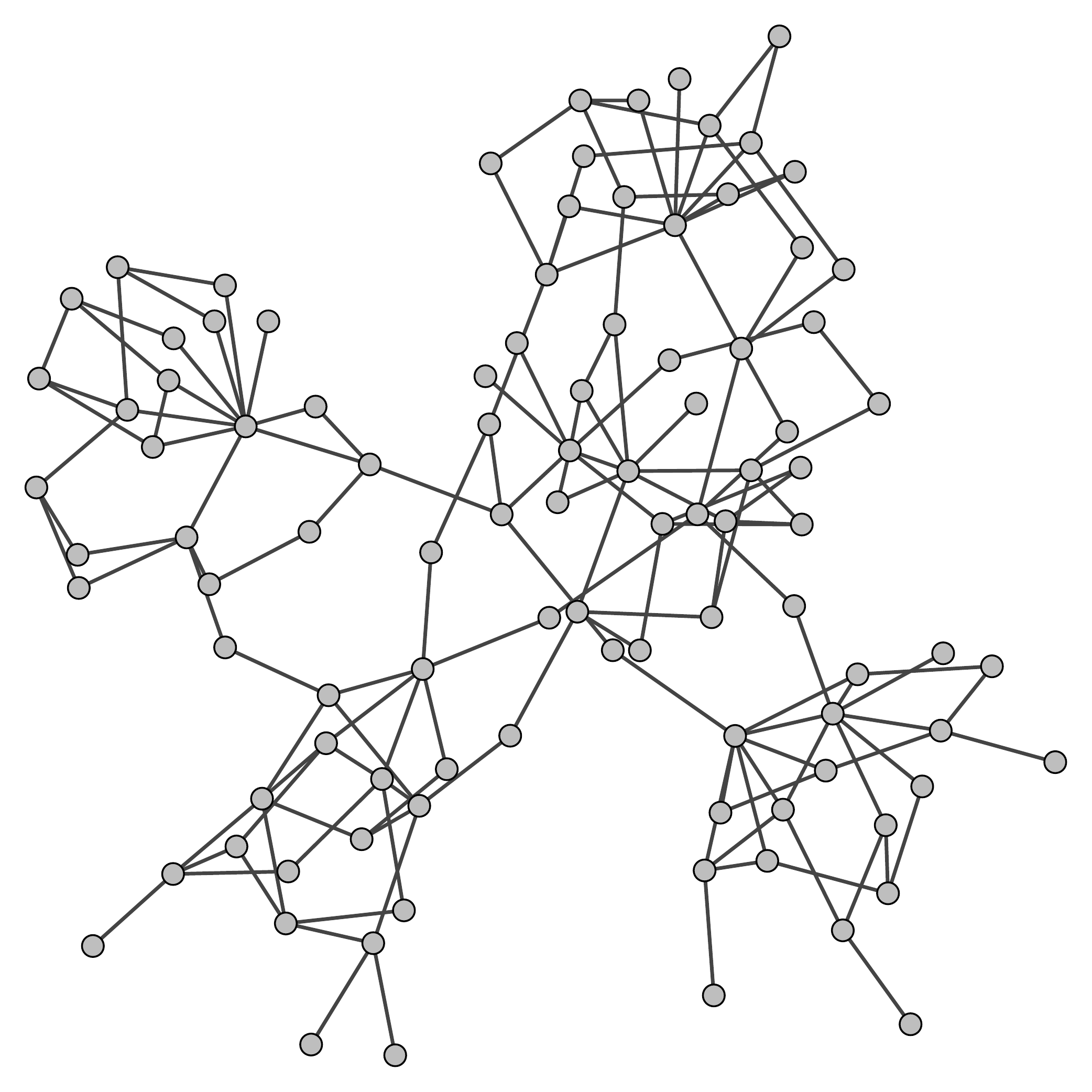}
}
\hfill
\subfloat[Sample topology \emph{HighCon}\label{uniform_subfig-3:HighCon}]{%
\includegraphics[width=0.24\textwidth]{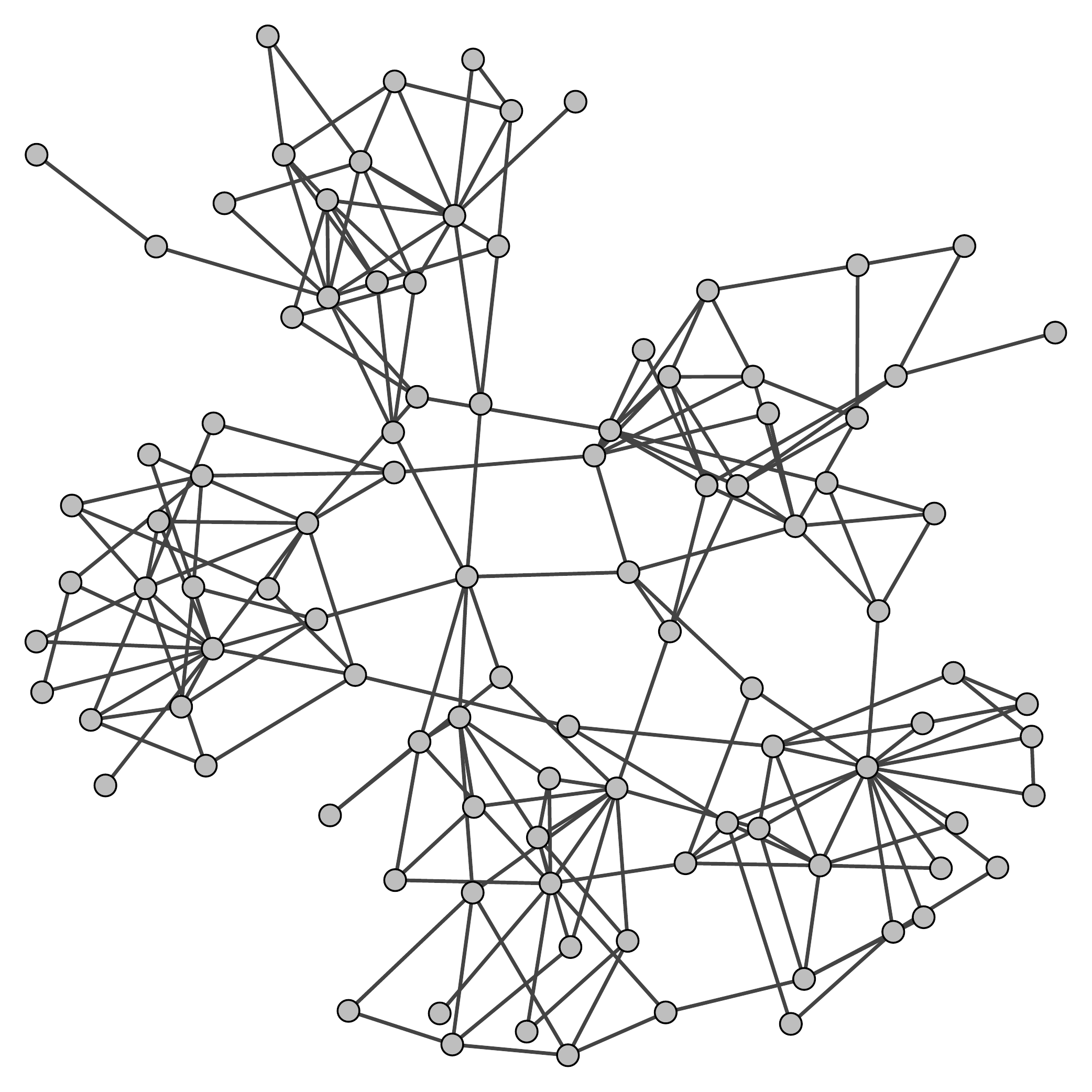}
}
\vspace{-0.1cm}
\caption{Sample topologies for the different network \textbf{connectivities} as defined in Table~\ref{tab:connectivity} (without client/server nodes).}
\label{fig:network_connectivities}
\vspace{-0.35cm}
\end{figure*}

\revised{In this section we investigate the performance of SAF by comparing it to related work and standard forwarding algorithms already present in the ns3/ndnSIM 2.0 simulator~\cite{ndnSIM}. The selected algorithms from related work are: ideal Nearest Replica Routing (iNRR)~\cite{Rossini:2014}, the Request Forwarding Algorithm (RFA)~\cite{Carofiglio:2013}, and the On-demand Multi-Path Interest Forwarding (OMP-IF)~\cite{multipathInterestFW}. As already mentioned in Section~\ref{sec:related_work}, iNRR uses an oracle to determine for every Interest the nearest node which holds a cached copy of the corresponding Data packet and to provide the shortest path to this node~\cite{Rossini:2014}. RFA has been implemented as described in~\cite{Carofiglio:2013}. OMP-IF has been implemented according to the descriptions given in~\cite{multipathInterestFW}. We provide an open source implementation of SAF together with the implementations of iNRR, RFA, and OMP-IF in ns3/ndnSIM 2.0 at \emph{github.com/danposch/SAF}. For SAF we use the presented measure $\mathcal{M}^{T}$, thus, maximizing the throughput for specific content prefixes at every node in the network.}

\subsection{Network Topology Generation}
\label{subsec:network_topologies}
%
%


For generating random network topologies we employ the network topology generator BRITE~\cite{brite}. BRITE was configured to build scale-free networks in a top-down fashion as the Internet topology is likely to be described by power-laws~\cite{Faloutsos:1999}. \revised{We model the infrastructure of large ISPs, interconnecting several autonomous systems and access networks.} The top level represents $\mu=5$ autonomous systems (AS). Each AS maintains $\nu=20$ nodes (bottom level) acting as ICN routers \revised{(in total we have $5\cdot20=100$ ICN routers)} \revised{ and serving as access nodes for later on added client and server nodes}. Both the top- and the bottom-level graphs are randomly generated based on the Barab\'{a}si-Albert model. The model was configured to connect each node with exactly one neighbor. This generates scale-free networks with no redundant paths. As ICN characteristics such as multi-path delivery and link-failure recovery can only be evaluated properly if redundant paths exist, we extended the generated graphs with additional random edges. This was necessary as BRITE only enables to set an integer number of links per node, and therefore does not support a ``probability-based'' link insertion. The alternative, to set the number of neighbours to values greater or equal two, results in already very well connected networks inhibiting a fine-granular investigation of the influence of redundant paths on the forwarding strategies.

Based on the aforementioned considerations, we generate nine different topology variants. These topology variants differ in graph connectivity and available bandwidth resources. Table~\ref{tab:connectivity} specifies the configured graph connectivities, where the second and third columns specify the number of additionally added edges at the top and at the bottom level, respectively. We define the connectivity $C(\mathcal{N})$ of a network $\mathcal{N}$ as $C(\mathcal{N}) = \textstyle\frac{1}{(|\mathcal{V}|-1)\cdot|\mathcal{V}|} \cdot \textstyle\sum_{v \in \mathcal{V}}deg(v)$, where $deg(v)$ denotes the (edge) degree of node $v$.
Figure~\ref{fig:network_connectivities} illustrates three sample topologies for the chosen connectivity values. Table~\ref{tab:bandwidth} lists the bandwidth resources provided to the links. The capacity of each link is randomly drawn from a uniform distribution limited by the indicated intervals. The nine topology variants arise from the cross product of Table~\ref{tab:connectivity} and Table~\ref{tab:bandwidth}. 
For instance, if we refer to variant \textit{LowConLowBW}, this refers to a topology defined by settings given in the first two lines of Table~\ref{tab:connectivity} and Table~\ref{tab:bandwidth}.
%
\subsection{\revised{Scenario Description}}
\label{sec:scenario_desc}

We evaluate the mentioned algorithms under two different request scenarios. First, the popularity of content is uniformly distributed among all clients. This leads to the fact that most of the content is concentrated at the caches in the core of the network~\cite{Wang:2013}. Second, we distribute the popularity of content according to a Zipf distribution which has the opposite effect (caches at the edges are utilized)~\cite{Wang:2013}. \revised{Before simulations are conducted, we perform a limited parameter investigation in Subsection~\ref{sec:saf_period} regarding the duration of the period (the time between two iterations of Algorithm~\ref{alg:updateOperations}) for SAF, as Theorem~\ref{thm:convergence} suggests that the duration of the period has significant influence on SAF's performance.}

\begin{table}[t]
\small
\centering
\begin{tabular}{|l|c|c|c|}
\hline
\textsc{identifier} &\textsc{top level }&\textsc{bottom level}& \textsc{$C(\mathcal{N})$}\\
\hline
\hline
LowCon& $\lfloor \mu / 2 \rfloor$ & $\lfloor\nu / 3 \rfloor$& 0.0265 \\\hline
MediumCon&  $\mu$ & $\lfloor \nu / 2 \rfloor$& 0.0311 \\\hline
HighCon& $\lfloor \mu * 2 \rfloor$ & $\nu$& 0.0422\\\hline
\end{tabular}
\caption{Additional edges per connectivity variant. }
\label{tab:connectivity}
\vspace{-0.1cm}
\end{table}
%
\begin{table}[t]
\small
\centering
\begin{tabular}{|l|c|c|}
\hline
\textsc{identifier} &\textsc{top level }&\textsc{bottom level}\\
\hline
\hline
LowBW& [2, 4] Mbps & [1, 2] Mbps \\\hline
MediumBW&  [3, 5] Mbps & [2, 4] Mbps \\\hline
HighBW& [4, 6] Mbps & [3, 5] Mbps \\\hline
\end{tabular}
\caption{Assigned link capacities.}
\label{tab:bandwidth}
\vspace{-0.1cm}
\end{table}

Given a generated topology we randomly placed $a=100$ clients and $b=10$ servers in the network. Clients are configured to request content from a single server with a rate of 30 Interests per second, uniformly distributed during a second.
This corresponds to a download rate of approximately 1 Mbps, as a single Interest always requests a 4~kB Data packet. The uniform distributed request rate represents a steady consumption of content, e.g., streaming video at a certain bit-rate. As already mentioned, for the first scenario we assume a perfect uniform content popularity. 
A client randomly starts to consume content within the first 30 seconds of the simulation. Each server provides unique content identified by an arbitrary prefix, e.g., \textit{/server\_id/}. Each node in the network is equipped with a 25~MB large cache which corresponds to approximately 1\% of the available content catalogue. We consider this size as sufficient and realistic for the following reason. For instance, consider the MediumConMediumBW-scenario. A node in this scenario maintains on average 3.11~links (cf. Table~\ref{tab:connectivity}). The maximum link capacity for this scenario is limited by 5~Mbps (cf. Table~\ref{tab:bandwidth}), which leads to an average traffic of $3.11 \cdot 5~Mbps = 15.55~Mbps$. So, given a node that fully utilizes all links may still cache more than 10~seconds of traffic on average. Considering the findings in~\cite{Anand:2009} that 40\% of all cache hits occur in the first 10~seconds, and that clients in the selected scenarios start to request content within a 30 second window, we consider the cache to be sufficiently large. We use the Least Recently Used (LRU) cache replacement strategy~\cite{Wang:2013} and nodes cache every packet they receive.

We further introduce 0, 50, or 100 random link failures during each simulation run. A link failure's point of occurrence and its duration are distributed uniformly. The duration of a link failure is drawn from the interval of $[0,\lfloor\frac{SimTime}{10}\rfloor]$, where $SimTime$ denotes the duration of a simulation run. Nodes are configured such that they know all possible routes to any content server at the beginning of each simulation run. This is necessary as many of SAF's competitors require that this information is provided by the routing layer. In this experiment, SAF uses the routing information only as starting point for the initial FWT. During the simulation no routing updates are enforced, which puts the responsibility to deal with short-term topology changes to the forwarding strategies. We simulate 1800~seconds of network traffic using the aforementioned parameters \revised{and conduct 50~runs per setting}.

\subsection{\revised{SAF: Influence of the Period on the Performance}}
\label{sec:saf_period}
\revised{In order to investigate the influence of different values for the duration of a period ($\tau$), we considered the scenario MediumConMediumBW with 50~link failures given uniform content popularity. Figure~\ref{fig:saf_period} depicts the 95\% confidence intervals of the average Interest satisfaction ratio, the cache hit ratio and the hop count for $\tau \in \{0.1,0.5,1.0,2.5,5.0,10.0,50.0,100.0\}$ seconds. The Interest satisfaction ratio denotes the ratio between received Data packets and generated Interests by all clients. The cache hit ratio is averaged over all network nodes (clients and servers maintain no cache). 
The hop count provides the number of links a Data packet traversed to satisfy an issued Interest by a client considering cache hits.}

\revised{One can observe in Figure~\ref{fig:saf_period} that in general shorter periods increase the performance regarding the Interest satisfaction ratio. This is consistent with Theorem~\ref{thm:convergence}, which provides the number of steps for SAF to converge to the optimum (shorter periods, faster steps). However, one can also see that a too small period, e.g. $\tau=0.1s$, may have a negative influence on the performance. This is due to the fact that the observed traffic within extremely short periods is not representative enough to deduce a suitable decision. The cache hit ratio varies only slightly considering different period durations. Nevertheless, one can observe that the highest ratios are achieved with period durations of few seconds. The average hop count does not change significantly, regardless of the selected period.
We use $\tau = 1.0~seconds$ for all further evaluations, as this setting achieves a good performance with acceptable computational overhead (cf. Section~\ref{sec:forwarding_strategy}).} 

\begin{figure}[tb!]
\centering
\includegraphics[scale=0.27]{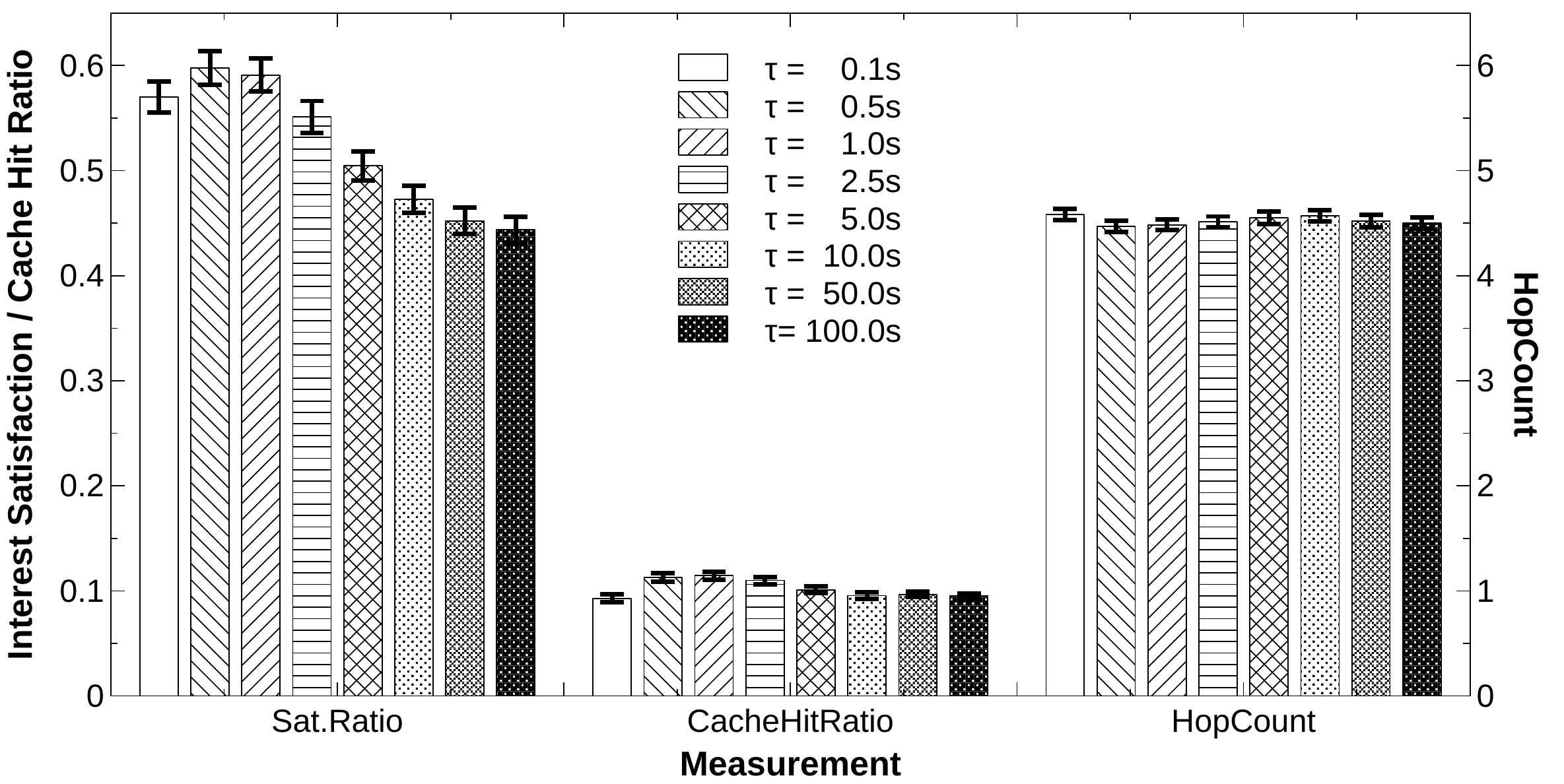}
\caption{\revised{Influence of different period durations ($\tau$) on SAF, given scenario \textit{MediumBWMediumCon} with 50~link failures}.}
\label{fig:saf_period}
\end{figure}

\subsection{\revised{Performance under Uniform Content Popularity}}
\label{sec:uniform_request_eval}
Figures~\ref{fig:forwarding_results_0LinkFailures_ndnsim2.0},~\ref{fig:forwarding_results_50LinkFailures_ndnsim2.0}, and~\ref{fig:forwarding_results_100LinkFailures_ndnsim2.0} depict the 95\% CI of the Interest satisfaction ratio for each pairing (topology variant, forwarding strategy) having 0, 50, or 100 random link failures during each simulation. It is evident that SAF outperforms its competitors in nearly all simulation scenarios in terms of Interests satisfied, regardless of the number of link failures, the connectivity of the underlying network and/or network topology. \revised{The two strongest competitors to SAF are iNRR and OMP-IF, especially in scenarios with more resources available. SAF's good performance considering all scenario settings can be explained by its ability to:
\begin{inparaenum}[\itshape i\upshape)]
\item smartly circumvent congested nodes and link failures by taking detours into account;
\item prevent further transmission of Interests on congested links by redirecting them to the virtual dropping face;
\item discover paths to cached content that are not indicated by the FIB.
\end{inparaenum}}
This is also reflected by Figure~\ref{fig:forwarding_results_0LinkFailures_ndnsim2.0_cache} depicting the cache hit ratio for the scenarios with zero link failures. In most of the cases SAF outperforms the other algorithms\revised{, with the exception of iNRR, which has perfect knowledge about the content chunks in the individual caches}. Figure~\ref{fig:hop_count_0LinkFailures_ndnsim2.0} illustrates the average hop count \revised{per satisfied Interest} for the scenarios with zero link failures.
\revised{As expected, SAF maintains a higher hop count than some of the other algorithms due to the detours it takes for maximizing the Interest satisfaction according to $\mathcal{M}^T$. The very low hop count of Broadcast and NCC show that both algorithms are able to obtain nearby replicas, however, due to their poor Interest satisfaction and cache hit ratio the low hop count is valueless.} 

\begin{table}[t]
\centering
\small
\begin{tabular}{|l|c|c|c|}
\hline
\textsc{Strategy} &\textsc{LowBW}&\textsc{MediumBW}& \textsc{HighBW}\\
\hline
\hline
Broadcast& 0.024 $\vert$ 0.133 & 0.097 $\vert$ 0.242 & 0.195 $\vert$ 0.333 \\\hline
NCC& 0.033 $\vert$ 0.155 & 0.102 $\vert$ 0.229 & 0.181 $\vert$ 0.287 \\\hline
ShortestRoute& 0.012 $\vert$ 0.098 & 0.065 $\vert$ 0.188 & 0.135 $\vert$ 0.247 \\\hline
RFA& 0.060 $\vert$ 0.218 & 0.109 $\vert$ 0.267 & 0.111 $\vert$ 0.263 \\\hline
iNRR& 0.016 $\vert$ 0.107 & 0.112 $\vert$ 0.226 & 0.204 $\vert$ 0.263 \\\hline
\revised{OMP-IF} & 0.045 $\vert$ 0.177 & 0.135 $\vert$ 0.233 & 0.252 $\vert$ 0.298 \\\hline
SAF& 0.090 $\vert$ 0.240 & 0.253 $\vert$ 0.334 & 0.288 $\vert$ 0.320 \\\hline
\end{tabular}
\caption{\textbf{Actual} (first value) and \textbf{relative} (second value) \textbf{performance loss} comparing Figure~\ref{uniform_subfig-2:00lf} and~\ref{uniform_subfig-2:100lf}.}
\label{tab:linkfailure_comp}
\end{table}

\begin{figure*}[tbh!]
\subfloat[Results for LowCon\label{uniform_subfig-1:0lf}]{%
\includegraphics[width=0.29\textwidth]{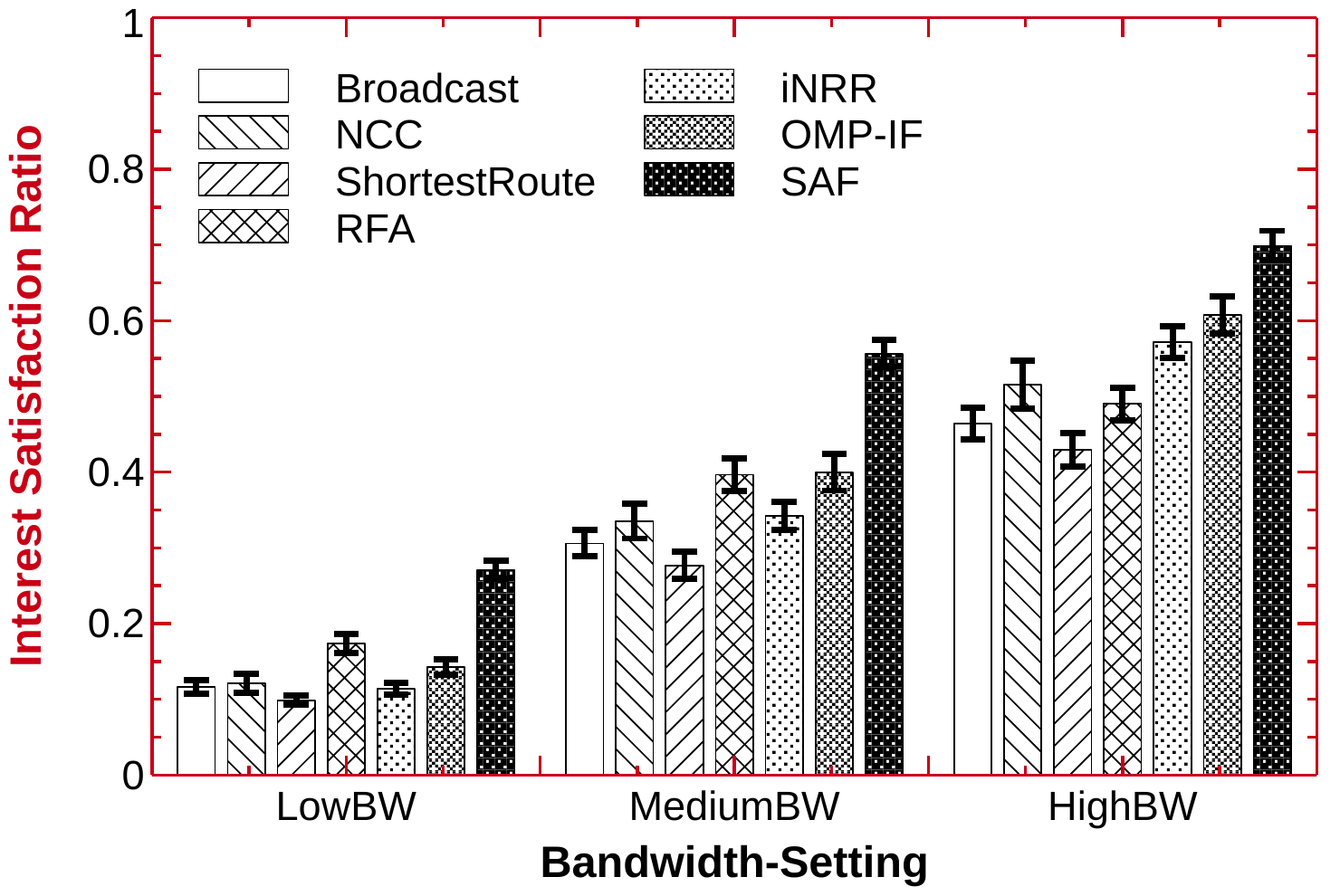}
}
\hfill
\subfloat[Results for MediumCon\label{uniform_subfig-2:00lf}]{%
\includegraphics[width=0.29\textwidth]{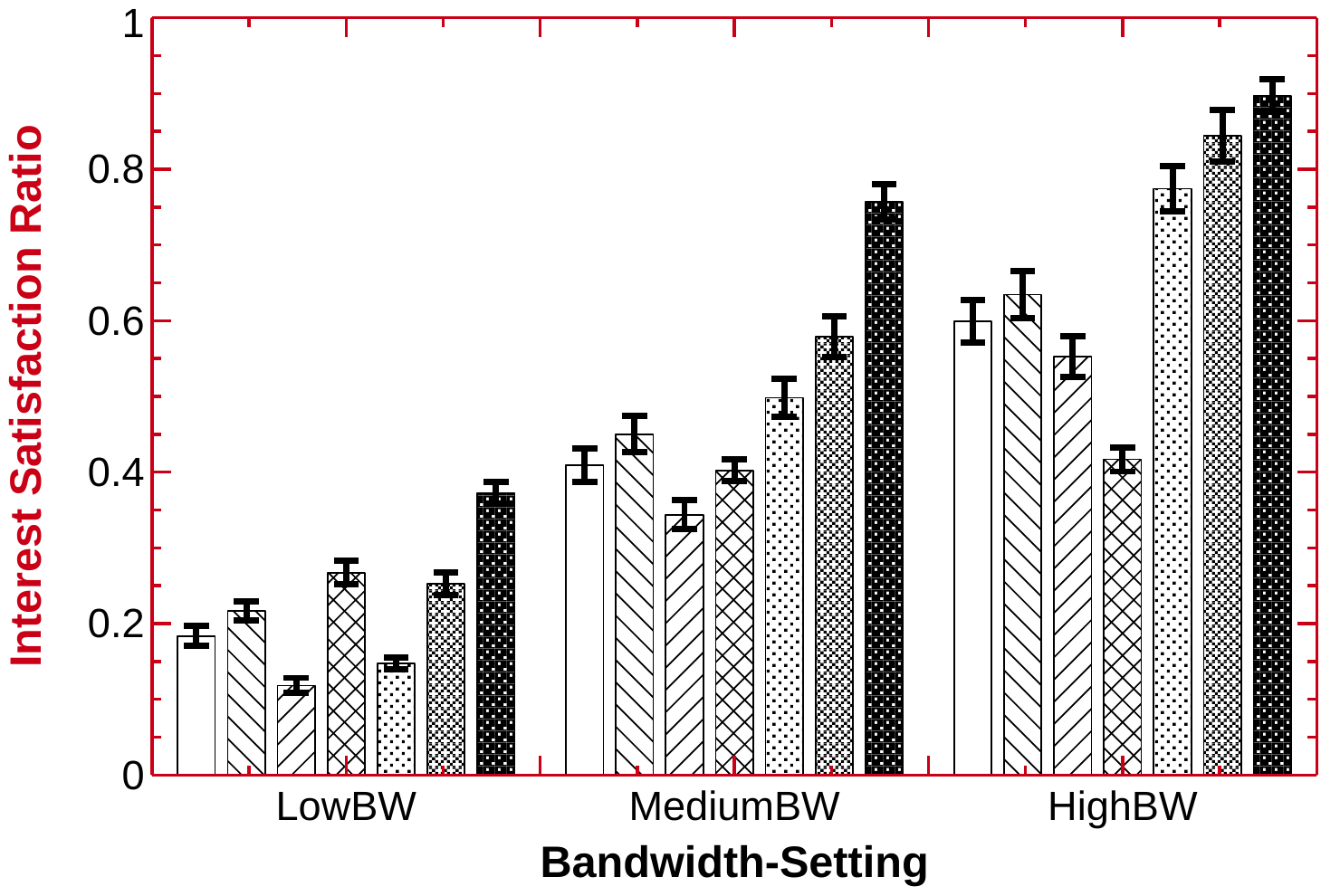}
}
\hfill
\subfloat[Results for HighCon\label{uniform_subfig-3:00lf}]{%
\includegraphics[width=0.29\textwidth]{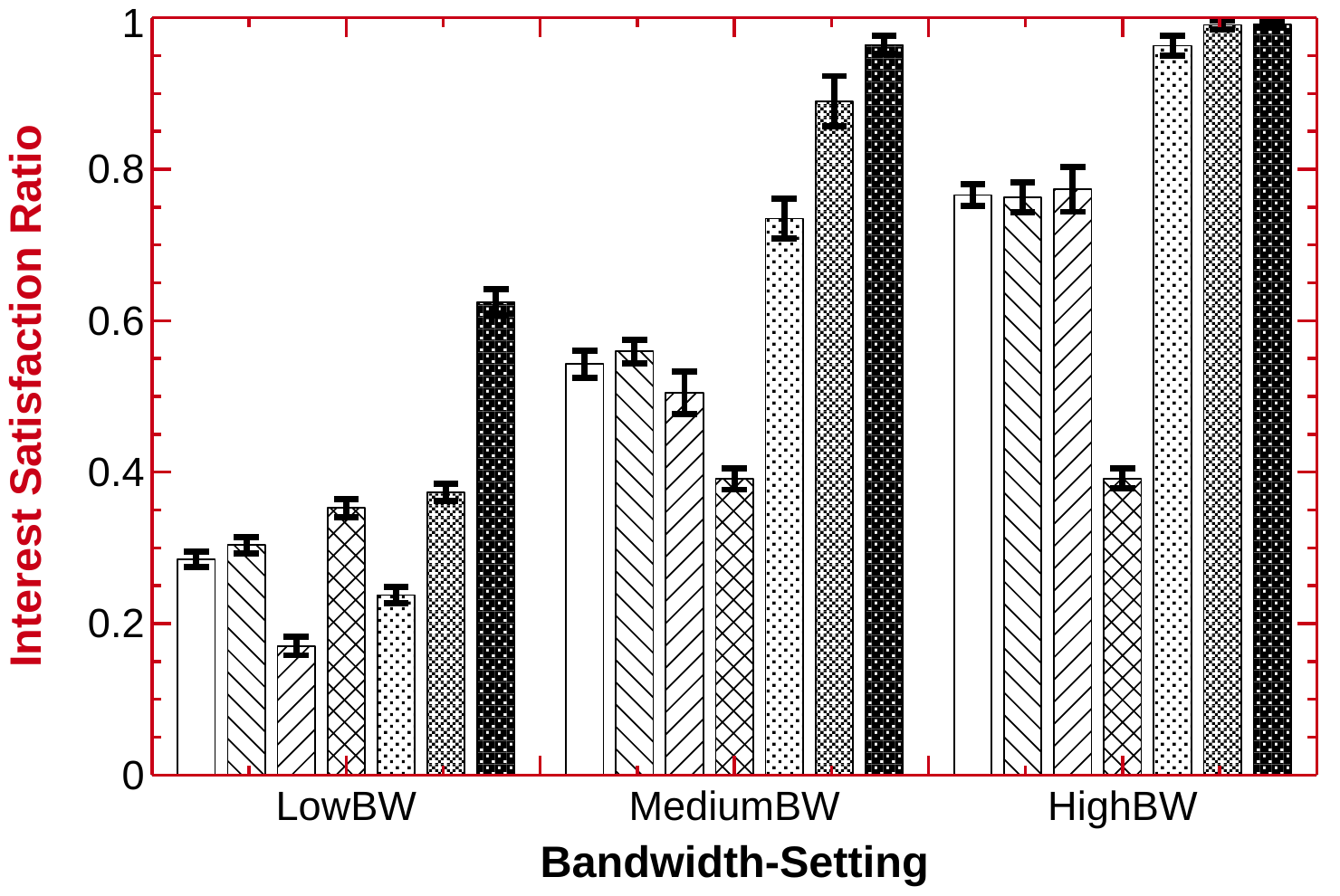}
}
\vspace{-0.1cm}
\caption{Average \textcolor{harvardcrimson}{\textbf{Interest satisfaction ratio}} and 95\% CI with \textbf{0 link failures} per simulation run (higher is better) [uniform].}
\label{fig:forwarding_results_0LinkFailures_ndnsim2.0}

\subfloat[Results for LowCon\label{uniform_subfig-1:50lf}]{%
\includegraphics[width=0.29\textwidth]{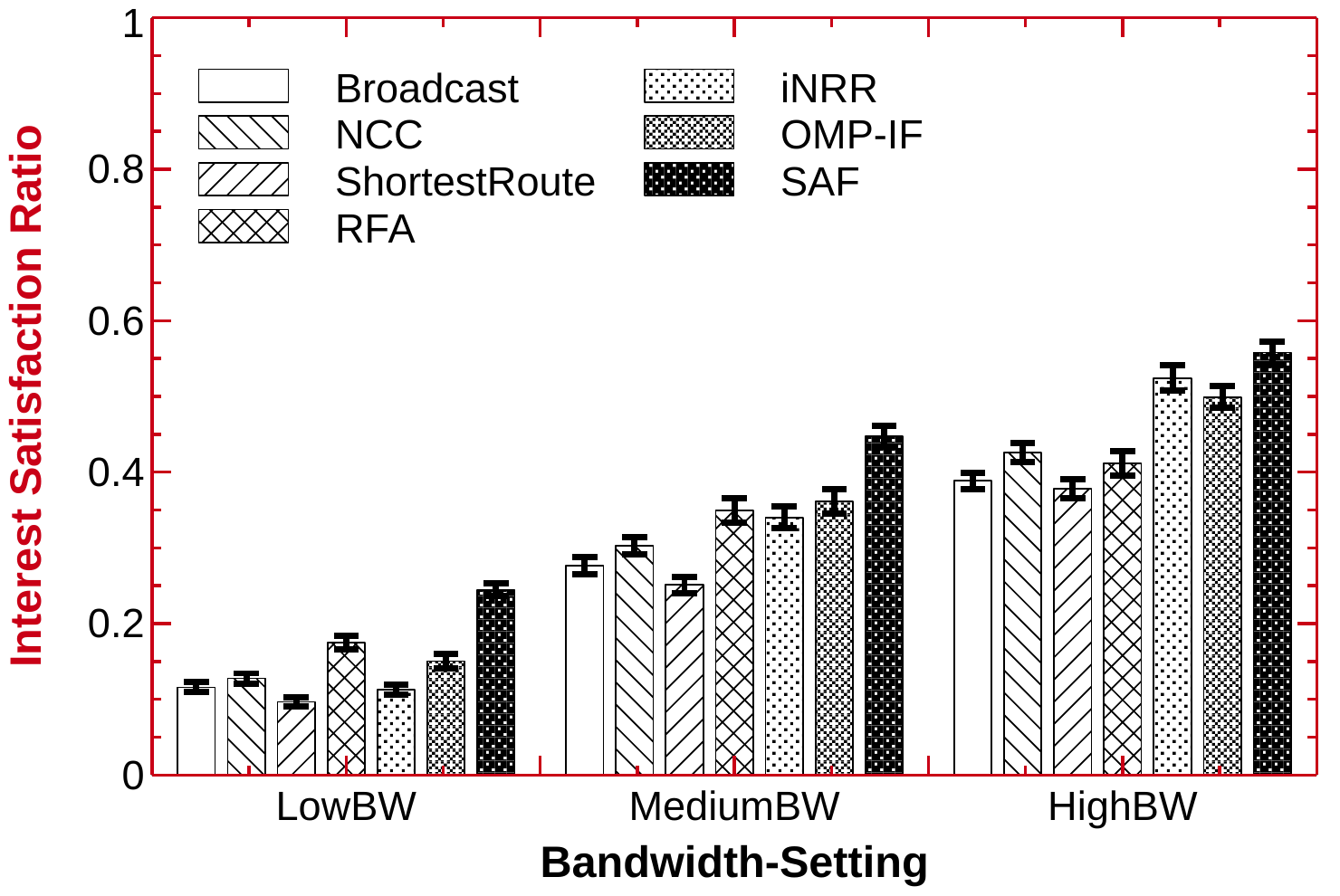}
}
\hfill
\subfloat[Results for MediumCon\label{uniform_subfig-2:50lf}]{%
\includegraphics[width=0.29\textwidth]{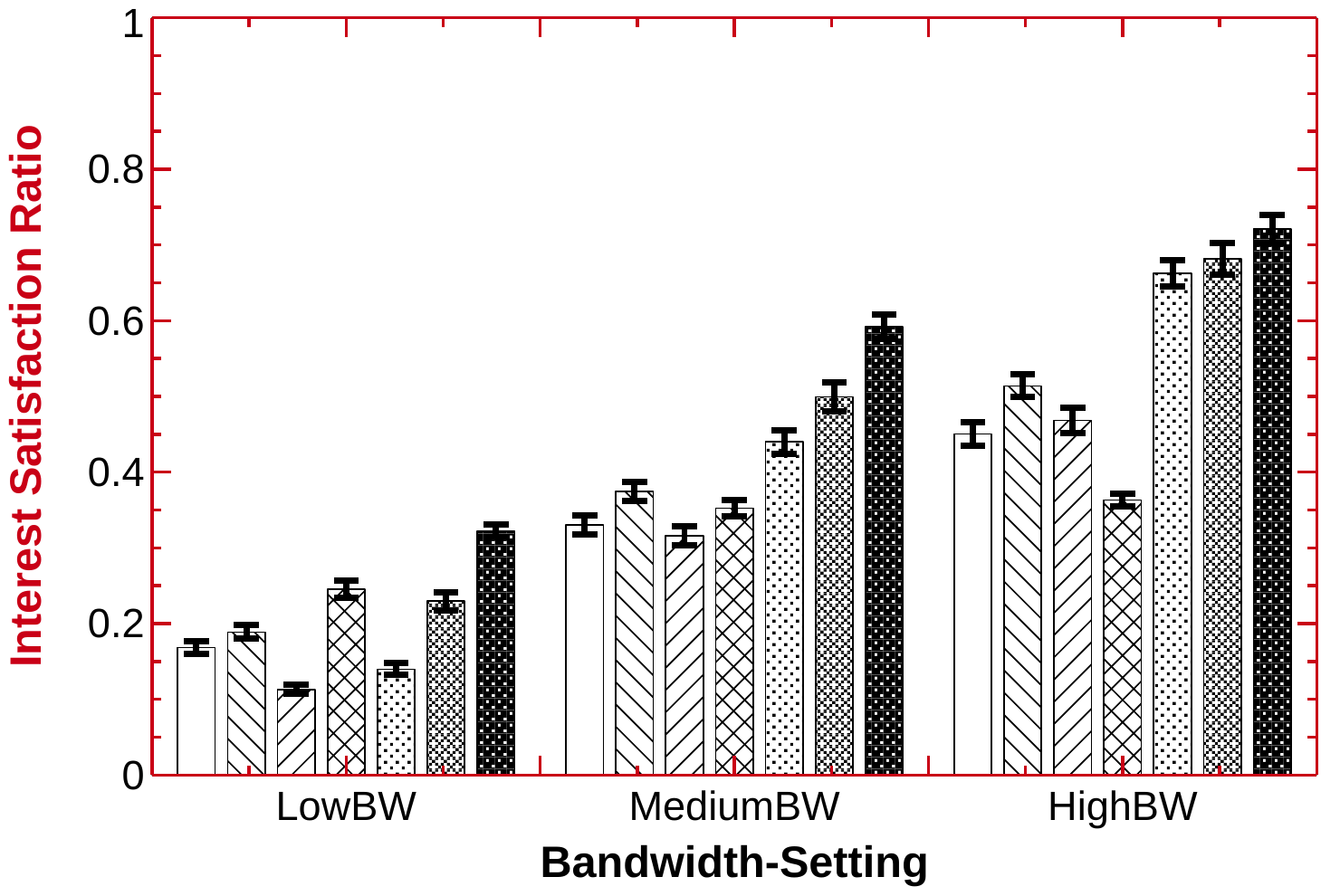}
}
\hfill
\subfloat[Results for HighCon\label{uniform_subfig-3:50lf}]{%
\includegraphics[width=0.29\textwidth]{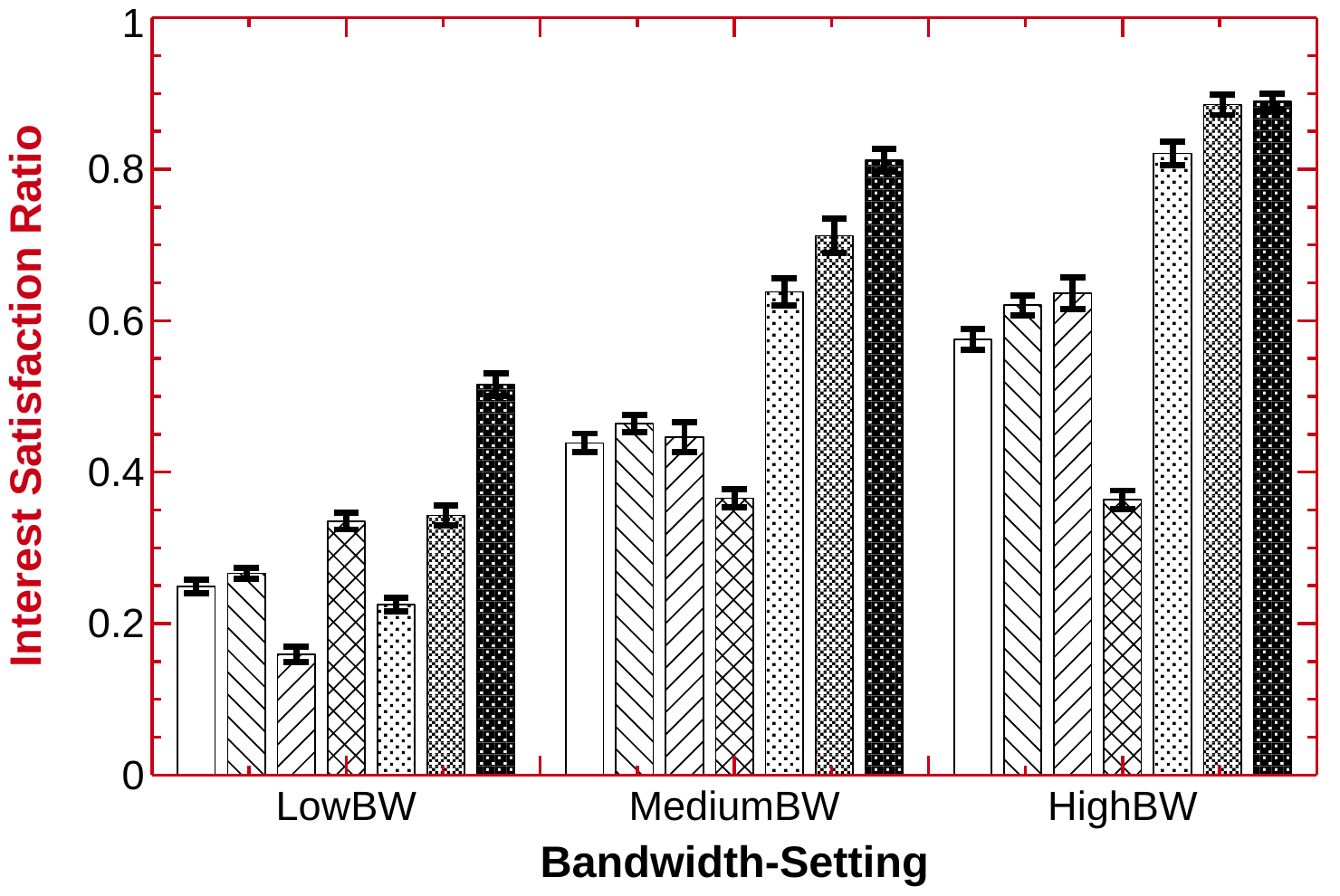}
}
\vspace{-0.1cm}
\caption{Average \textcolor{harvardcrimson}{\textbf{Interest satisfaction ratio}} and 95\% CI with \textbf{50 link failures} per simulation run (higher is better) [uniform].}
\label{fig:forwarding_results_50LinkFailures_ndnsim2.0}

\subfloat[Results for LowCon\label{uniform_subfig-1:100lf}]{%
\includegraphics[width=0.29\textwidth]{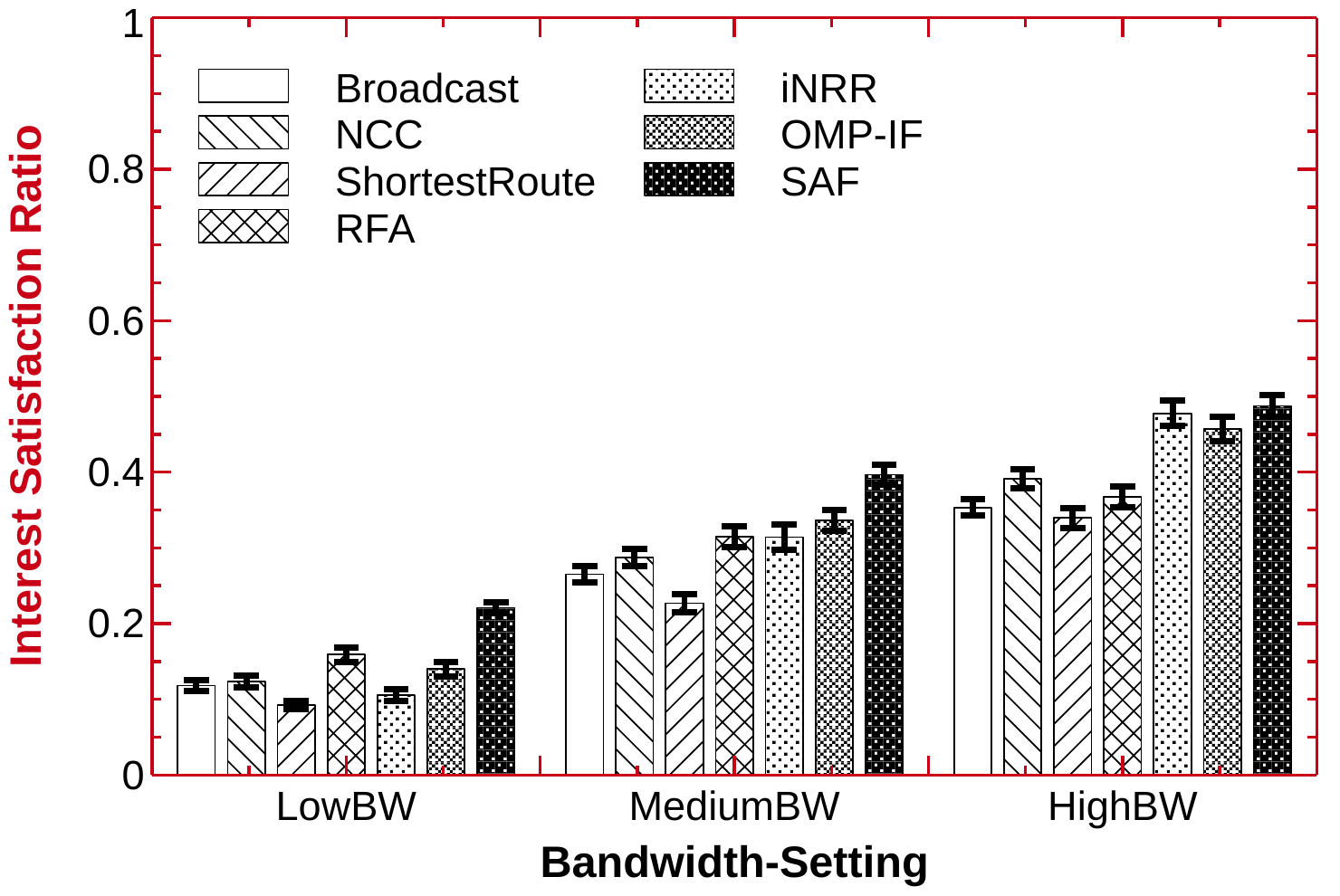}
}
\hfill
\subfloat[Results for MediumCon\label{uniform_subfig-2:100lf}]{%
\includegraphics[width=00.29\textwidth]{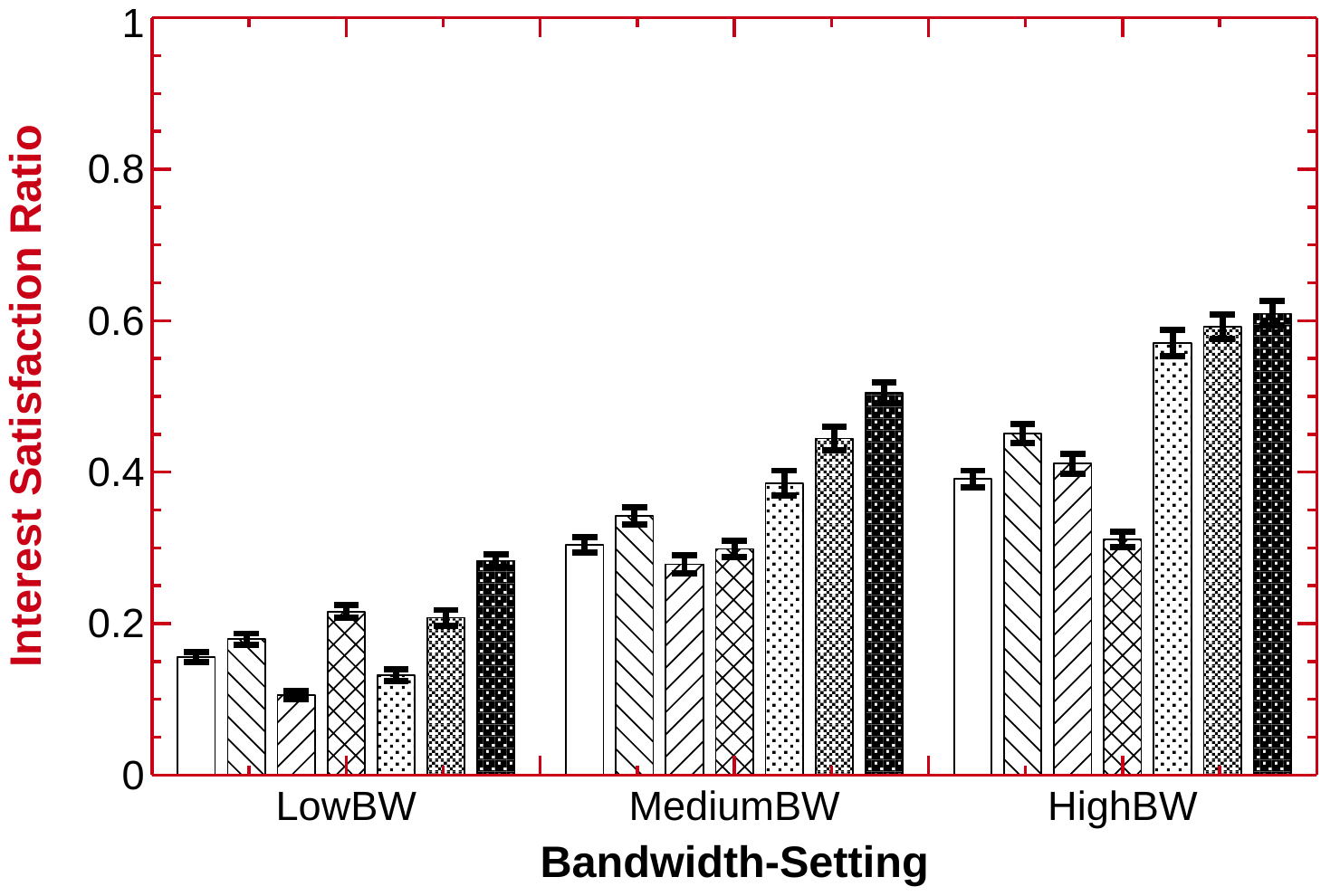}
}
\hfill
\subfloat[Results for HighCon\label{uniform_subfig-3:100lf}]{%
\includegraphics[width=0.29\textwidth]{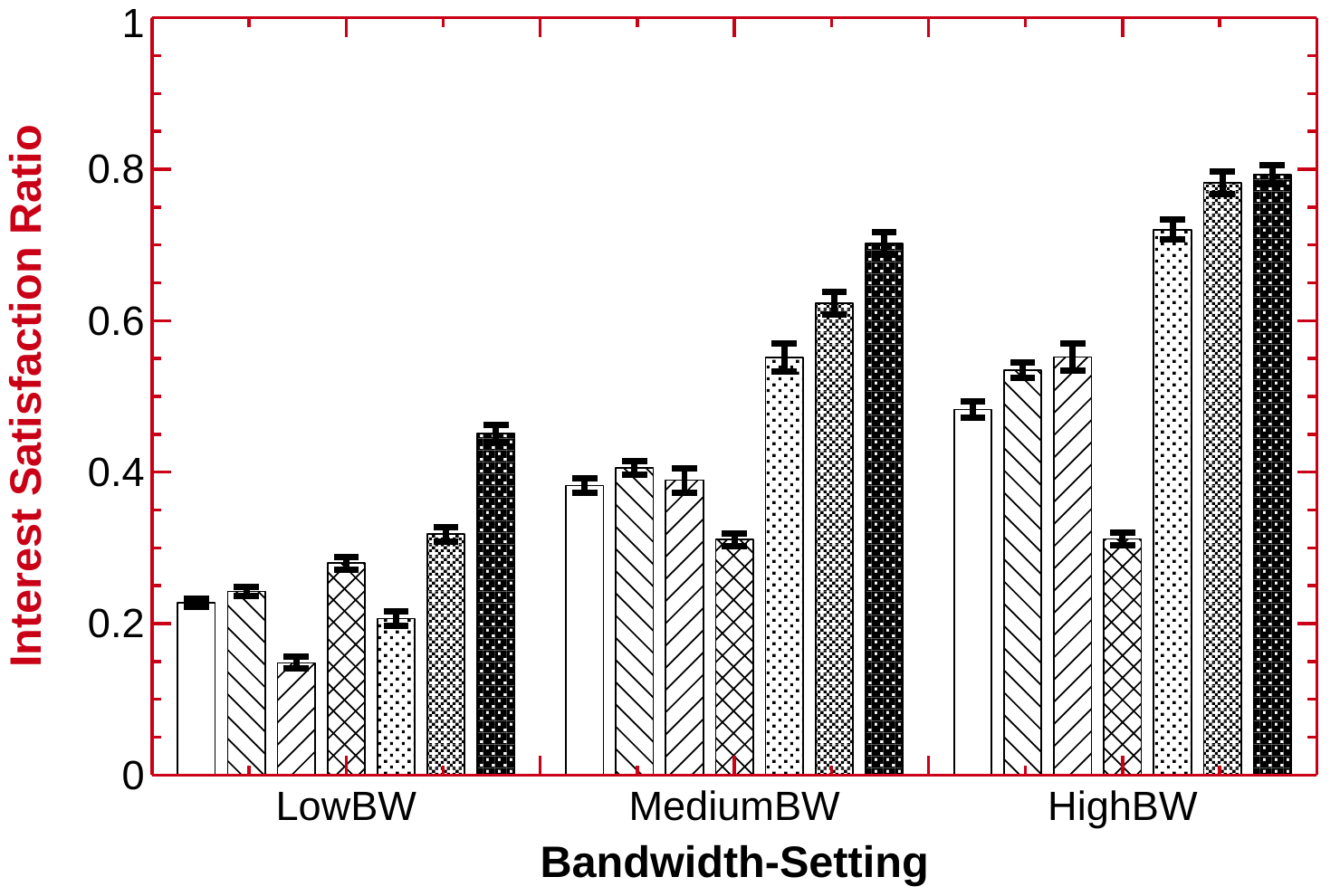}
}
\vspace{-0.1cm}
\caption{Average \textcolor{harvardcrimson}{\textbf{Interest satisfaction ratio}} and 95\% CI with \textbf{100 link failures} per simulation run (higher is better) [uniform].}
\label{fig:forwarding_results_100LinkFailures_ndnsim2.0}
%
\subfloat[Results for LowCon\label{subfig-1:0lf_cache}]{%
\includegraphics[width=0.29\textwidth]{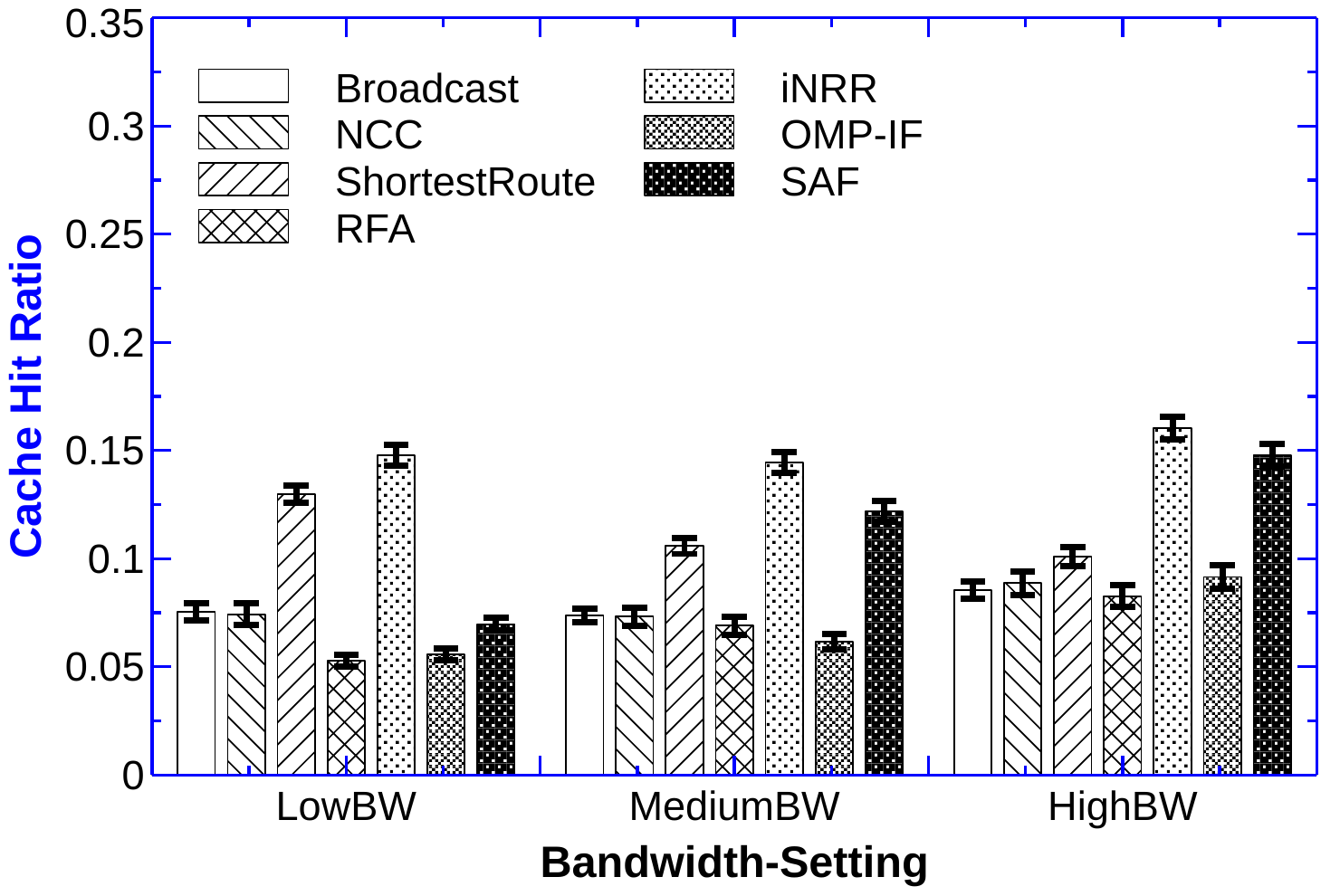}
}
\hfill
\subfloat[Results for MediumCon\label{subfig-2:00lf_cache}]{%
\includegraphics[width=0.29\textwidth]{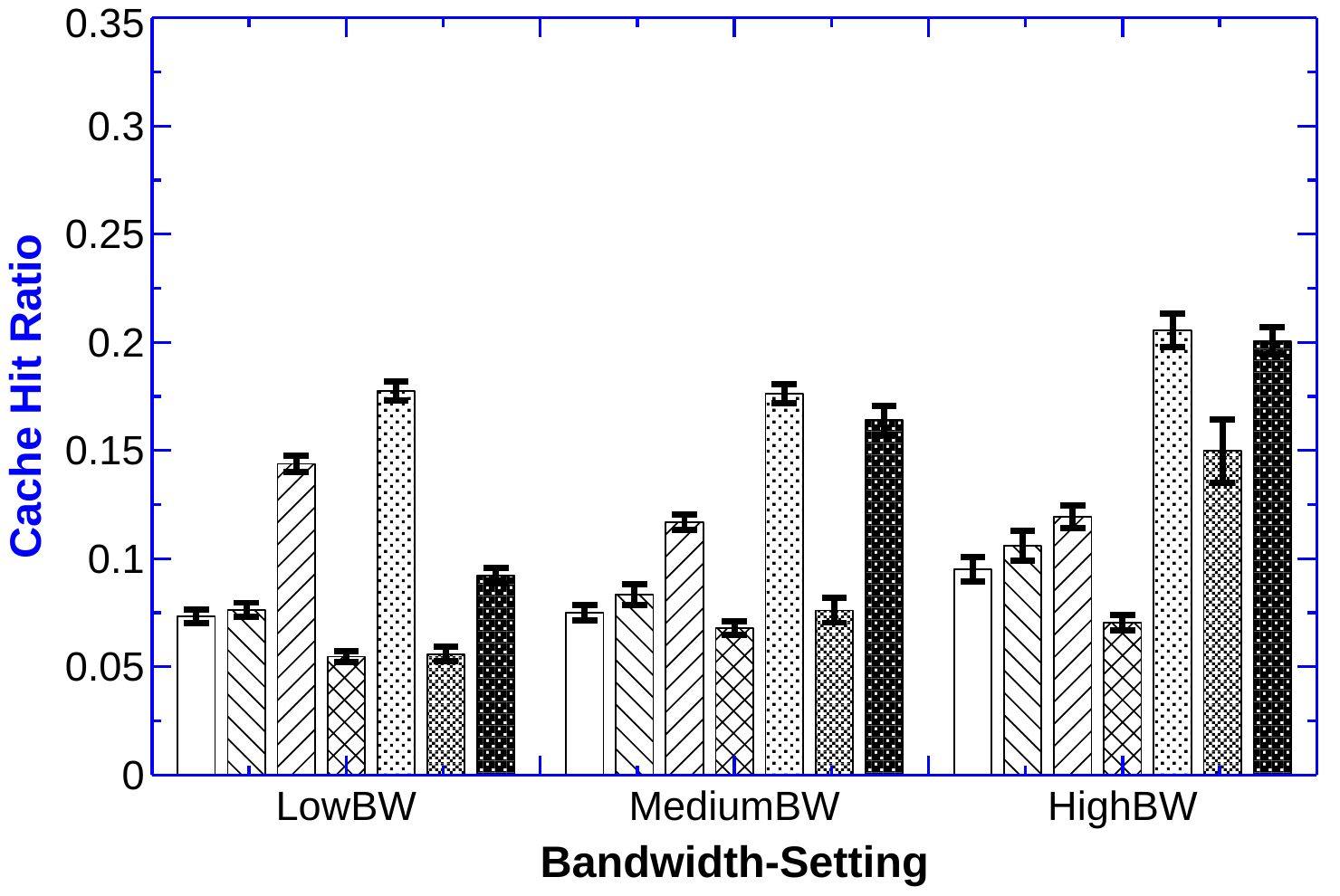}
}
\hfill
\subfloat[Results for HighCon\label{subfig-3:00lf_cache}]{%
\includegraphics[width=0.29\textwidth]{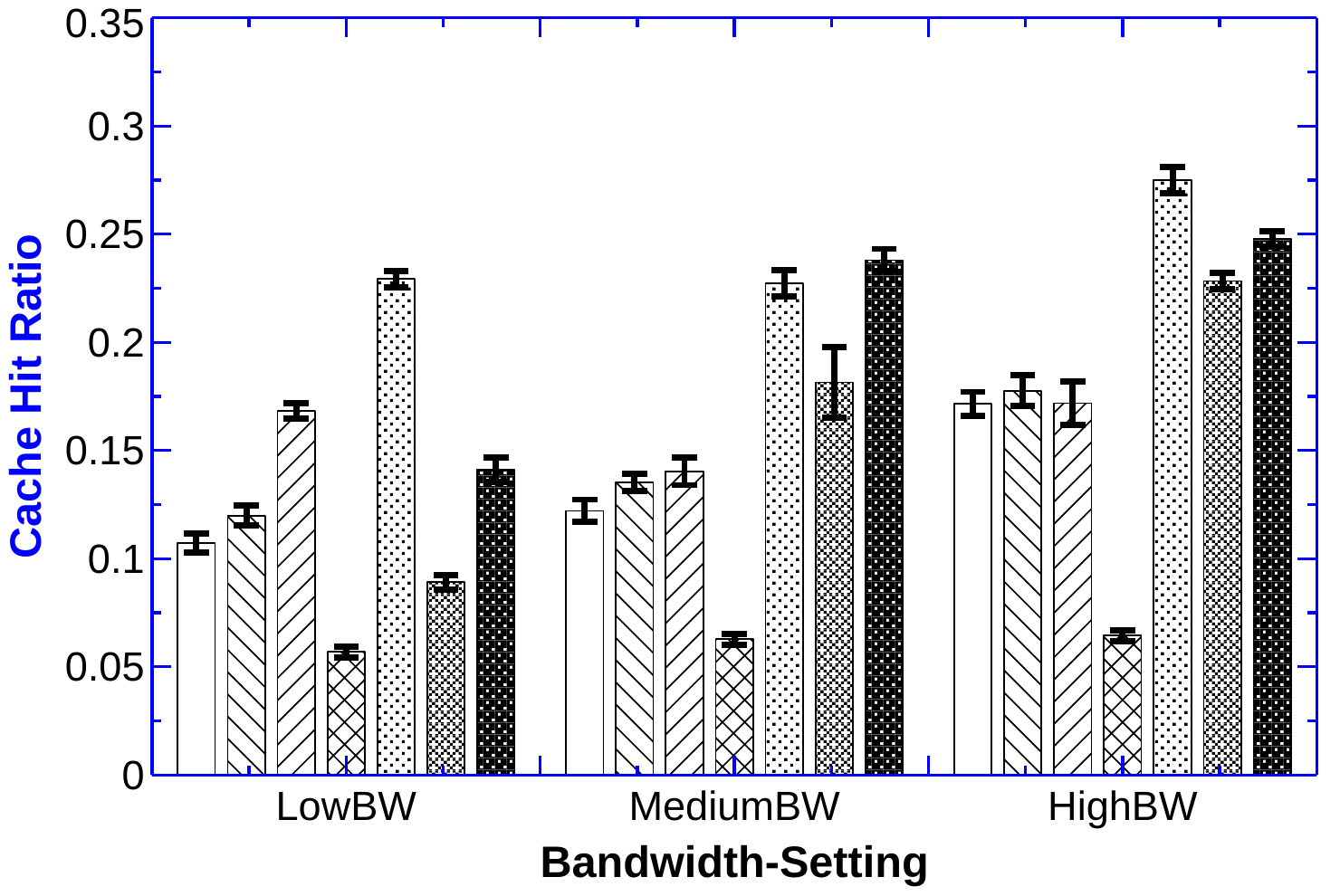}
}
\vspace{-0.1cm}
\caption{Average \textcolor{blue}{\textbf{cache hit ratio}} and 95\% CI in the network with \textbf{0 link failures} per simulation run (higher is better) [uniform].}
\label{fig:forwarding_results_0LinkFailures_ndnsim2.0_cache}

\subfloat[Results for LowCon\label{subfig-1:00lf_hops}]{%
\includegraphics[width=0.29\textwidth]{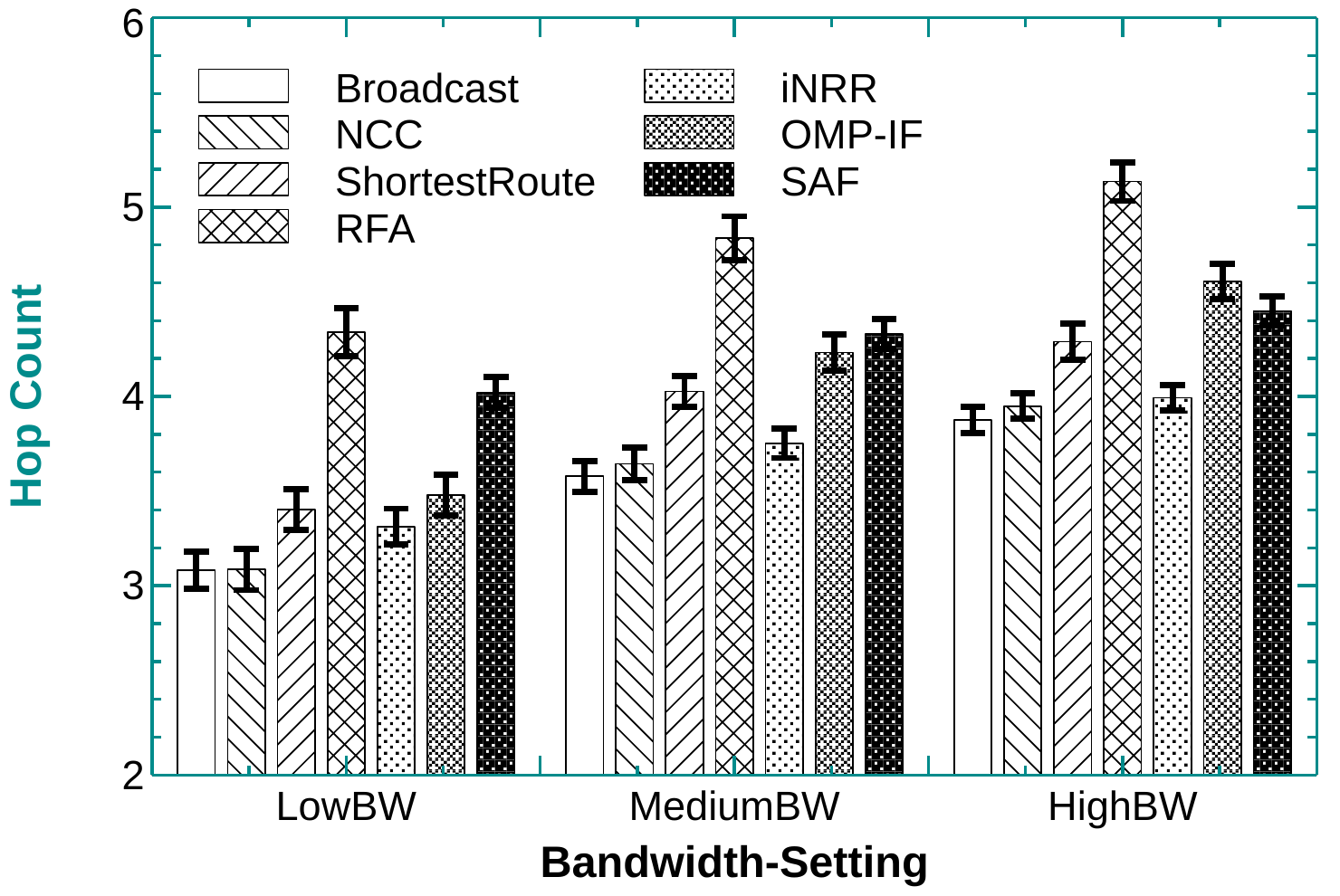}
}
\hfill
\subfloat[Results for MediumCon\label{subfig-2:00lf_hops}]{%
\includegraphics[width=0.29\textwidth]{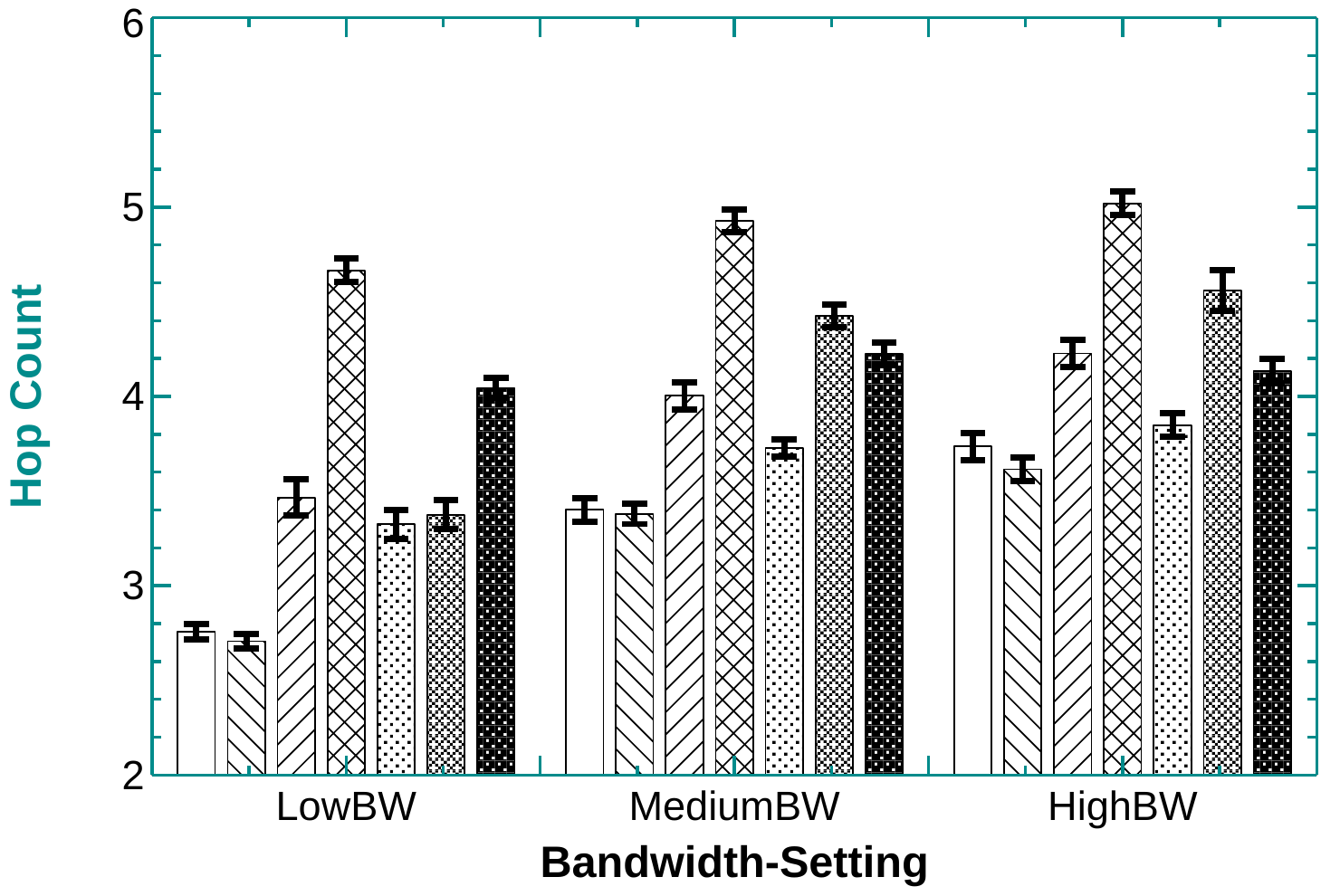}
}
\hfill
\subfloat[Results for HighCon\label{subfig-3:00lf_hops}]{%
\includegraphics[width=0.29\textwidth]{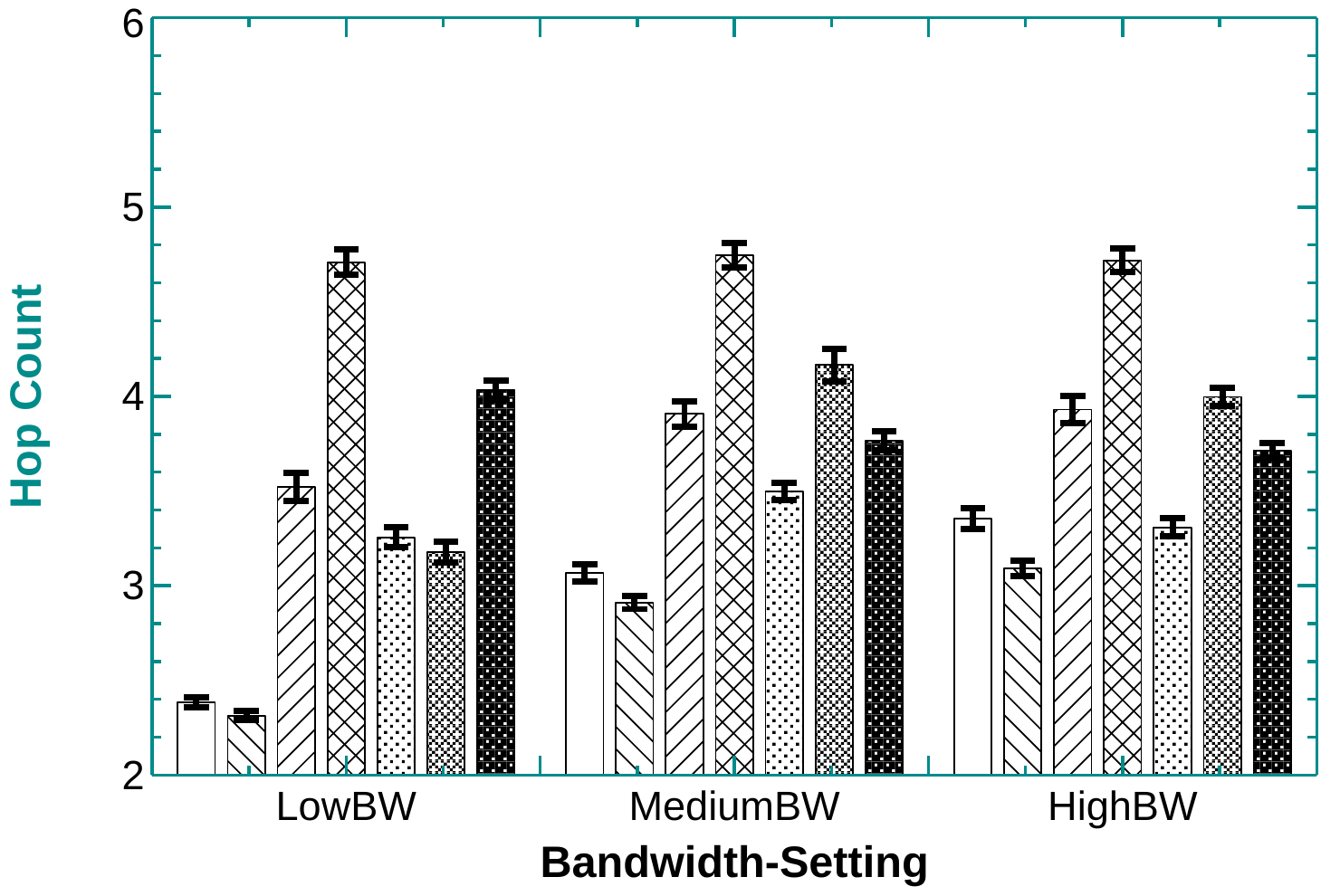}
}
\vspace{-0.1cm}
\caption{Average \textcolor{darkcyan}{\textbf{hop count}} per \emph{satisfied} Interest and 95\% CI with \textbf{0 link failures} per simulation run (lower is better) [uniform].}
\vspace{-0.1cm}
\label{fig:hop_count_0LinkFailures_ndnsim2.0}
\end{figure*}

\begin{figure*}[tbh!]
\subfloat[Results for LowCon\label{zipf_subfig-1:0lf}]{%
\includegraphics[width=0.29\textwidth]{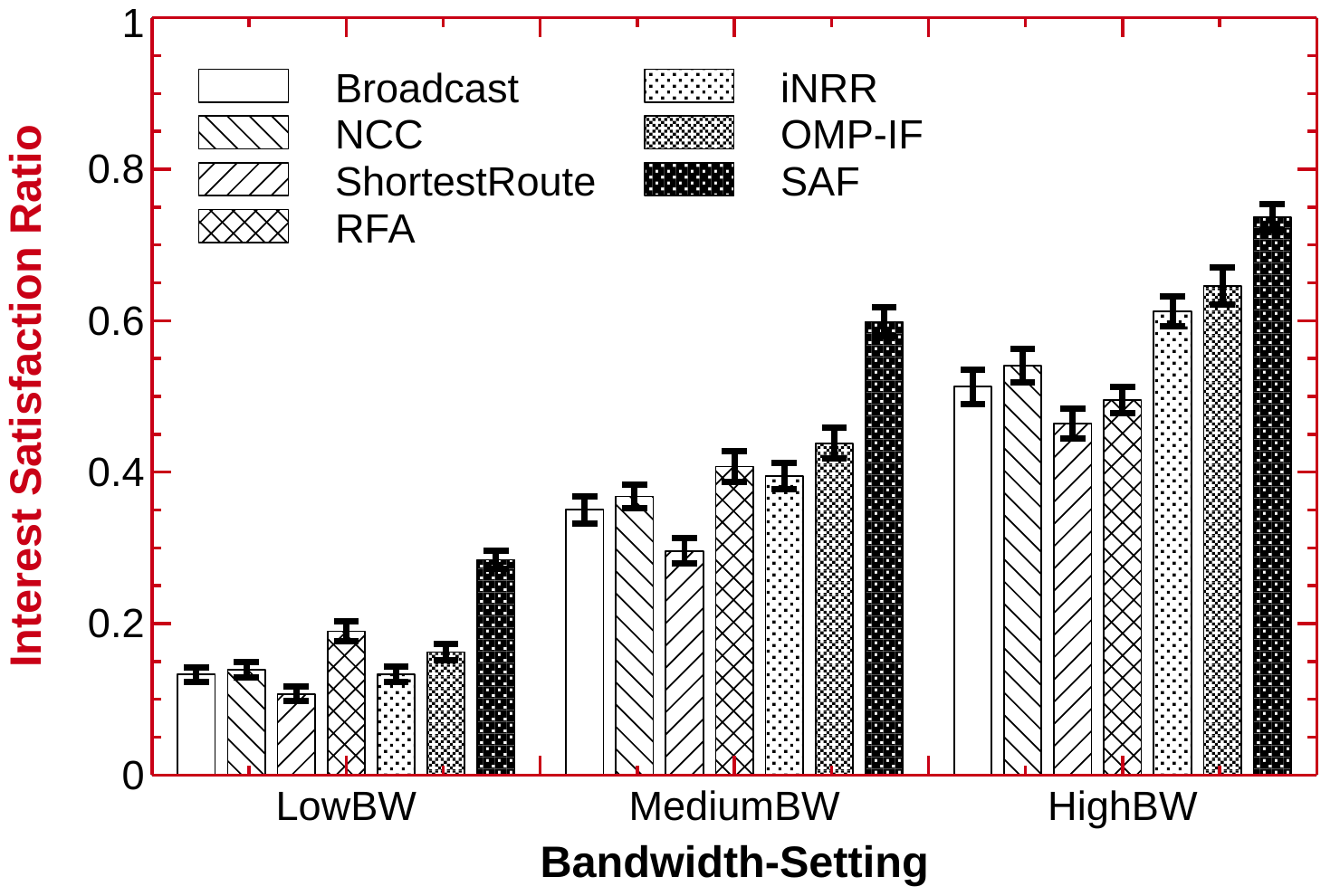}
}
\hfill
\subfloat[Results for MediumCon\label{zipf_subfig-2:00lf}]{%
\includegraphics[width=0.29\textwidth]{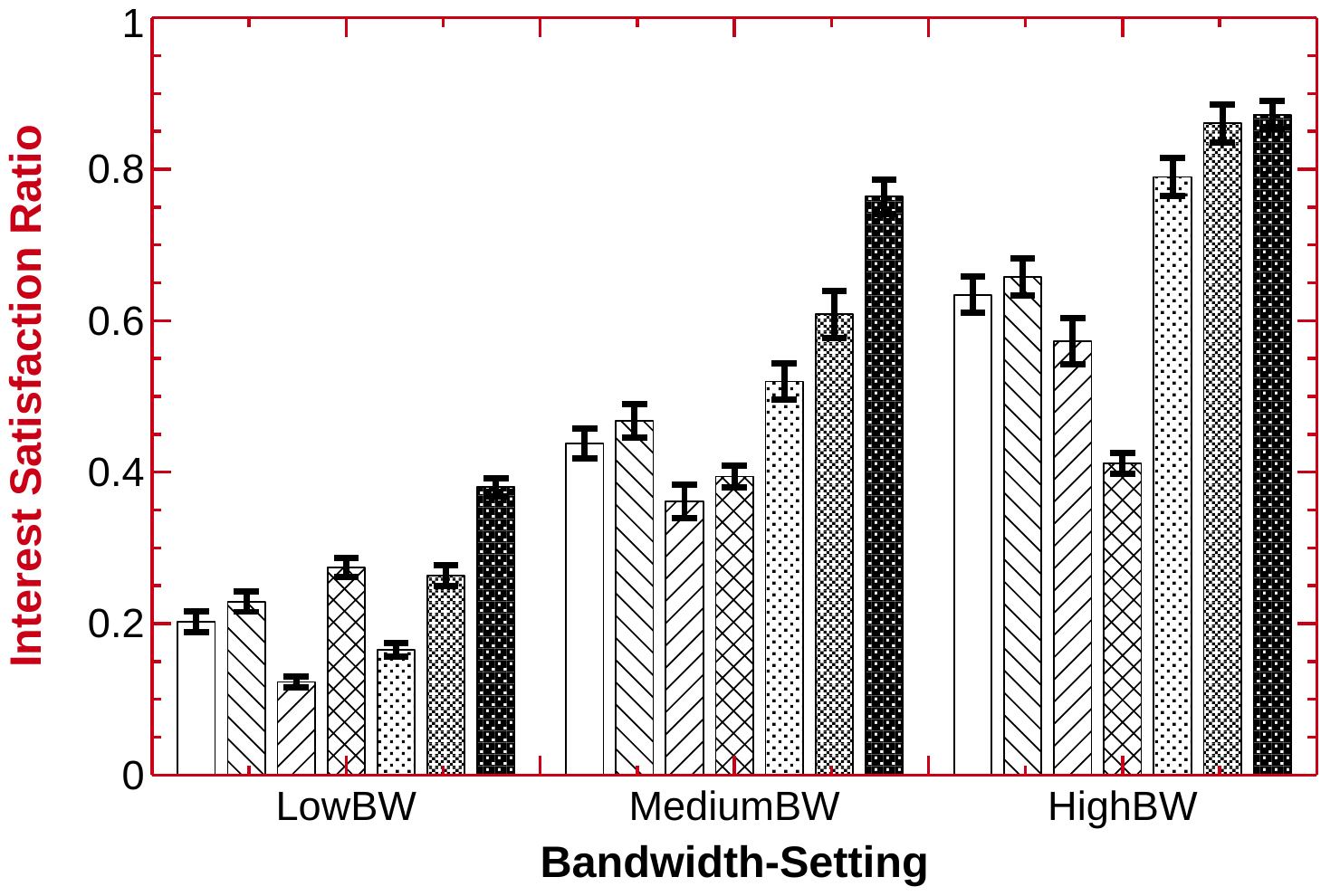}
}
\hfill
\subfloat[Results for HighCon\label{zipf_subfig-3:00lf}]{%
\includegraphics[width=0.29\textwidth]{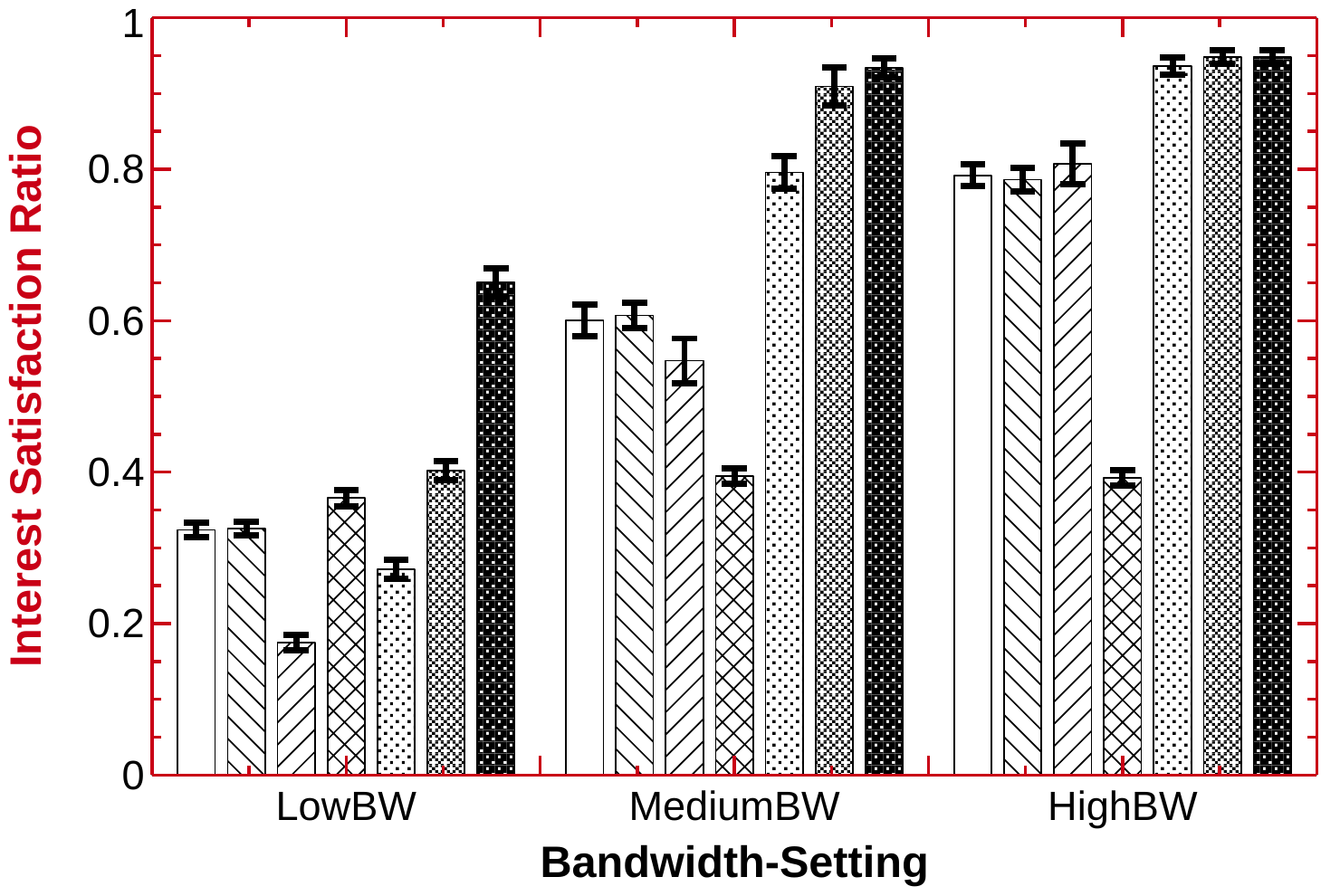}
}
\vspace{-0.1cm}
\caption{Average \textcolor{harvardcrimson}{\textbf{Interest satisfaction ratio}} and 95\% CI with \textbf{0 link failures} per simulation run (higher is better) [zipf].}
\label{fig:zipf_forwarding_results_0LinkFailures_ndnsim2.0}

\subfloat[Results for LowCon\label{zipf_subfig-1:0lf_cache}]{%
\includegraphics[width=0.29\textwidth]{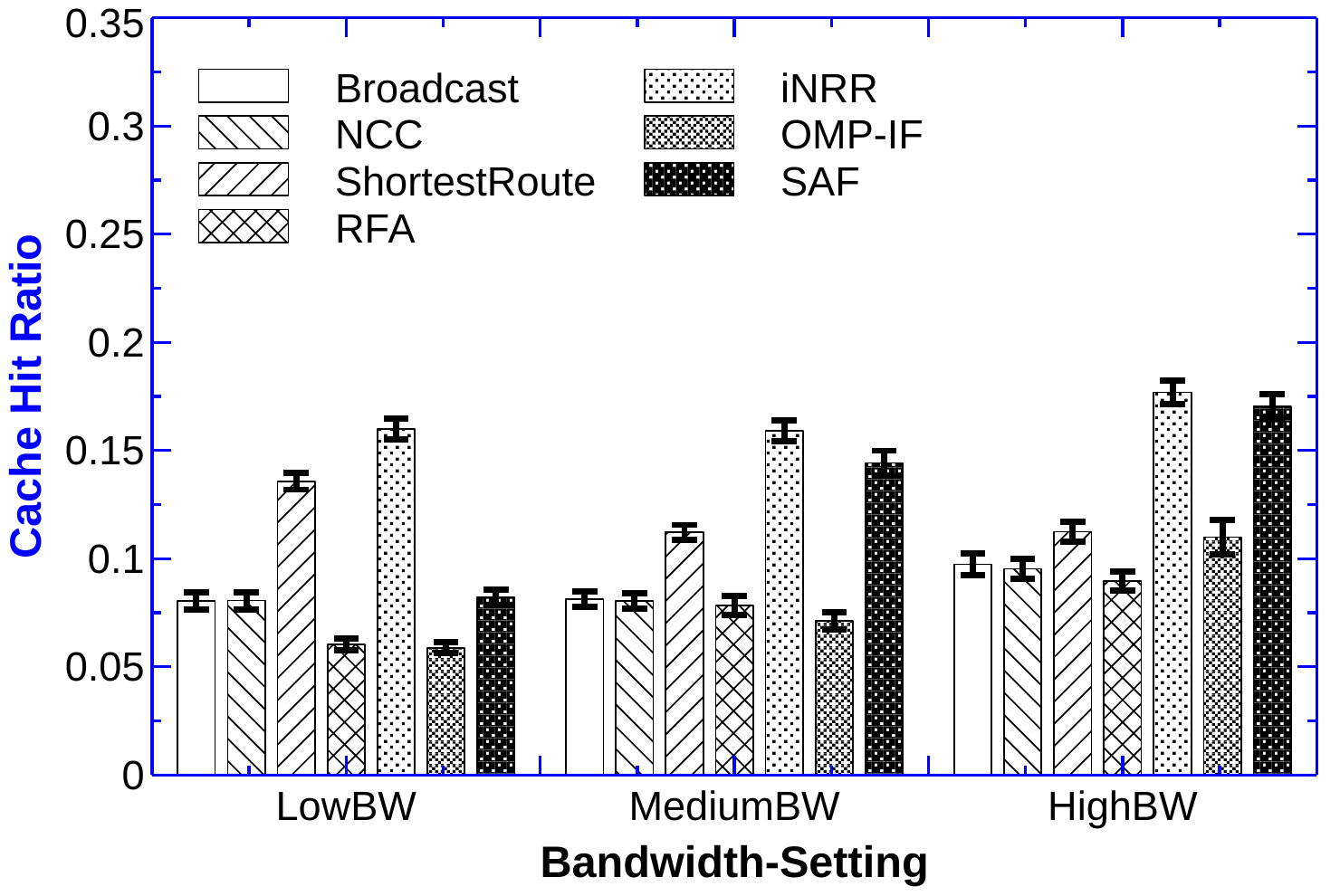}
}
\hfill
\subfloat[Results for MediumCon\label{zipf_subfig-2:00lf_cache}]{%
\includegraphics[width=0.29\textwidth]{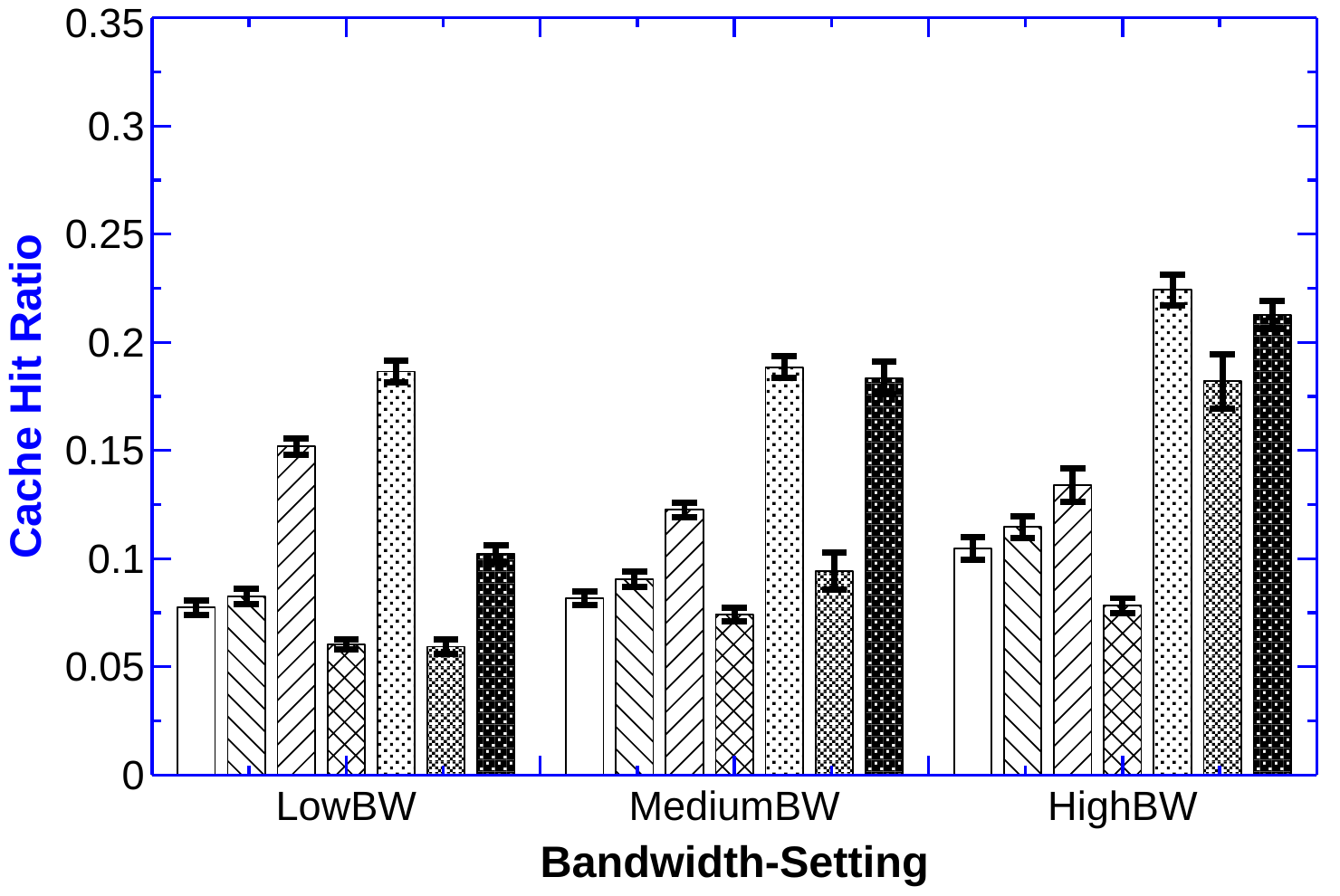}
}
\hfill
\subfloat[Results for HighCon\label{zipf_subfig-3:00lf_cache}]{%
\includegraphics[width=0.29\textwidth]{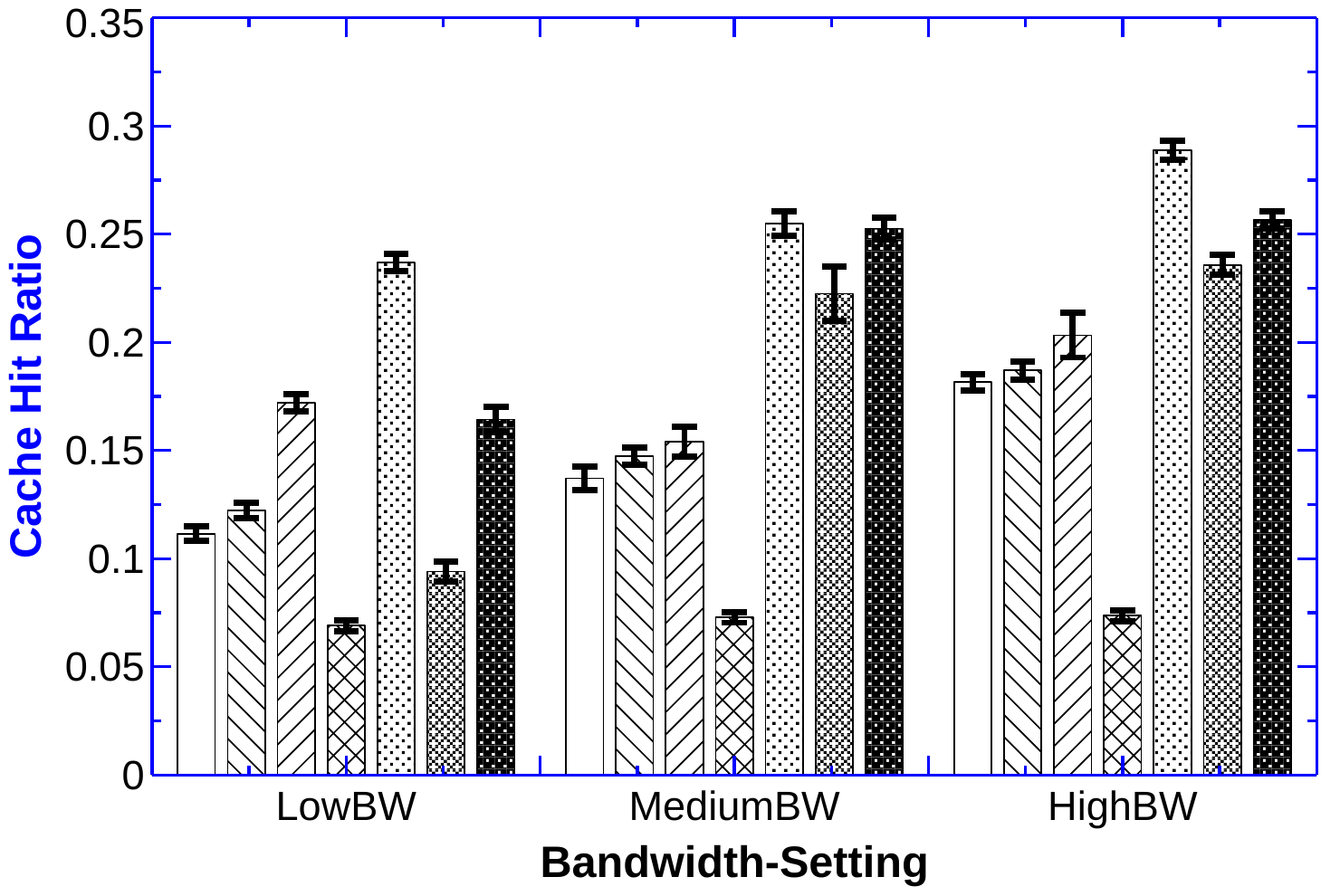}
}
\vspace{-0.1cm}
\caption{Average \textcolor{blue}{\textbf{cache hit ratio}} and 95\% CI in the network with \textbf{0 link failures} per simulation run (higher is better) [zipf].}
\label{fig:zipf_forwarding_results_0LinkFailures_ndnsim2.0_cache}

\subfloat[Results for LowCon\label{zipf_subfig-1:0lf_hops}]{%
\includegraphics[width=0.29\textwidth]{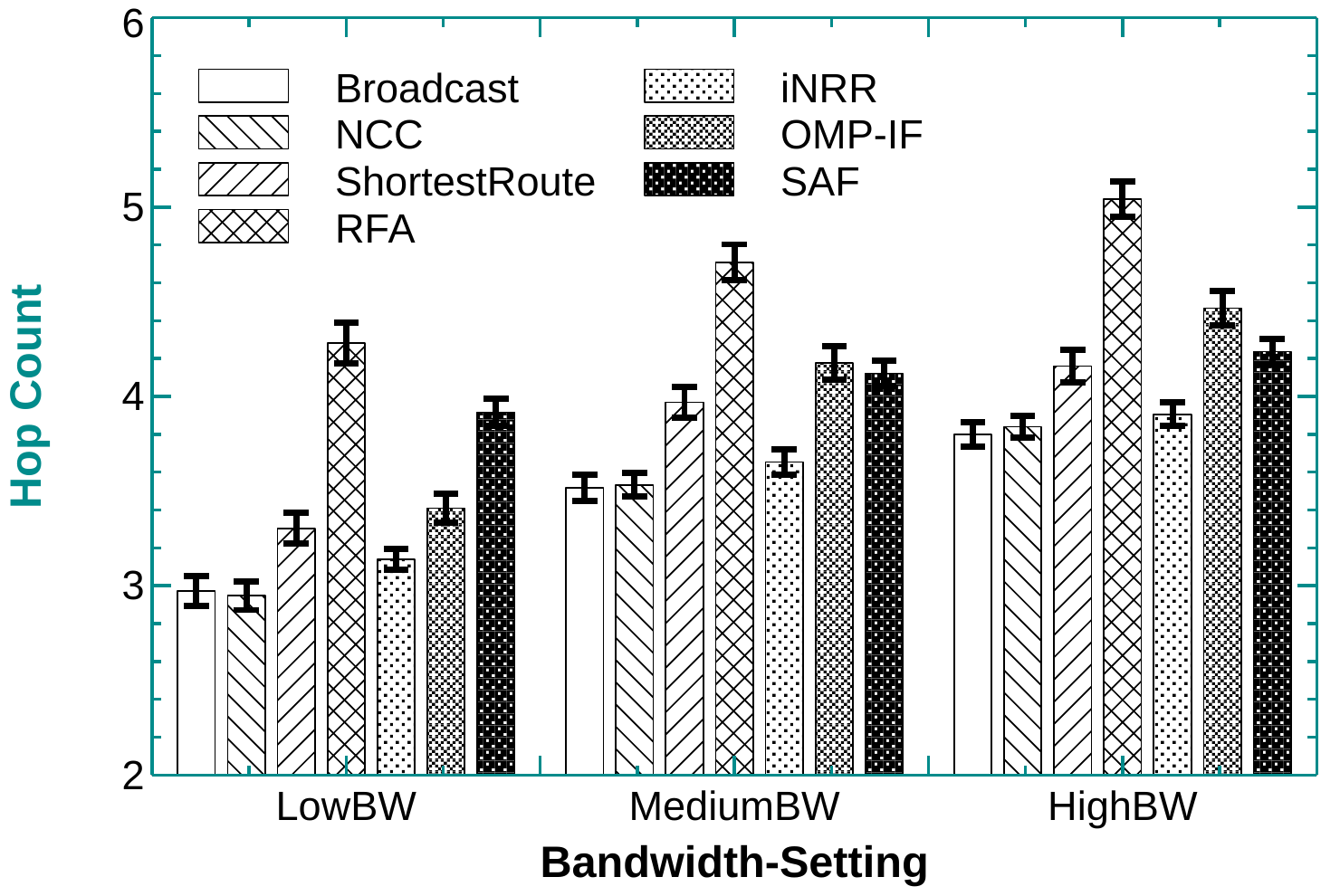}
}
\hfill
\subfloat[Results for MediumCon\label{zipf_subfig-2:0lf_hops}]{%
\includegraphics[width=0.29\textwidth]{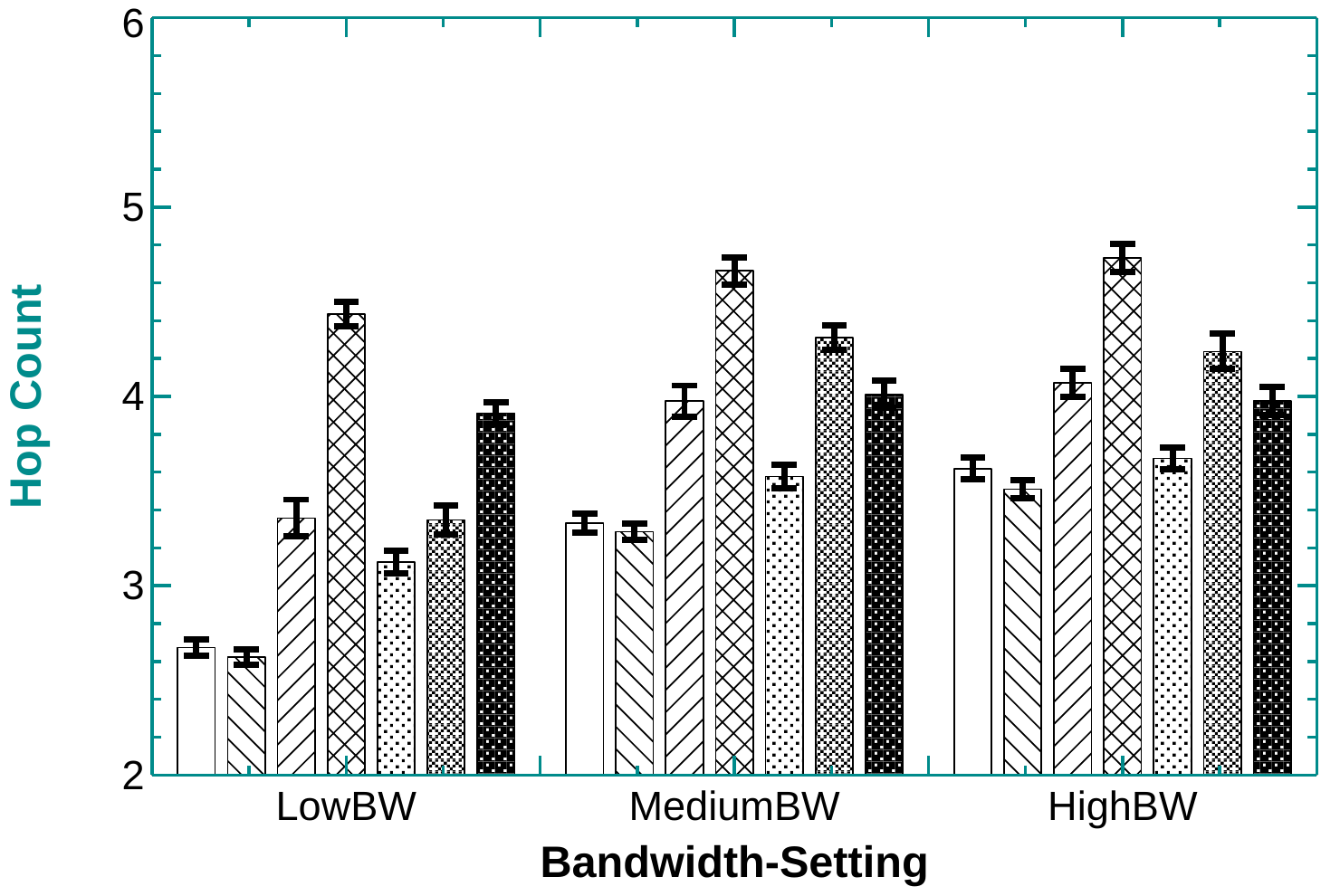}
}
\hfill
\subfloat[Results for HighCon\label{zipf_subfig-3:0lf_hops}]{%
\includegraphics[width=0.29\textwidth]{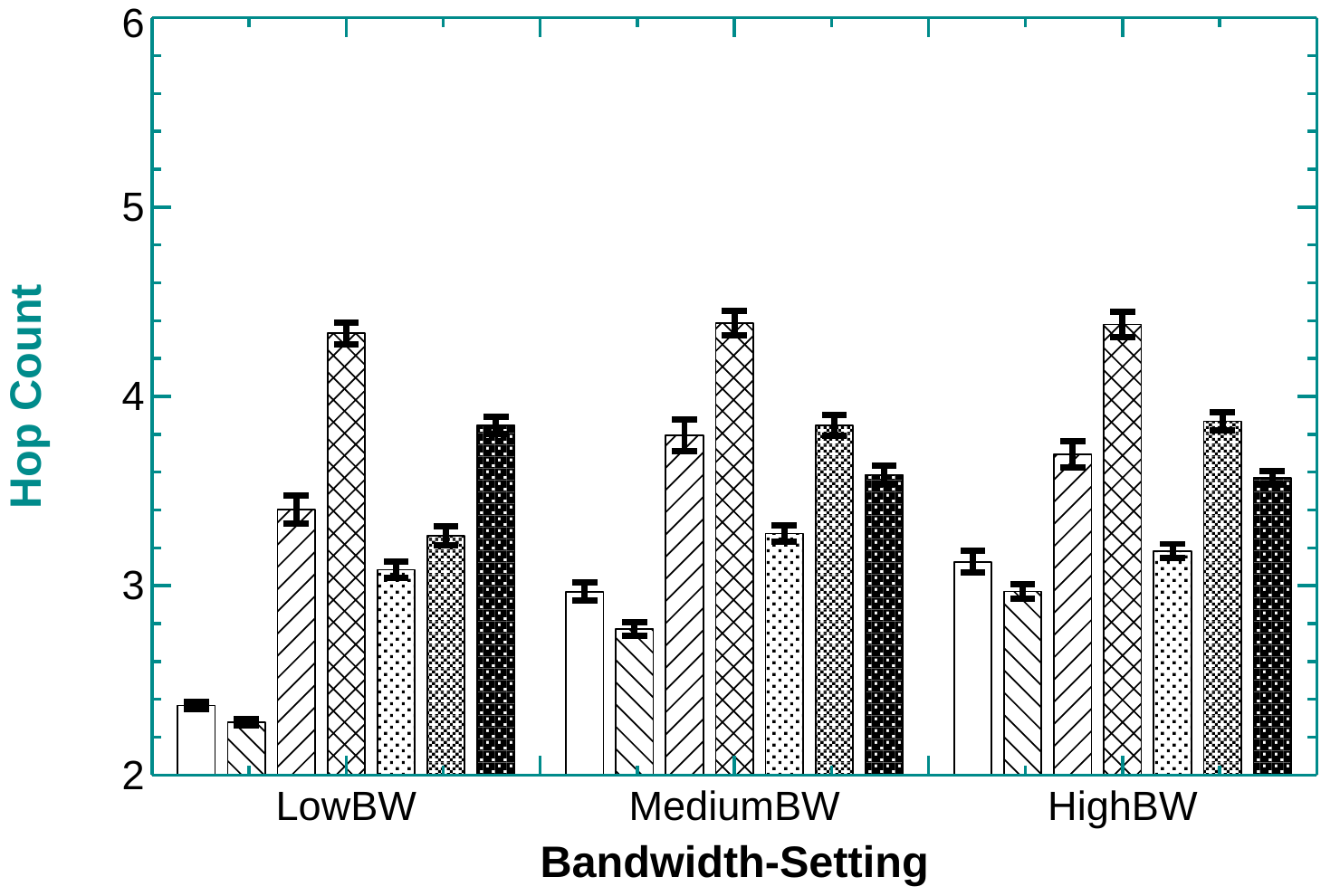}
}
\vspace{-0.1cm}
\caption{Average \textcolor{darkcyan}{\textbf{hop count}} per \emph{satisfied} Interest and 95\% CI with \textbf{0 link failures} per simulation run (lower is better) [zipf].}
\label{fig:zipf_hop_count_0LinkFailures_ndnsim2.0}
\vspace{-0.2cm}
\end{figure*}

\begin{figure*}[tbh!]
\begin{tabular*}{\textwidth}{c @{\extracolsep{\fill}} cccc}
\hspace{-0.2cm}
\subfloat[Broadcast\label{heatmap-broadcast}]{%
\includegraphics[scale=0.22]{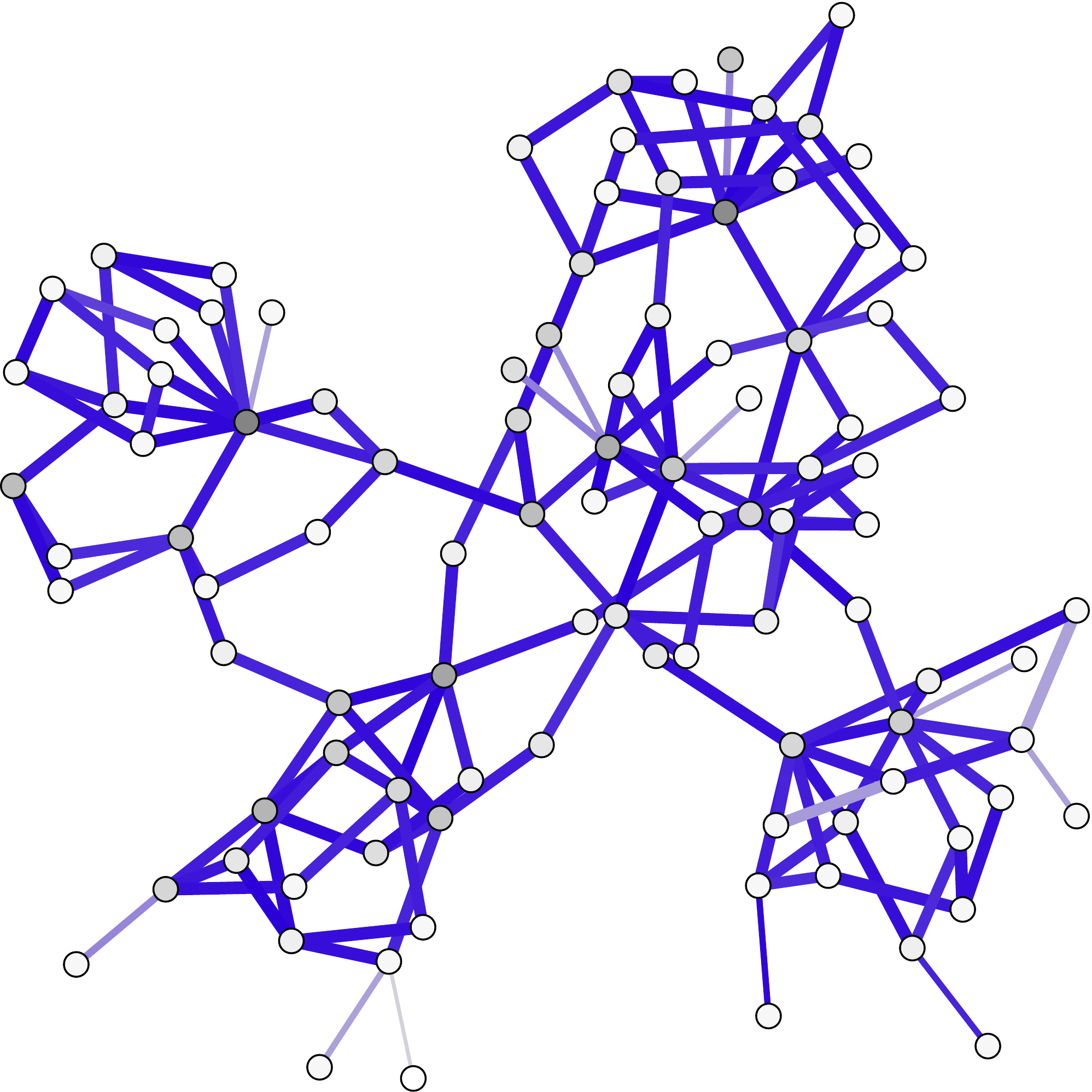}
}
&
\hspace{-0.35cm}
\subfloat[NCC\label{heatmap-ncc}]{%
\includegraphics[scale=0.21]{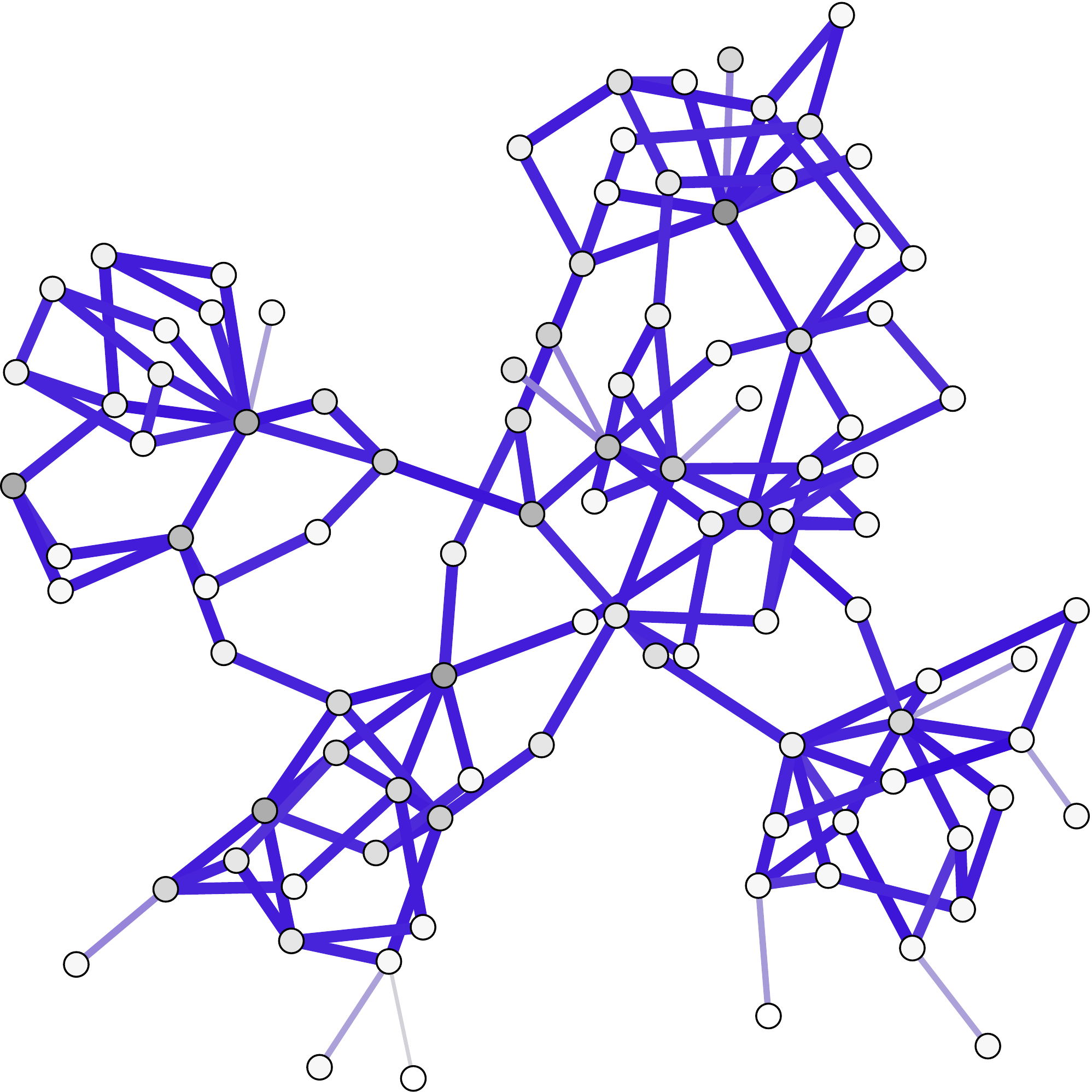}
}
&
\hspace{-0.76cm}
\subfloat[ShortestRoute\label{heatmap-shortestroute}]{%
\includegraphics[scale=0.21]{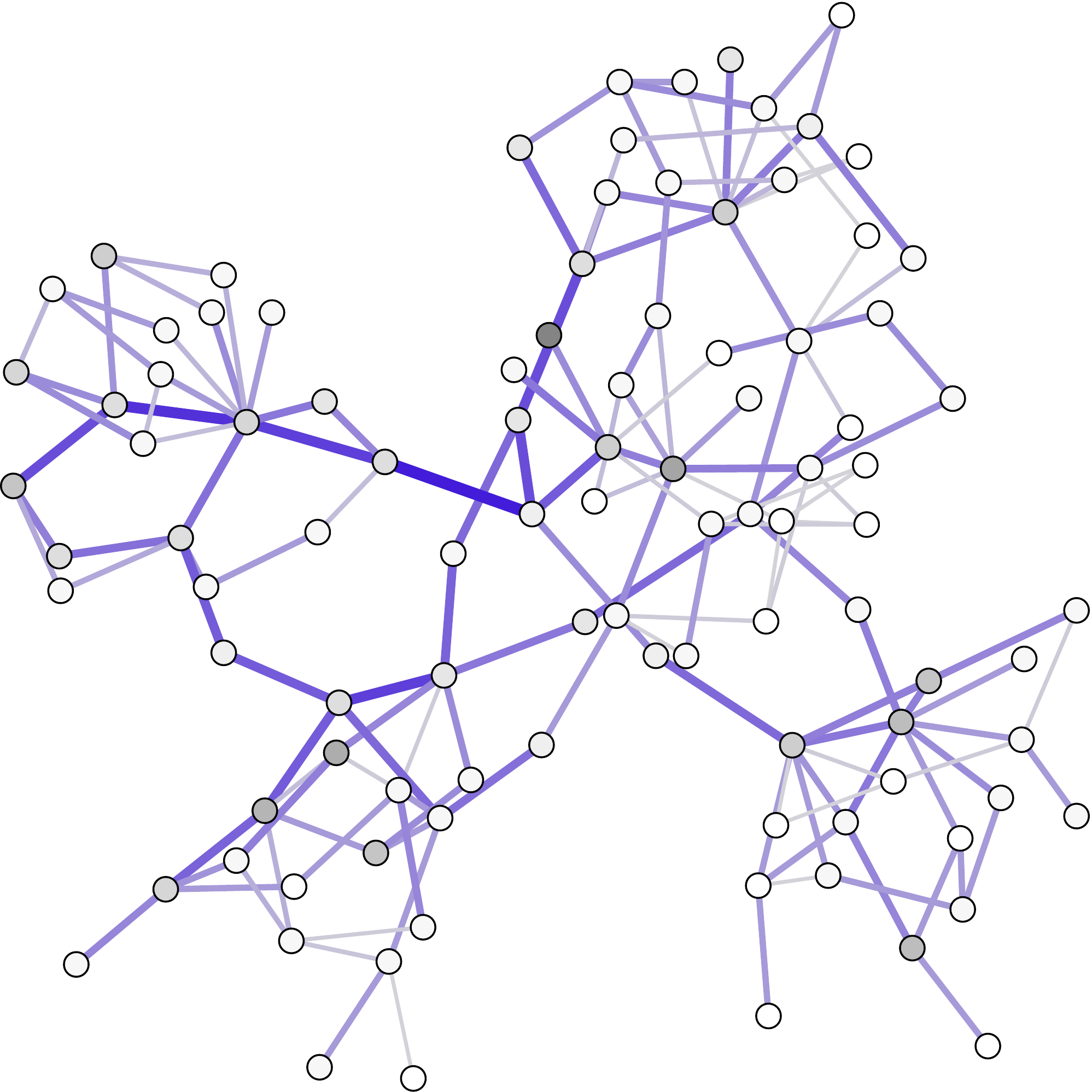}
}
&
\hspace{-0.76cm}
\subfloat[RFA\label{heatmap-rfa}]{%
\includegraphics[scale=0.21]{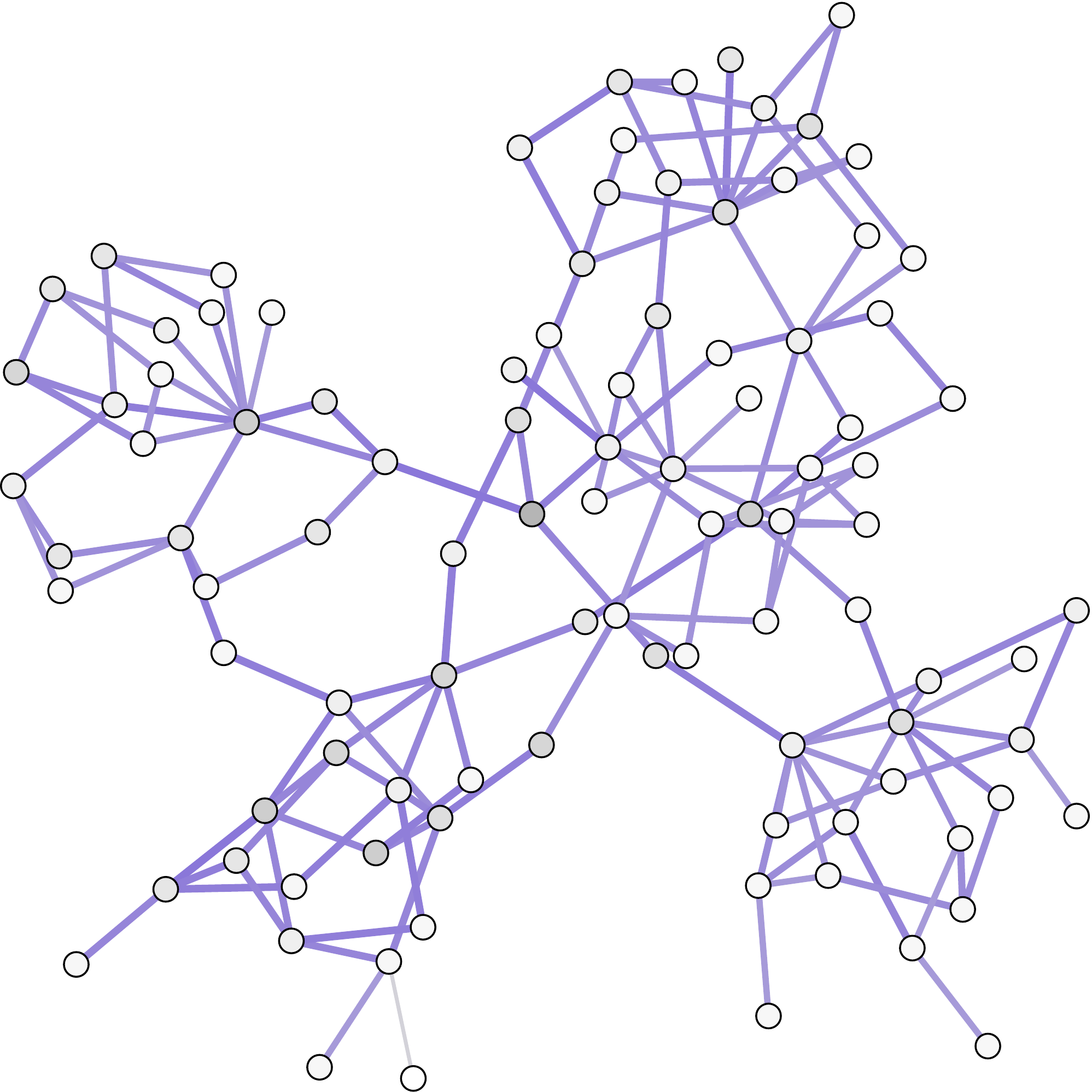}
}
\\

\setcounter{subfigure}{3}

\hspace{-0.2cm}
\subfloat[iNRR\label{heatmap-inrr}]{%
\includegraphics[scale=0.21]{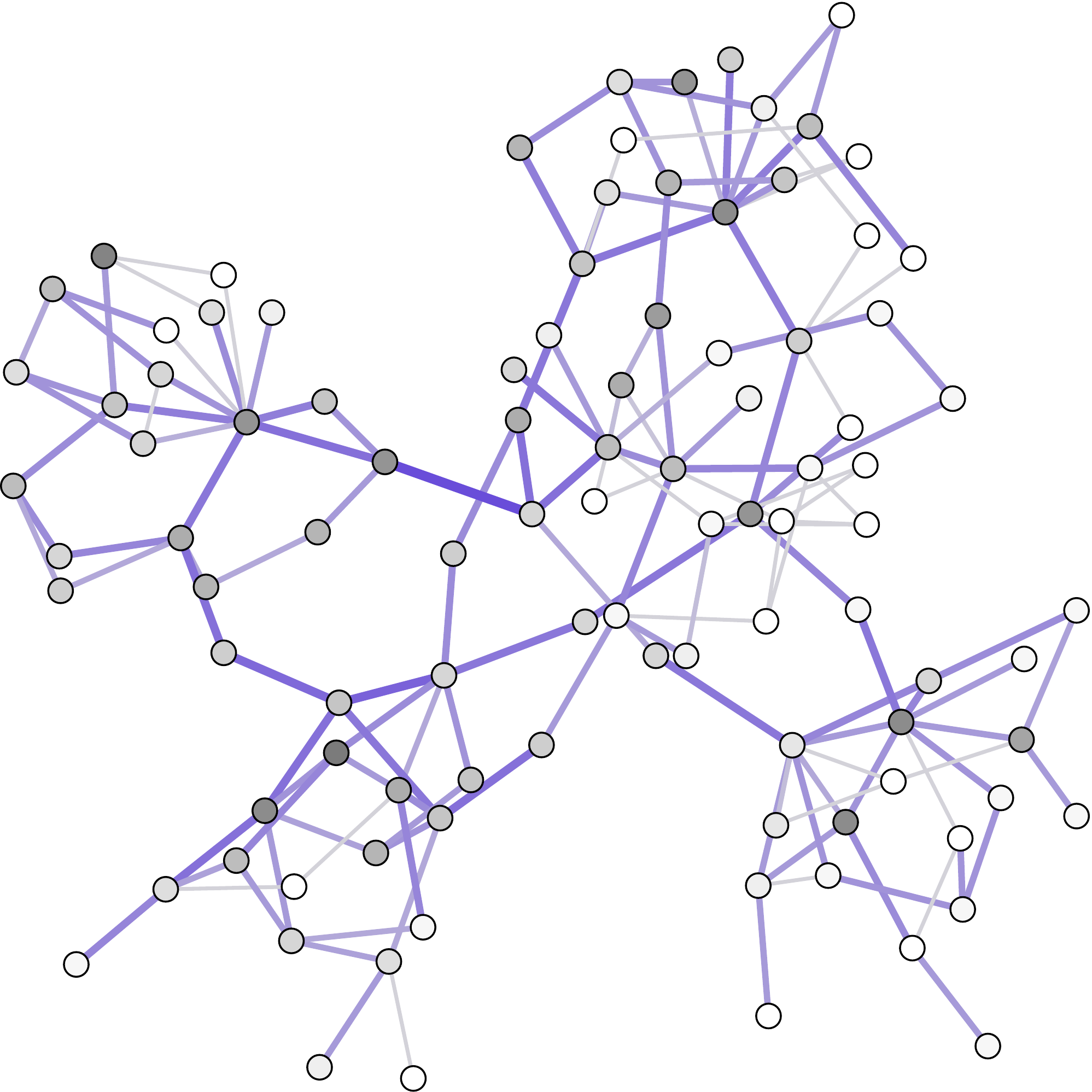}
}
&
\hspace{-0.35cm}
\subfloat[\revised{OMP-IF}\label{heatmap-ompif}]{%
\includegraphics[scale=0.21]{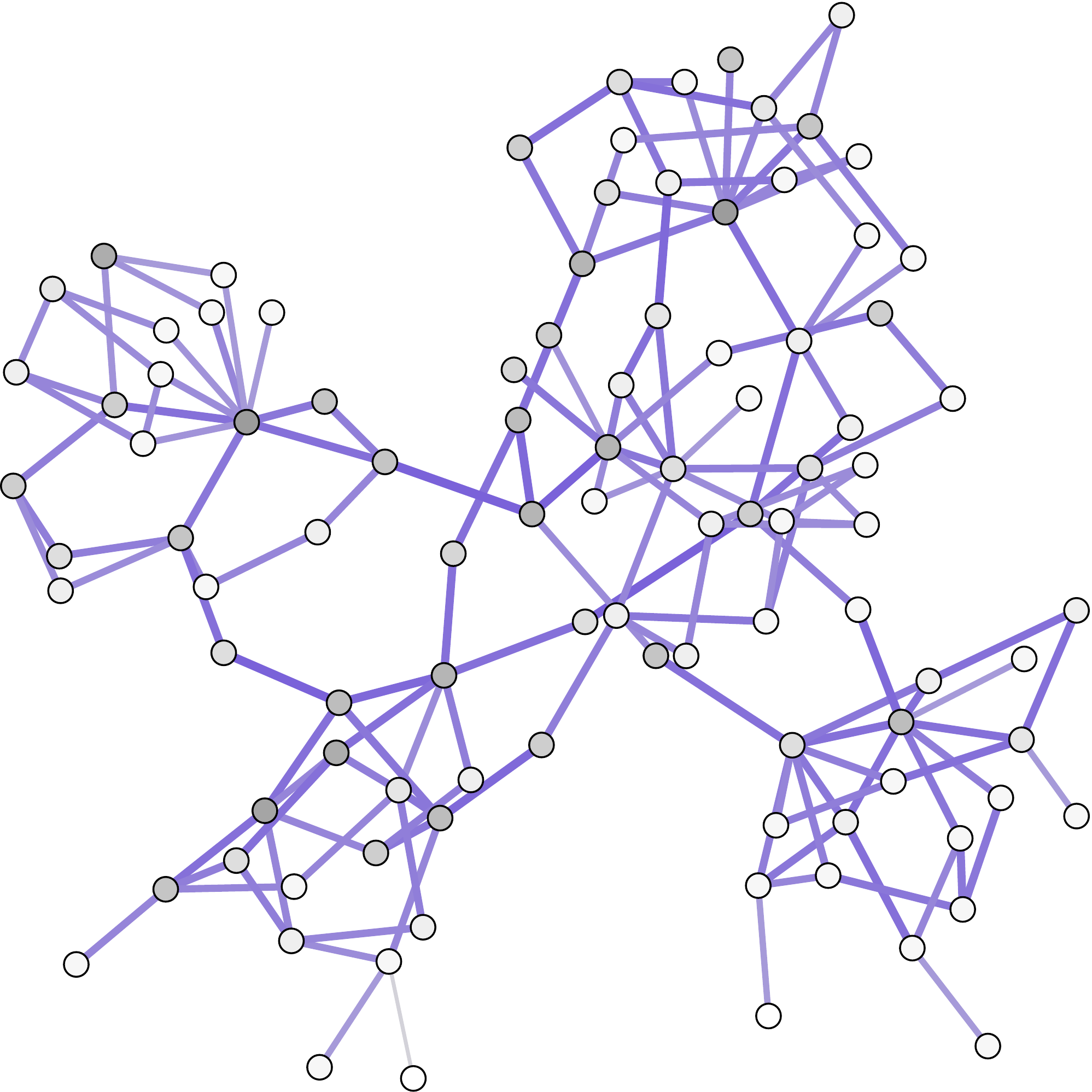}
}
&
\hspace{-0.76cm}
\subfloat[SAF\label{heatmap-saf}]{%
\includegraphics[scale=0.21]{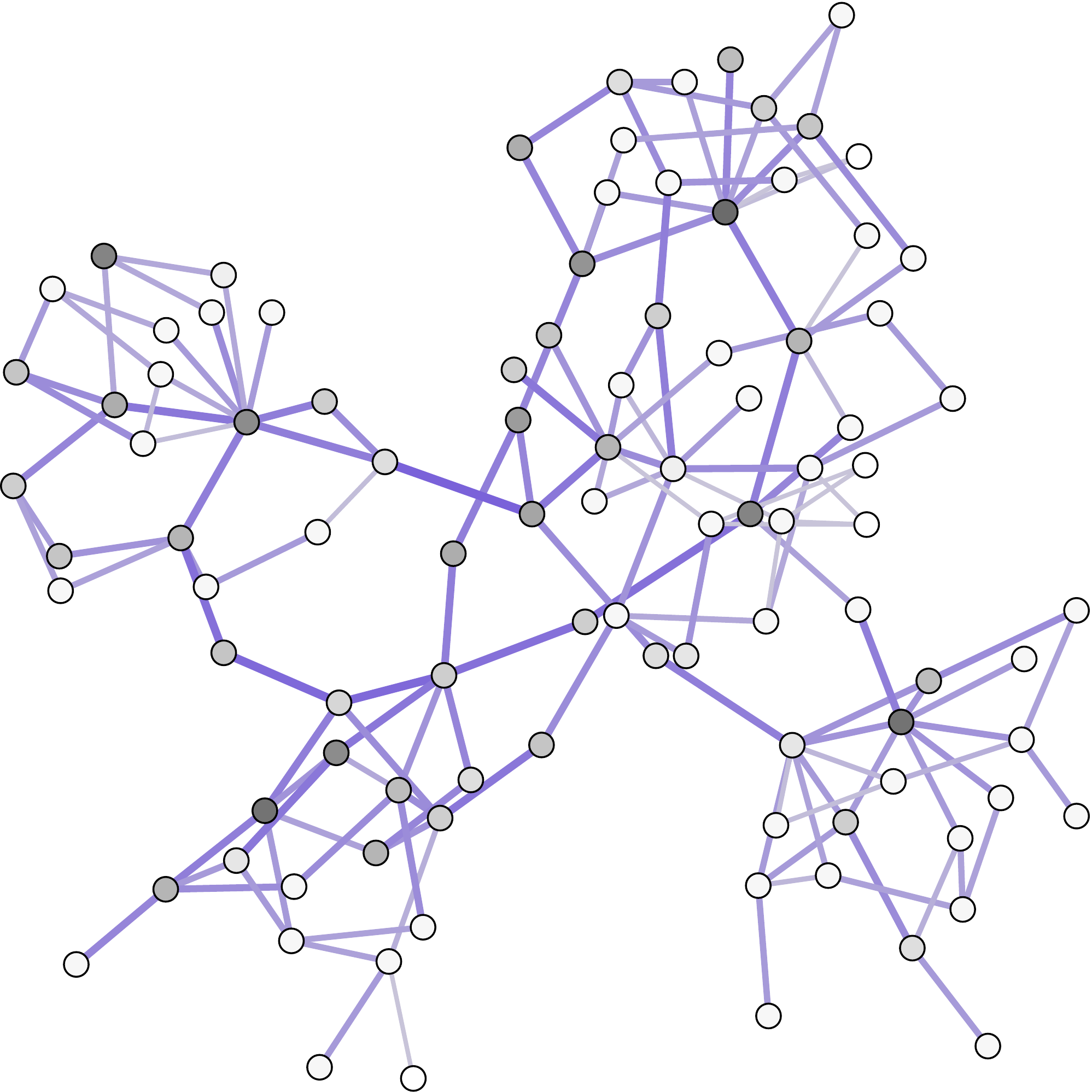}
}
&
\hspace{-0.5cm}\subfloat{%
\includegraphics[scale=0.46]{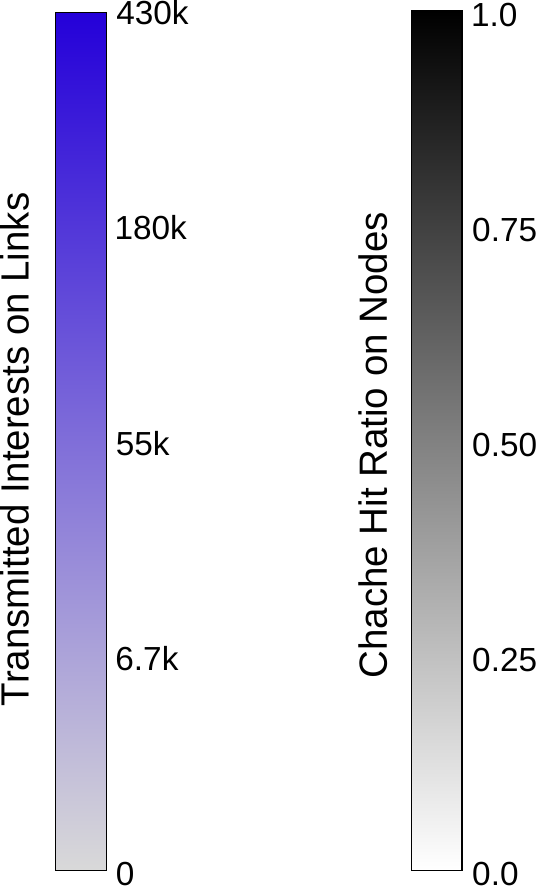}
}
\end{tabular*}
\caption{Heat maps illustrating the distribution of the transmitted Interests and cache hits. The wider and the more saturated a link, the more Interests are transmitted over the given link. The darker the coloring of a node, the higher the cache hit ratio.
}\vspace{-0.35cm}
\label{fig:zipf_heatmap}
\end{figure*}

An increasing number of link failures has a significant negative impact on the Interest satisfaction ratio for all algorithms. For instance, consider Figure~\ref{uniform_subfig-2:00lf} and~\ref{uniform_subfig-2:100lf} illustrating the results for 0 and 100 link failures with \emph{medium} graph connectivity. Table~\ref{tab:linkfailure_comp} depicts the actual and relative performance loss for each strategy. SAF has the highest actual performance loss, however, the relative loss is in a similar range as for its competitors. In these scenarios SAF is not able to fully circumvent link failures due to the resource shortages that are inherent in the selected scenarios. iNRR also suffers from link failures as the oracle used for finding cached replicas does only consider the distance to the content and does not consider link attributes such as reliability or capacity. Interestingly, ShortestRoute maintains the smallest actual and relative performance loss. First, these numbers must be considered with respect to the absolute number of satisfied Interests, and second, this can be explained due to the rather short paths used to satisfy Interests (cf. Figure~\ref{subfig-2:00lf_hops}), which of course have a lower chance of being affected by \revised{randomly emerging link failures. As will be shown later, Broadcast and NCC can not take advantage of their low hop count due to extensive Interest replication (cf. Figure~\ref{fig:zipf_heatmap}).}

Especially if the network resources and the connectivity \revised{are very limited}, RFA is the closest competitor to SAF considering the Interest satisfaction ratio. This is due to the extensive use of all provided routes to the content origins by RFA. \revised{Thus, it is able to distribute the load equally on congested nodes/paths}. Nevertheless, in better connected scenarios and with higher network resources RFA is not anymore able to reach the performance of the \emph{standard algorithms} (Broadcast, NCC, ShortestRoute). RFA uses the number of pending Interests on a face (the one with the lowest number pending is selected) for deciding to which one of the faces in the FIB an Interest shall be forwarded. It does not store any \emph{long term} state or classification of a face. Therefore, the more routes are available to RFA and the higher the network resources, the less RFA is able to determine which faces are the most appropriate for specific content prefixes. For instance, this can be observed in Figure~\ref{fig:forwarding_results_0LinkFailures_ndnsim2.0}c. \revised{This is also reflected by the low cache hit ratio and high hop count shown in Figures~\ref{fig:forwarding_results_0LinkFailures_ndnsim2.0_cache} and~\ref{fig:hop_count_0LinkFailures_ndnsim2.0} indicating that RFA is not able to select those faces that provide good service (e.g., cache hits), but rather takes unnecessary detours}.

\revised{In scenarios with plenty resources available, iNRR and OMP-IF are able to catch up with SAF.}
Figure~\ref{fig:forwarding_results_0LinkFailures_ndnsim2.0_cache} shows that iNRR achieves the highest cache hit ratio, which also is an explanation for iNRR's good performance. 
The higher hop count \revised{compared to Broadcast and NCC} (cf. Figure~\ref{fig:hop_count_0LinkFailures_ndnsim2.0}), can be explained by the deletion of cache entries and the resulting detours iNRR imposes on Interests. In order to understand iNRR's behavior, suppose three nodes $a,b,c \in \mathcal{V}$, whereof $c$ is the content origin. Assume that $a$ receives an Interest and searches for caches that can satisfy it. The nearest cache is $b$ satisfying the aforementioned constraints (cf. Section \ref{sec:related_work}). Therefore, $a$ forwards the Interest to $b$, but at arrival or on the way to $b$, $b$ has already evicted the desired content. Thus, the Interest is forwarded to $c$. Therefore, the triangle inequality ($d(a,c) \leq d(a,b) + d(b,c))$ explains the higher hop count of iNRR compared to Broadcast and NCC. This leads to the assumption that larger caches will improve iNRR's performance. 

\revised{The strong performance of OMP-IF in scenarios with many resources can be explained by its principle of node disjointness. Subsequent Interests requesting the same content will be forwarded with high probability on the same paths resulting in a high cache hit ratio (cf. Figure\ref{subfig-3:00lf_cache}). Furthermore, the usage of probing Interests (which are broadcasted on all faces in the FIB to identify the best performing path) performs well in these scenarios. However, as resources become scarce (cf. Figure\ref{uniform_subfig-2:00lf}) these probing Interests amplify congestion. This leads to the fact that the used paths are switched frequently impairing the cache hit ratio (cf. Figure~\ref{subfig-2:00lf_cache} and~\ref{subfig-3:00lf_cache}) and the overall performance significantly.}


\subsection{Performance under Zipf-like Content Popularity}
In order to compare the results using a uniform content popularity and the results using a Zipf distribution for popularity, we use the same network topologies and simulation parameters as for the previous evaluation (cf. Section~\ref{sec:uniform_request_eval}). Only the distribution of the content popularity was changed to a Zipf distribution with $\alpha = 0.668$ according to~\cite{Cheng:2008}. 
We again conduct 50~simulation runs for each scenario. 
In Section~\ref{sec:uniform_request_eval} we presented the results for scenarios with 0, 50, and 100 link failures. It can be observed that the performance of the algorithms decreases proportionally with an increase in link failures (cf. Figures~\ref{fig:forwarding_results_0LinkFailures_ndnsim2.0},~\ref{fig:forwarding_results_50LinkFailures_ndnsim2.0} and~\ref{fig:forwarding_results_100LinkFailures_ndnsim2.0}). Therefore, we omit the results for 50 and 100 link failures for the following discussion. 
Figures~\ref{fig:zipf_forwarding_results_0LinkFailures_ndnsim2.0},~\ref{fig:zipf_forwarding_results_0LinkFailures_ndnsim2.0_cache}, and~\ref{fig:zipf_hop_count_0LinkFailures_ndnsim2.0} depict the average Interest satisfaction ratio per node, the average cache hit ratio, and the average hop count per satisfied Interest, respectively. The Interest satisfaction ratio for Zipf-distributed content popularity (cf. Figure~\ref{fig:zipf_forwarding_results_0LinkFailures_ndnsim2.0}) paints a picture similar to Figure~\ref{fig:forwarding_results_0LinkFailures_ndnsim2.0} for uniform content popularity. SAF outperforms its competitors in terms of satisfied Interests. Again iNRR and OMP-IF are the strongest competitors to SAF and in well connected scenarios with many resources available these algorithms achieve similar performance.
As expected, iNRR beats SAF in terms of cache hit ratio since iNRR is designed to forward Interests to the nearest cache holding the desired Data packet. \revised{Broadcast and NCC} again provide the lowest average hop count (cf. Figure~\ref{fig:zipf_hop_count_0LinkFailures_ndnsim2.0}) outperforming the other algorithms regarding this metric for the very same reasons as outlined in Section~\ref{sec:uniform_request_eval}.
\revised{Although iNRR always maintains a lower hop count and higher cache hit ratio than SAF, SAF outperforms iNRR regarding the number of satisfied Interests. This is due to the extensive usage of multi-path forwarding considering paths which may not provide optimal performance regarding cache hit ratio or hop count, however, it maximizes the number of satisfied Interests. OMP-IF also focuses on the usage of multiple but node-disjoint paths per content. Nevertheless, as can be seen from Figures~\ref{fig:zipf_forwarding_results_0LinkFailures_ndnsim2.0_cache} and~\ref{fig:zipf_hop_count_0LinkFailures_ndnsim2.0} in most of the cases it maintains a higher hop count and lower cache hit ratio than SAF resulting in a smaller number of overall satisfied Interests.}

Figure~\ref{fig:zipf_heatmap} depicts the distribution of the transmitted Interests and cache hits in the network. For each of the evaluated algorithms, a single simulation run is depicted considering a fixed scenario (MediumCon, HighBW; cf. Figure~\ref{uniform_subfig-2:MediumCon}) with a Zipf-distributed content popularity. This figure illustrates how differently the forwarding algorithms behave. Broadcast replicates Interests on every node and pushes them to all neighboring nodes (cf. Figure~\ref{fig:zipf_heatmap}a). NCC shows a similar behavior. ShortestRoute focuses on the shortest paths from a client to the content provider and, therefore, the caches on the shortest paths are heavily used. RFA performs a kind of load balancing distributing Interests equally on all routes. Thus, it achieves a low cache hit ratio and Interest satisfaction ratio (cf. Figure~\ref{fig:zipf_forwarding_results_0LinkFailures_ndnsim2.0}). iNRR tries to maximize the cache hit ratio by explicitly preferring nearby cached Data packets in the network instead of retrieving them from the content origin. This circumstance is reflected in Figure~\ref{heatmap-inrr} showing several nodes with high cache hit ratios (dark coloring). \revised{OMP-IF also achieves an acceptable cache hit ratio due to its principle of node disjointness. Frequent probing and path switching leads to a similar load balancing behavior as in RFA. SAF is able to provide both, effective load balancing and high cache hit ratios at the relevant nodes (cf. Figure~\ref{heatmap-saf}). Compared to OMP-IF, which balances the load very equally on all available links, SAF focuses on paths that provide better performance, leading to a better cache hit ratio and an overall better performance.}

\vspace{-0.15cm}
\section{Conclusion and Future Work}
\label{sec:conclusion}

This paper introduced Stochastic Adaptive Forwarding (SAF), a novel forwarding strategy for Named Data Networking. SAF provides probability-based forwarding on a per-content/per-prefix basis. It ensures effective forwarding with incomplete and/or invalid routing information, and resolves unexpected network topology changes without relying on the routing plane. The extensive usage of multi-path transmission is the foundation for SAF's performance. SAF is flexible in that it can be configured with various measures defining the forwarding objectives. The effectiveness of SAF was illustrated by conducting extensive simulations using the ns-3/ndnSIM framework. We presented the throughput-based measure $\mathcal{M}^{T}$ optimizing the Interest satisfaction ratio in a given network. Simulations show that SAF is able to outperform existing algorithms. Our results can be reproduced as we provide an open source implementation of SAF and its competitors (iNRR, \revised{OMP-IF}, RFA) at \emph{\url{github.com/danposch/SAF}}.

SAF allows to optimize its forwarding behavior with respect to a measure. The only assumption that SAF has on the measure is that it allows classifying Interests into satisfied and unsatisfied ones. Thus, SAF can be easily tailored to specific use cases. For example, if Interests shall not exceed a specific hop count, one may introduce a measure that classifies Interests as satisfied if the hop count is below this specific threshold. SAF further allows to make decisions on a context and content level. Since SAF keeps track of the forwarding probabilities for all content prefixes, it is possible to adapt the forwarding probabilities with respect to content prefixes. This can be used to prioritize traffic from specific content prefixes. For instance, (real-time) multimedia traffic that takes the same path through the network as bulk file transfer can be assigned higher priority. In this case, one may introduce a weighting such that the probability of the virtual dropping face ($F_{D}$, cf. Section~\ref{subsec:network_content_node_model}) for bulk file transfer is increased. This leads to a decrease of the forwarding probability of $F_{D}$ for multimedia traffic and, therefore, to an increase of the forwarding probability on the desired faces. We plan to design and investigate algorithms for content- and context-aware forwarding in our future work.

{\small
\vspace{-0.2cm}
\section*{Acknowledgment}
This work was supported in part by the Austrian Science Fund (FWF) under the CHIST-ERA project CONCERT,
project nr. \textit{I1402}.
\vspace{-0.3cm}
}

\ifCLASSOPTIONcaptionsoff
  \newpage
\fi


\bibliographystyle{bib/IEEEtran}
\bibliography{bib/quellen}

\begin{thebibliography}{10}
\providecommand{\url}[1]{#1}
\csname url@samestyle\endcsname
\providecommand{\newblock}{\relax}
\providecommand{\bibinfo}[2]{#2}
\providecommand{\BIBentrySTDinterwordspacing}{\spaceskip=0pt\relax}
\providecommand{\BIBentryALTinterwordstretchfactor}{4}
\providecommand{\BIBentryALTinterwordspacing}{\spaceskip=\fontdimen2\font plus
\BIBentryALTinterwordstretchfactor\fontdimen3\font minus
  \fontdimen4\font\relax}
\providecommand{\BIBforeignlanguage}[2]{{%
\expandafter\ifx\csname l@#1\endcsname\relax
\typeout{** WARNING: IEEEtran.bst: No hyphenation pattern has been}%
\typeout{** loaded for the language `#1'. Using the pattern for}%
\typeout{** the default language instead.}%
\else
\language=\csname l@#1\endcsname
\fi
#2}}
\providecommand{\BIBdecl}{\relax}
\BIBdecl

\bibitem{Labovitz:2000}
\BIBentryALTinterwordspacing
C.~Labovitz, A.~Ahuja, A.~Bose, and F.~Jahanian, ``{Delayed Internet Routing
  Convergence},'' \emph{SIGCOMM Computer Communication Review}, vol.~30, no.~4,
  pp. 175--187, Aug. 2000.
\BIBentrySTDinterwordspacing

\bibitem{Labovitz:2001}
\BIBentryALTinterwordspacing
C.~Labovitz, A.~Ahuja, R.~Wattenhofer, and S.~Venkatachary, ``{The Impact of
  Internet Policy and Topology on Delayed Routing Convergence},'' in
  \emph{Proceedings of 20th Joint Conference of the IEEE Computer and
  Communications Societies}, 2001, pp. 537--546.
\BIBentrySTDinterwordspacing

\bibitem{Kini:2009}
S.~Kini, S.~Ramasubramanian, A.~Kvalbein, and A.~Hansen, ``{Fast Recovery from
  Dual Link Failures in IP Networks},'' in \emph{Proc. of the 28th IEEE
  INFOCOM}, April 2009, pp. 1368--1376.

\bibitem{icnSurvey}
\BIBentryALTinterwordspacing
B.~Ahlgren, C.~Dannewitz, C.~Imbrenda, D.~Kutscher, and B.~Ohlman, ``{A Survey
  of Information-Centric Networking},'' \emph{IEEE Communications Magazine},
  vol.~50, no.~7, pp. 26--36, 2012.
\BIBentrySTDinterwordspacing

\bibitem{Xylomenos:2014}
\BIBentryALTinterwordspacing
G.~Xylomenos, C.~Ververidis, V.~Siris, N.~Fotiou, C.~Tsilopoulos, X.~Vasilakos,
  K.~Katsaros, and G.~Polyzos, ``{A Survey of Information-Centric Networking
  Research},'' \emph{IEEE Communications Surveys and Tutorials}, vol.~16,
  no.~2, pp. 1024--1049, 2014.
\BIBentrySTDinterwordspacing

\bibitem{netinf}
\BIBentryALTinterwordspacing
C.~Dannewitz, D.~Kutscher, B.~Ohlman, S.~Farrell, B.~Ahlgren, and H.~Karl,
  ``{Network of Information (NetInf) - An Information-Centric Networking
  Architecture},'' \emph{Computer Comm.}, vol.~36, no.~7, 2013.
\BIBentrySTDinterwordspacing

\bibitem{psirp}
\BIBentryALTinterwordspacing
N.~Fotiou, D.~Trossen, and G.~Polyzos,
  ``\BIBforeignlanguage{English}{{Illustrating a Publish-Subscribe Internet
  Architecture}},'' \emph{\BIBforeignlanguage{English}{Telecomm. Systems}},
  vol.~51, no.~4, 2012.
\BIBentrySTDinterwordspacing

\bibitem{Jacobson:2009}
\BIBentryALTinterwordspacing
V.~Jacobson, D.~Smetters, J.~Thornton, M.~Plass, N.~Briggs, and R.~Braynard,
  ``{Networking Named Content},'' in \emph{{Proc. of the 5th ACM Int. Conf. on
  Emerging Networking Experiments \& Technologies}}, 2009.
\BIBentrySTDinterwordspacing

\bibitem{dona}
\BIBentryALTinterwordspacing
T.~Koponen, M.~Chawla, B.-G. Chun, A.~Ermolinskiy, K.~H. Kim, S.~Shenker, and
  I.~Stoica, ``{A Data-oriented (and Beyond) Network Architecture},''
  \emph{SIGCOMM Comp. Comm. Rev.}, vol.~37, no.~4, 2007.
\BIBentrySTDinterwordspacing

\bibitem{ndn:2014}
\BIBentryALTinterwordspacing
L.~Zhang, A.~Afanasyev, J.~Burke, V.~Jacobson, K.~Claffy, P.~Crowley,
  C.~Papadopoulos, L.~Wang, and B.~Zhang, ``{Named Data Networking},''
  \emph{ACM SIGCOMM Comp. Comm. Rev.}, vol.~44, no.~3, pp. 66--73, 2014.
\BIBentrySTDinterwordspacing

\bibitem{Yi:2013}
\BIBentryALTinterwordspacing
C.~Yi, A.~Afanasyev, I.~Moiseenko, L.~Wang, B.~Zhang, and L.~Zhang, ``{A Case
  for Stateful Forwarding Plane},'' \emph{Computer Communications}, vol.~36,
  no.~7, pp. 779 -- 791, 2013.
\BIBentrySTDinterwordspacing

\bibitem{yi:routing}
\BIBentryALTinterwordspacing
C.~Yi, J.~Abraham, A.~Afanasyev, L.~Wang, B.~Zhang, and L.~Zhang, ``{On the
  Role of Routing in Named Data Networking},'' in \emph{Proc. of the 1st Int.
  Conference on Information-Centric Networking}.\hskip 1em plus 0.5em minus
  0.4em\relax ACM, 2014.
\BIBentrySTDinterwordspacing

\bibitem{Rossini:2014}
\BIBentryALTinterwordspacing
G.~Rossini and D.~Rossi, ``{Coupling Caching and Forwarding: Benefits,
  Analysis, and Implementation},'' in \emph{Proc. of the 1st Int. Conference on
  Information-Centric Networking}.\hskip 1em plus 0.5em minus 0.4em\relax ACM,
  2014, pp. 127--136.
\BIBentrySTDinterwordspacing

\bibitem{Yeh:2014}
\BIBentryALTinterwordspacing
E.~Yeh, T.~Ho, Y.~Cui, M.~Burd, R.~Liu, and D.~Leong, ``{VIP: A Framework for
  Joint Dynamic Forwarding and Caching in Named Data Networking},'' in
  \emph{Proc. of the 1st Int. Conference on ICN}.\hskip 1em plus 0.5em minus
  0.4em\relax ACM, 2014.
\BIBentrySTDinterwordspacing

\bibitem{rfc3272}
D.~Awduche, A.~Chiu, A.~Elwalid, I.~Widjaja, and X.~Xiao, ``{Overview and
  Principles of Internet Traffic Engineering},'' RFC 3272, IETF, 2002.

\bibitem{adaptiveForwarding}
\BIBentryALTinterwordspacing
C.~Yi, A.~Afanasyev, L.~Wang, B.~Zhang, and L.~Zhang, ``{Adaptive Forwarding in
  Named Data Networking},'' \emph{SIGCOMM Computer Communication Review},
  vol.~42, no.~3, pp. 62--67, Jun. 2012.
\BIBentrySTDinterwordspacing

\bibitem{ndnSIM}
A.~Afanasyev, I.~Moiseenko, and L.~Zhang, ``{{ndnSIM}: {NDN} Simulator for
  {NS-3}},'' Tech. Report. NDN-0005, 2012, {University of California, LA}.

\bibitem{ndnSIM2.0}
S.~Mastorakis, A.~Afanasyev, I.~Moiseenko, and L.~Zhang, ``{{ndnSIM 2.0}: A new
  Version of the {NDN} Simulator for {NS-3}},'' Tech. Report NDN-0028, 2015,
  \url{http://named-data.net/publications/techreports/}.

\bibitem{nfd}
A.~Afanasyev, J.~Shi, B.~Zhang, L.~Zhang \emph{et~al.}, ``{NFD Developer's
  Guide},'' Tech. Report NDN-0021, 2014, {University of California, LA}.

\bibitem{Rosensweig:2010}
\BIBentryALTinterwordspacing
E.~Rosensweig, J.~Kurose, and D.~Towsley, ``{Approximate Models for General
  Cache Networks},'' in \emph{Proc. of the 29th IEEE INFOCOM}, 2010.
\BIBentrySTDinterwordspacing

\bibitem{Chiocchetti:2013}
\BIBentryALTinterwordspacing
R.~Chiocchetti, D.~Rossi, and G.~Rossini, ``{ccnSim: An Highly Scalable CCN
  Simulator},'' in \emph{Proc. of IEEE ICC}, 2013, pp. 2309--2314.
\BIBentrySTDinterwordspacing

\bibitem{INFORM}
\BIBentryALTinterwordspacing
R.~Chiocchetti, D.~Perino, G.~Carofiglio, D.~Rossi, and G.~Rossini, ``{INFORM:
  A Dynamic Interest Forwarding Mechanism for ICN},'' in \emph{Proc. of the 3rd
  ACM SIGCOMM Workshop on ICN}, 2013, pp. 9--14.
\BIBentrySTDinterwordspacing

\bibitem{Carofiglio:2013}
\BIBentryALTinterwordspacing
G.~Carofiglio, M.~Gallo, L.~Muscariello, M.~Papalini, and S.~Wang, ``{Optimal
  Multipath Congestion Control and Request Forwarding in Information-Centric
  Networks},'' in \emph{Proceedings of 21st IEEE International Conference on
  Network Protocols (ICNP)}, Oct 2013.
\BIBentrySTDinterwordspacing

\bibitem{probForwarding}
H.~Qian, R.~Ravindran, G.-Q. Wang, and D.~Medhi, ``{Probability-based Adaptive
  Forwarding Strategy in Named Data Networking},'' in \emph{IFIP/IEEE Int.
  Symposium on Integrated Network Management}, 2013.

\bibitem{multipathInterestFW}
\BIBentryALTinterwordspacing
A.~Udugama, X.~Zhang, K.~Kuladinithi, and C.~Goerg, ``{An On-demand Multi-Path
  Interest Forwarding Strategy for Content Retrievals in CCN},'' in
  \emph{Proceedings of IEEE NOMS}, 2014.
\BIBentrySTDinterwordspacing

\bibitem{brite}
A.~Medina, I.~Matta, and J.~Byers, ``{BRITE: A Flexible Generator of Internet
  Topologies},'' 2000, \url{www.cs.bu.edu/brite/} (Boston University).

\bibitem{Faloutsos:1999}
\BIBentryALTinterwordspacing
M.~Faloutsos, P.~Faloutsos, and C.~Faloutsos, ``{On Power-law Relationships of
  the Internet Topology},'' \emph{SIGCOMM Computer Communication Review},
  vol.~29, no.~4, pp. 251--262, Aug. 1999.
\BIBentrySTDinterwordspacing

\bibitem{Wang:2013}
\BIBentryALTinterwordspacing
Y.~Wang, Z.~Li, G.~Tyson, S.~Uhlig, and G.~Xie, ``{Optimal Cache Allocation for
  CCN},'' in \emph{Proceeding of 21st IEEE ICNP}, 2013.
\BIBentrySTDinterwordspacing

\bibitem{Anand:2009}
\BIBentryALTinterwordspacing
A.~Anand, C.~Muthukrishnan, A.~Akella, and R.~Ramjee, ``{Redundancy in Network
  Traffic: Findings and Implications},'' \emph{SIGMETRICS Performance
  Evaluation Rev.}, vol.~37, no.~1, pp. 37--48, 2009.
\BIBentrySTDinterwordspacing

\bibitem{Cheng:2008}
\BIBentryALTinterwordspacing
X.~Cheng, C.~Dale, and J.~Liu, ``{Statistics and Social Network of YouTube
  Videos},'' in \emph{Proc. of the 16th Int. Workshop on QoS}, 2008.
\BIBentrySTDinterwordspacing

\end{thebibliography}


\end{document}